\documentclass{sig-alternate}
\usepackage{algorithm}
\usepackage{amsmath}
\usepackage{algorithmic}
\usepackage{cite}
\usepackage{graphicx}
\usepackage{multirow}
\usepackage{amssymb}
\usepackage{color}
\usepackage{url}
\usepackage{cite}
\usepackage[T1,OT1]{fontenc}
\usepackage{wrapfig}

\usepackage{relsize}
\usepackage{tikz}
\usepackage{subfig}
\usepackage{lmodern}
\usepackage{afterpage}
\usetikzlibrary{arrows}
\usepackage{colortbl}
\tikzstyle{block}=[draw opacity=0.7,line width=1.4cm]

\usetikzlibrary{positioning}
\usetikzlibrary{arrows, decorations.pathmorphing}
\usepackage{balance}

\usepackage[
  pass,
]{geometry}

\def\infinity{\rotatebox{90}{8}}

\tikzset{snake arrow/.style=
{-triangle 45,
line width=1.4pt,
decorate,
decoration={snake,amplitude=1mm,segment length=10mm,post length=2mm}},
}
\newcommand{\ignore}[1]{}

\newtheorem{problem}{Problem}
\newtheorem{definition}{Definition}
\newtheorem{theorem}{Theorem}
\newtheorem{lemma}{Lemma}
\newtheorem{corollary}{Corollary}
\newtheorem{conclusion}{Conclusion}

\newtheorem{example}{Example}

\newlength{\oldtextfloatsep}
\newlength{\oldfloatsep}

\newcommand{\moveup}{\vspace*{-2mm}}
\newcommand{\moveups}{\vspace*{-1mm}}

\newcommand{\tabcaption}[1]{\vspace*{-3mm}\caption{#1}\vspace*{-5mm}}
\newcommand{\figcaption}[1]{\vspace*{-3mm}\caption{#1}\vspace*{-5mm}}

\newcommand{\thatsymbol}{\fontencoding{T1}\selectfont \TH}
\newcommand{\expectation}{\mathbb{E}}
\newcommand{\flow}{\Gamma}
\newcommand{\diameter}{\mathcal{D}}
\newcommand{\pathset}{\mathbb{P}}
\newcommand{\Out}{\textbf{Out}}
	\newcommand{\In}{\textbf{In}}
\newcommand{{\kempegreedy}}{G{\scriptsize{REEDY}}}
\newcommand{{\ourgreedy}}{ScoreG{\scriptsize{REEDY}}}

\ignore{
\newcommand\subparagraph{%
  \@startsection{subparagraph}{5}
  {\parindent}
  {3.25ex \@plus 1ex \@minus .2ex}
  }
\makeatother
\usepackage[compact]{titlesec}
\titlespacing{\section}
  {0pt}{*1}{*1}
\titlespacing{\subsection}
  {0pt}{*1}{*1}
\titlespacing{\subsubsection}
  {0pt}{*1}{*1}
}
\usepackage{times}

\newcommand*\samethanks[1][\value{footnote}]{\footnotemark[#1]}
\DeclareMathOperator*{\argmax}{arg\,max}

\begin{document}
\sloppy

\clubpenalty=10000
\widowpenalty = 10000

\conferenceinfo{SIGMOD'16}{San Francisco, USA}
\title{Holistic Influence Maximization: Combining Scalability and Efficiency with Opinion-Aware Models}

\ignore{
\numberofauthors{3}

\author{
\alignauthor
Sainyam Galhotra\thanks{The first two authors have contributed equally to this work.}\\
\alignauthor
Akhil Arora\samethanks\\
\alignauthor
Shourya Roy\\
%
\sharedaffiliation
	\affaddr{\{sainyam.galhotra, akhil.arora, shourya.roy\}@xerox.com}\\
	\affaddr{Text \& Graph Analytics (TGA), Xerox Research Centre India (XRCI), Bangalore, India}
}

\def\sharedaffiliation{
\end{tabular}
\begin{tabular}{c}}
}

\numberofauthors{1}

\author{
	\alignauthor
		Sainyam Galhotra\thanks{The first two authors have contributed equally to this work.}  ~  Akhil Arora\samethanks  ~  Shourya Roy\\
		\affaddr{Text and Graph Analytics (TGA), Xerox Research Centre India (XRCI), Bangalore, India}\\
		\email{sainyam@cs.umass.edu ~\{akhil.arora, shourya.roy\}@xerox.com}
}

\ignore{
\numberofauthors{1}
\author{\alignauthor Sainyam Galhotra, Akhil Arora, Shourya Roy \\
\affaddr{Xerox Research Centre India, Bangalore} \\ 
\email{\{sainyam.galhotra, akhil.arora, shourya.roy\}@xerox.com}
}
}

\date{\today}

\maketitle
\begin{abstract}
The steady growth of graph data from social networks has resulted in wide-spread research in finding solutions to the influence maximization problem. In this paper, we propose a \emph{holistic} solution to the influence maximization (\emph{IM}) problem. (1) We introduce an \emph{opinion-cum-interaction} (OI) model that closely mirrors the real-world scenarios.
Under the OI model, we introduce a novel problem of \emph{Maximizing the Effective Opinion (MEO)} of influenced users. We prove that the \emph{MEO} problem is NP-hard and cannot be approximated within a constant ratio unless P=NP. (2) We propose a heuristic algorithm \emph{OSIM} to efficiently solve the MEO problem. To better explain the \emph{OSIM} heuristic, we first introduce \emph{EaSyIM} -- the opinion-oblivious version of \emph{OSIM}, a scalable algorithm capable of running within practical compute times on commodity hardware. In addition to serving as a fundamental building block for \emph{OSIM}, \emph{EaSyIM} is capable of addressing the scalability aspect -- \emph{memory consumption} and \emph{running time}, of the IM problem as well.

Empirically, our algorithms are capable of maintaining the deviation in the spread always within 5\% of the best known methods in the literature. In addition, our experiments show that both \emph{OSIM} and \emph{EaSyIM} are effective, efficient, scalable and significantly enhance the ability to analyze real datasets.
\ignore{
\\ \textbf{ToDO -- Writing in the abstract}\\
1) What should be called as EaSyIM -- the union algorithm OR the aggregation algorithm?\\
	The aggregation algorithm is called EaSyIM.
2) Add the error bound (if possible) for EaSyIM\\
	a) Union algorithm on general graphs with depth (d)...\\
	b) The error of aggregation algorithm over union algorithm on DAGs..\\
3) Remove point number (3), if it does not make much sense..
}
%
\end{abstract}

\ignore{
\vspace{1.5mm}
\noindent
{\bf Categories and Subject Descriptors:} H.2.8 {[{\bf Database Management}]}: {Database Applications} -- \emph{Data Mining}

\vspace{1mm}
}


\section{Motivation}
\label{sec:intro}

\ignore{
\textbf{@SR: \\ (ignore (AA \& SG))\\
1) Once we mention IC and LT in the introduction, can we briefly introduce them as well. This would make sense here and we simply state in one line (very briefly) that how is our model different than them.\\
2) Please also use the word activation of a node while explaining them, so that we can use this word safely in the subsequent sections. [Model mainly]\\
3) We think we should also include the examples presented by us in our ID. This would show clearly the holistic nature and motivate the model, objective function and cost-effectiveness. If needed we can leave the cost-effective example.[Owing to space constraints].\\
4) [A mild suggestion, Please guide us on this] \\
For making the consistency between sentiment, polarity and opinion. We can say (in the introduction) that the word sentiment/polarity is interchangeably used to represent the feeling about an event. More-formally we introduce a notion of an opinion later in the document and design our problems around it. \\
}
}




Growth and pervasiveness of online social networks is no longer a new phenomenon. They have become an integral part of the day-to-day life of almost every Internet user. Their wide-spread reach paves the way for a host of applications -- (1) \emph{Viral marketing/ad-targeting} \cite{virad1, inf2, virad2}, (2) \emph{Outbreak detection} \cite{disease, celf}, (3) \emph{Community formation, evolution and detection} \cite{communityFormation1,im_community_detect2,im_community_analysis1}, (4) \emph{Recommendations} using social-media \cite{reco2, reco3} and many more. The \emph{influence maximization} (IM) problem \cite{kempe} with its applicability in solving the above mentioned problems and beyond, has thus, been one of the most widely studied problems over the past decade. This problem is to identify a set of \emph{seed nodes} so that the overall \emph{spread} of information in a network, which is the potential collective impact of imparting that piece of information to these nodes, is maximized.

Given that the objective of IM is to maximize the \emph{spread} of information about a content, which can be a product, person, event and many more, an important aspect of this problem is the underlying \emph{information diffusion} model. A diffusion model defines the dynamics of information propagation and also controls the way this information is perceived by the nodes in a network. Nevertheless, a substantially large fraction of the literature in this field has focussed on devising efficient and scalable algorithms \cite{kempe, celf, celfPlus, opt1, tim, skim, mia, pmia, irie, simpath, ldag, parallel_im, im_prunedmc, im_staticgreedy, im_rank} for IM using the classical information diffusion models \cite{kempe, inf2}. The two fundamental diffusion models proposed by Kempe et al. \cite{kempe} -- \emph{Independent Cascade} (IC) and \emph{Linear Threshold} (LT) have been almost exclusively followed in majority while extended in a small fraction of all the subsequent work.

\begin{figure}[t]
\centering
\scalebox{0.4}{

\begin{tikzpicture}[node distance=15mm,
round/.style={fill=green!50!black!20,draw=green!50!black,minimum size=5mm,text width=10mm,align=center,circle,thick}]
\centering
\node at (0,1.1) {\textbf{\huge $\mathbf{o_{B}=0}$}};
\node at (0,-3.1) {\textbf{\huge $\mathbf{o_{D}=-0.3}$}};
\node at (7.1,0) {\textbf{\huge $\mathbf{o_{C}=0.6}$}};
\node at (-7.1,0) {\textbf{\huge $\mathbf{o_{A}=0.8}$}};
\node at (-3.3,.35) {\textbf{\huge $\mathbf{p_{BA}=0.1}$}};
\node at (-3,-.55) {\textbf{\huge $\mathbf{\varphi_{BA}=0.7}$}};
\node at (3.3,.4) {\textbf{\huge $\mathbf{p_{BC}=0.1}$}};
\node at (2.8,-.6) {\textbf{\huge $\mathbf{\varphi_{BC}=0.8}$}};
\node at (-2.8,-1.5){\textbf{\huge $\mathbf{p_{AD}=0.8}$}};
\node at (-3.6,-2.5) {\textbf{\huge $\mathbf{\varphi_{AD}=0.9}$}};
\node at (2.8,-1.6) {\textbf{\huge $\mathbf{p_{CD}=0.9}$}};
\node at (3.3,-2.5) {\textbf{\huge $\mathbf{\varphi_{CD}=0.1}$}};s

\node[round] (a) {\huge $\mathbf{B}$};
\node[round] (c) [below right =.03cm and 6cm of a]{\huge $\mathbf{C}$};
\node[round] (b) [below left =.03cm and 6cm of a]{\huge $\mathbf{A}$};
\node[round] (d) [below left =.03cm and 6cm of c]{\huge $\mathbf{D}$};

\draw[-latex,-triangle 45,line width=1.6pt](a) to[bend right = 10] (b);
\draw[-latex,-triangle 45,line width=1.6pt](a) to[bend left = 10] (c);
\draw[-latex,-triangle 45,line width=1.6pt](b) to[bend right = 10] (d);
\draw[-latex,-triangle 45,line width=1.6pt](c) to[bend left = 10] (d);
\end{tikzpicture}

}
\figcaption{\textbf{A sample representation of the Twitter network.}}
\label{fig:oi_example}
\vspace{-1mm}
\end{figure}
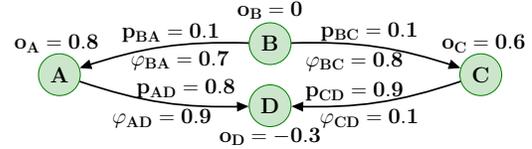

Having said that, the most obvious limitation of these models is that the notion of \emph{spread}, of a set of seed-nodes $S$, is defined as a function of the number of nodes that get activated using these seeds. The problem with this definition is two fold -- (1) Each node is considered to be contributing fully and positively towards the \emph{spread} of information about a content without considering the personal \emph{opinion} of the nodes which could be even negative, (2) A newly active node is always considered to perceive the information with the same intent as that of the node that activated it, while the former may tend to disagree with the latter owing to the \emph{interactions} between them in the past. These scenarios are not scarce in real-world settings, where a node can possess an \emph{opinion} (Def.~\ref{def:opinion}) towards a content which is the subject of IM. Moreover, usually past \emph{interactions} (Def.~\ref{def:interaction}) between a pair of nodes govern the way in which any new information originating from one of the two would be perceived by the other and vice-versa. To better understand these aspects, let us consider a (sample) real-world scenario in Example~\ref{example:oi1}.
\begin{example}
\label{example:oi1}
	Suppose Apple wants to market, utilizing the notion of viral-marketing, the iPhone 6 Plus using the Twitter network given a marketing budget of $1$ seed-node. Figure~\ref{fig:oi_example} portrays a snapshot of such a network consisting of $4$ nodes. Nodes $A$ and $C$ follow $B$, while $D$ follows $A$ and $C$. The opinion ($o$) of a node towards the new iPhone can be calculated using its opinion towards similar products in the past -- previous iPhone models or phones belonging to the same market segment. For example, the node $A$ has $o_A=0.8$, which means that she held highly positive opinions towards previous iPhone models or similar products. The opinion for other nodes can be computed in a similar way. Moreover, each edge possesses an additional parameter, called the interaction ($\varphi$) probability. In our example, the directed edge $BC$ has $\varphi_{BC}=0.8$ which means that the node $C$ perceives the information originating from the node $B$ with the same intent $4$ out of $5$ times. In other words, $C$ agrees with $B$ $4$ out of $5$ times while disagrees otherwise.
\end{example}

In the light of above discussions, it is intuitive that despite of being widely adopted, the fundamental (opinion-oblivious) diffusion models lack the power to leverage the expressibility of the social networks of today, namely -- \emph{Facebook}, \emph{Twitter} etc. Thus, there is a need to explore potential extensions, over and above these models, with the ability to better model information diffusion under real-world scenarios. To further this cause, we incorporate the concepts of \emph{opinion} and \emph{interaction} \cite{opinion, twitter_opinion, twitter_topic} from the social media domain and propose a new \emph{Opinion-cum-Interaction} (OI) model of information diffusion capable of addressing the limitations discussed above. Using the OI model, we define a more realistic notion of \emph{opinion-spread} (Defs.~\ref{def:op_spread} and~\ref{def:eop_spread}).

Let us try to analyze the importance of \emph{opinion-spread} using the example in Figure~\ref{fig:oi_example}. As discussed earlier, opinion-oblivious diffusion models, viz. IC, rely on the influence probability ($p$) alone without paying heed to the personal opinion of the activated nodes, which may even be negative. Thus, the likelihood of $C$ being chosen as a seed under the IC model is fairly high, owing to the high value of $p_{CD}$ which in turn results in high (opinion-oblivious) \emph{spread}. However, owing to $o_{D}$ being negative it will always decrease the overall \emph{opinion-spread} if activated. Thus, $C$ might not be the best choice for a seed-node. Moreover, a model capable of capturing the effect of \emph{opinions}, viz. OI, will penalize $C$ and thus decrease its likelihood of being chosen as a seed node.
\ignore{
\begin{example}
\label{example:oi}
	Suppose Apple wants to market, utilizing the notion of viral-marketing, their most recent iPhone 6 Plus using the Twitter network given a marketing budget of $1$ seed-node. Figure~\ref{fig:oi_example} portrays a snapshot of such a network consisting of $4$ nodes along with all the model-specific parameters. Nodes $B$ and $C$ follow $A$, while $D$ follows $B$ and $C$. For the sake of brevity, the analysis of spread in this scenario assumes (1) the IC model and (2) the OI model using IC at the first-layer. The influence probability (of an edge) $p$ possesses the same meaning as in the previous models. The opinion ($o$) of a node towards the new iPhone can be calculated using its opinion for similar products in the past -- previous iPhone models or phones belonging to the same market segment. For example, the node $B$ has $o_B=0.8$ which means that she held highly positive opinions towards previous iPhone models or similar products in the past. The opinion for other nodes can be computed in a similar way. Moreover, the OI model introduces an additional parameter with each edge, called the interaction ($\varphi$) probability. In our example, the directed edge $AC$ has $\varphi_{AC}=0.8$ which means that the node $C$ perceives the information originating from the node $A$ with the same intent $4$ out of $5$ times. In other words, $C$ agrees with $A$ $4$ out of $5$ times while disagrees otherwise. Under the IC model, the expected spread of the nodes are: $\sigma(A)=0.3628,\ \sigma(B)=0.8,\ \sigma(C)=0.9$ and $\sigma(D)=0$. Thus, $C$ is selected as the seed node. Under the OI model, the expected opinion-spread of the nodes are: $\sigma^o(A)=-0.022564,\ \sigma^o(B)=p_{BD}\big(\varphi_{BD}(o_D+o_B)/2 + (1-\varphi_{BD})(o_D-o_B)/2\big)=0.8\big(0.9(0.8-0.3)/2+0.1(-0.3-0.8)/2\big)=0.136,\ \sigma^o(C)=-0.351$ and $\sigma^o(D)=0$. Thus, $B$ is selected as the seed node. An interesting observation is that, the seed identified using the IC model, i.e. $C$, would have resulted in the worst-possible opinion-spread. This clearly shows the importance of the notion of opinion-spread under the OI model over opinion-oblivious spread using the IC model, in real-world scenarios where opinion of nodes and interaction between nodes play an important part in governing the process of information propagation.
\end{example}
}

We strengthen the above mentioned observation further, by plotting the \emph{opinion-spread} against the number of seeds for different real datasets (details in Sec.~\ref{sec:exp}). Figure~\ref{fig:model_motivate} shows that the opinion-spread obtained by the seeds selected using the OI model is much better when compared to that using the IC model, and thus, establishes the importance of \emph{OI} over the IC/LT models of information-diffusion with better conformance to the real-world scenarios.

As mentioned above, there has not been much research in devising \emph{opinion-aware} information-diffusion models, with \emph{IC-N} \cite{negIC} and \emph{OC} \cite{ovm} being the only two models capable of incorporating \emph{negative opinion} in the diffusion process. Despite of their capability to model opinions, they possess the following limitations.

{\bf(1) Constrained and Specific}: The IC-N model is rather simplistic, where a negatively activated node retains its opinion throughout its lifetime, while a \emph{quality-factor} qf (same for each node) defines the probability of transitioning to a negative opinion for the positively activated nodes. Although, a uniform qf allows retention of submodularity, which further facilitates extension of existing algorithms for IM to this setting, it renders the model too \emph{specific}.

{\bf(2) Simplistic diffusion dynamics}: The IC-N model possesses a strict constraint that a negatively activated node will always activate other nodes as negative. This assumption ignores the personal opinion of a node completely, and will be violated when the target node possesses a highly positive opinion about a content.

The OC model, where the change in opinion of a node is dependent upon its own opinion and the opinion of the nodes that activate it, does not possess the above two limitations. The OC model thus, facilitates a more involved process of modelling opinions at an added cost of losing submodularity.

{\bf(3) Lack capability to capture Interaction}: Both IC-N and OC lack the capability to capture the way in which information is perceived between a pair of nodes (\emph{interaction}) in the network. 

{\bf(4) Lack of backward-compatibility with fundamental diffusion models}: The IC-N model is tuned to work with IC at the first layer, while the OC model is designed to work with LT alone.

The OI model attempts to address these limitations by allowing -- (1) Nodes to possess different opinions, (2) Edges to possess an \emph{interaction} probability for modelling the perception of information between a node-pair, and (3) Capability of working with both IC and LT at the first layer. In addition to the analytical explanations, Figure~\ref{fig:model_motivate} also shows that the opinion-spread with seeds selected using OI is better when compared to those selected using OC. In summary, the OI model is a more generic, natural and realistic model when compared to IC-N and OC.
\ignore{
As mentioned above, there has not been much research in devising \emph{opinion-aware} information-diffusion models. The \emph{IC-N} model \cite{negIC} was one of the first models to incorporate negative opinion in information-diffusion. This model was rather simplistic, where a negatively activated node retained its opinion throughout its lifetime, while for the positively activated nodes a \emph{quality-factor} $qf$ (same for each node) defined the probability of transitioning to a negative opinion. The same $qf$ for each node facilitated retention of submodularity which further made it easy to extend existing efficient algorithms for IM to this setting. The OI model is much more generic as, unlike IC-N (same qf on all nodes), nodes are allowed to possess different opinions. The spread function ($\Gamma(\cdot)$) is no longer submodular. The IC-N model possesses a strict constraint that a negatively activated node will always activate other nodes as negative. This assumption ignores the personal opinion of a node completely, and will be violated when the target node possesses a highly positive opinion about a content. The OI model being generic does not possess such strict assumptions. In the IC-N model, the perception of information between a node-pair is governed by $qf$ (a node property), whereas the OI model allows \emph{interaction}, associated to an edge between the node pair, to govern the same and thus, OI is a more natural model. In addition, the OI model can work with both IC and LT models as the first layer while IC-N is just tuned for the IC model. Next, the OC model \cite{ovm} was introduced, where the change in opinion of a node is dependent upon its own opinion and the opinion of the nodes that activate it. Similar to the OI model, the $\Gamma(\cdot)$ function is not submodular for the OC model as well. However, the OC model possesses the following limitations. (1) The OC model has no notion of interaction, and hence modelling of how the information is perceived between a pair of nodes is not possible and (2) The OC model is designed to work with the $LT$ model alone. Figure~\ref{fig:model} also shows that the opinion-spread with seeds selected using OI is better when compared to those selected using OC. The OI model handles these limitations with being a more generic and realistic model when compared to OC \cite{ovm}.
}
\ignore{
\begin{figure}[t]
\centering
\includegraphics[width = 0.4\linewidth]{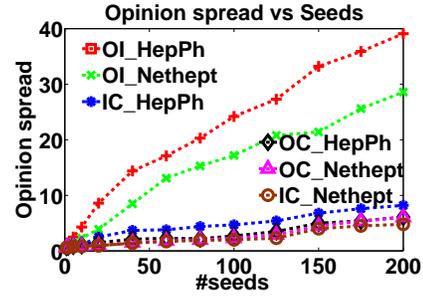}
\figcaption{Opinion spread for different diffusion models.}
\label{fig:model_motivate}
\end{figure}

	\subfloat[Opinion-Spread (OI vs IC)]
	{
		\scalebox{0.3}{
			\includegraphics[width = 0.99\linewidth]{motivate_model}
		}
		\label{fig:model_motivate}
	}
\end{figure}

\begin{figure}[t]
\centering
	\subfloat[OI vs OC and IC]
	{
		\scalebox{0.33}{
			\includegraphics[width = 0.99\linewidth]{motivate_model}
		}
		\label{fig:model}
	}
	\subfloat[$\lambda=1$ vs $\lambda=0$]
	{
		\scalebox{0.35}{
			\includegraphics[width = 0.99\linewidth]{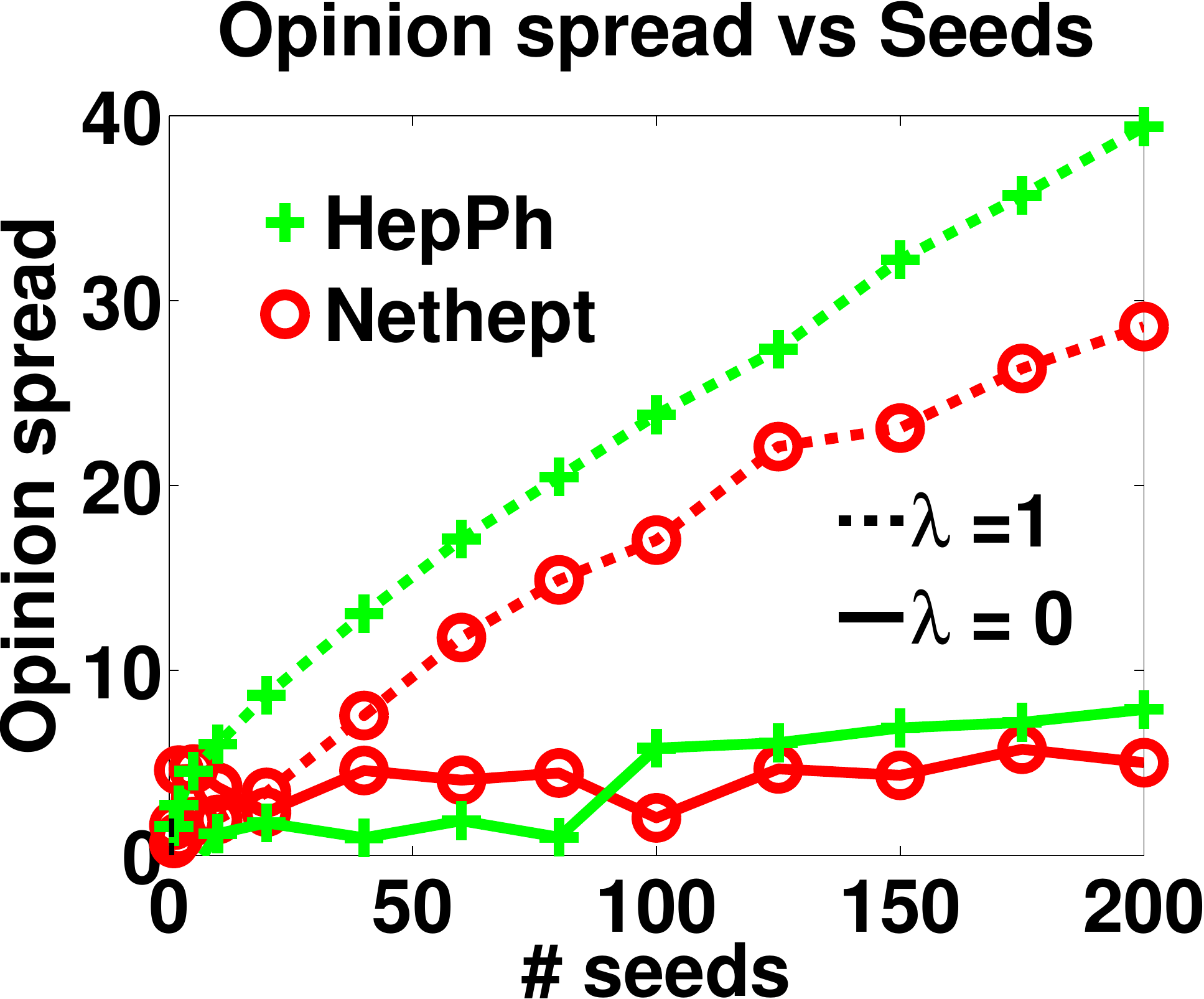}
		}
		\label{fig:objective}
	}
\figcaption{\textbf{Opinion spread for different diffusion models.}}
\label{fig:model_motivate}
\end{figure}
}

\begin{figure}[t]
\centering
	\scalebox{0.65}{
		\includegraphics[width = 0.99\linewidth]{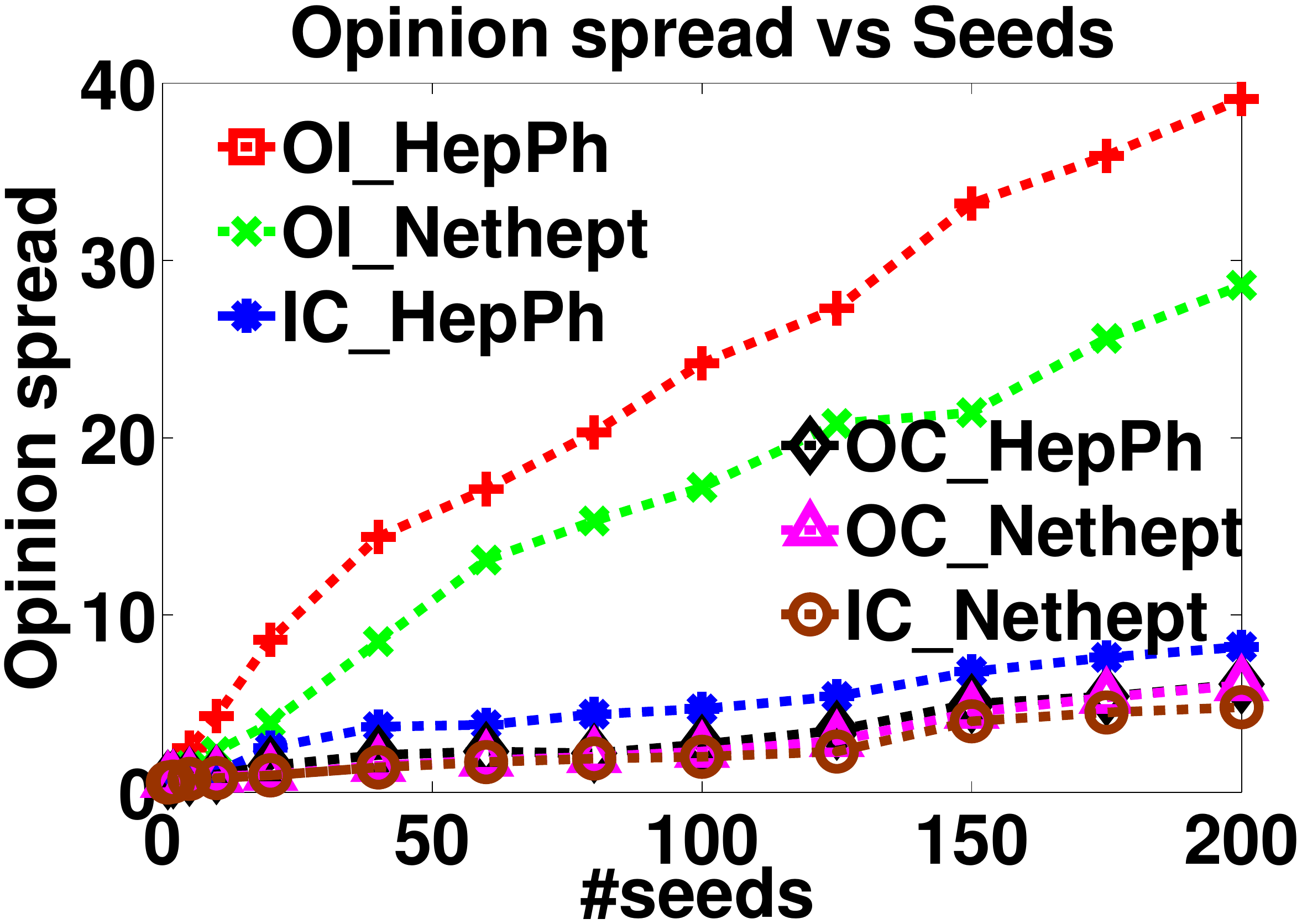}
	}
\figcaption{\label{fig:model_motivate}\textbf{Opinion spread for different diffusion models.}}
\end{figure}

Having successfully motivated (both qualitatively and empirically) and differentiated (from current state-of-the-art) the OI model, we propose a novel problem (MEO) of maximizing the \emph{spread} of information under opinion-aware settings. Since \emph{MEO} is NP-hard and difficult to approximate within a constant ratio, we incorporate ideas from the \emph{opinion-oblivious} scenarios to design \emph{scalable} algorithms for the \emph{opinion-aware} case. We propose \emph{OSIM}, capable of maximizing the \emph{opinion-spread} while accommodating for the change of \emph{opinion} as the information propagates. To this end, we first analyze a representative set of algorithms for IM under the fundamental diffusion models.

\ignore{
\textbf{Contributions.} In summary, we make the following contributions. We propose a \emph{holistic} solution to the IM problem. We introduce the OI model (Sec.~\ref{subsec:model}) which, to the best of our knowledge, is the most generic information diffusion model with the ability to address most of the limitations in the existing models and closely mirror real-world scenarios. Under this model, we propose a novel problem of maximizing the effective opinion of influenced users, called \emph{MEO} (Sec.~\ref{subsec:prob_formulation}) and present its tractability analysis (Sec.~\ref{subsec:model_analysis}). Since \emph{MEO} is not approximable, after analyzing the state-of-the-art algorithms for the opinion-oblivious case we first propose \emph{EaSyIM} (Sec.~\ref{subsubsec:easyim}) an effective, efficient and scalable algorithm for the IM problem and then extend it further to design \emph{OSIM} (Sec.~\ref{subsubsec:osim}) to solve the \emph{MEO} problem. Both of our algorithms run in linear time ($O(k\diameter(m+n))$) and space ($O(n)$). We provide a thorough theoretical analysis (Sec.~\ref{subsec:score_analysis}) for both the algorithms to prove their effectiveness. We also believe that such an in-depth analysis for the class of algorithms that estimate influence using a function of simple paths is not present in previous works. Finally, through an in-depth empirical analysis (Sec.~\ref{sec:exp}) we show that our algorithms are the first to possess the capability of handling \emph{huge} graphs on \emph{commodity-hardware} and provide the best trade-off between \emph{running-time} and \emph{memory-consumption} while retaining the quality of \emph{spread}.
}
Although, the Kempe's {\kempegreedy} algorithm \cite{kempe} provides the best possible approximation guarantees for the IM problem, it is inefficient. The CELF++ \cite{celfPlus} algorithm with optimizations over {\kempegreedy} renders it as the most efficient technique within this class of algorithms. However, these techniques never improved the asymptotic time-complexity of IM and thus, cannot be employed for large graphs. To further this, recent works employed the use of \emph{sampling} with \emph{memoization} \cite{opt1,tim,skim,imm}, with TIM$^{+}$ and its improvement IMM being the most efficient, to achieve superior efficiency while retaining approximation guarantees. However, these algorithms cannot be termed \emph{scalable} owing to their \emph{exorbitantly high} memory footprint. For example, the memory footprint of TIM$^{+}$ can be as high as 100 GB for a graph with $1M$ nodes and $3M$ edges which is not uncommmon these days. Owing to these limitations, neither CELF++ nor TIM$^+$ serve as potential candidates for extensions to \emph{opinion-aware} scenarios. Our proposition, \emph{EaSyIM}, tackles these problems by incorporating the idea that influence of a node can be estimated using a function of the number of simple paths starting at that node. Note that paths of length $l$ from a node $u$ can be calculated as the sum of all paths of length $l-1$ from its neighbors. We exploit this idea to devise data-structures and algorithms capable of running in \emph{linear space} and \emph{time}. This renders our algorithm \emph{scalable} yet \emph{efficient}. Therefore, \emph{EaSyIM} is chosen for extensions to the \emph{opinion-aware} case to produce \emph{OSIM}, which being fundamentally similar possesses the same analysis as \emph{EaSyIM}.\\\\
\textbf{Contributions.} In summary, we make the following contributions. We propose a \emph{holistic} solution to the IM problem. Towards that,
\begin{itemize}
	\item We introduce the OI model (Sec.~\ref{subsec:model}) which, to the best of our knowledge, is the most generic information diffusion model with the ability to address most of the limitations in the existing models and closely mirror real-world scenarios. In addition to this, we are the first to motivate the need of \emph{opinion-aware} models by citing a novel application – analyzing churn, using real-world data (Sec.~\ref{subsubsec:analyzing_churn}).
	\item Under OI, we propose a novel problem of maximizing the effective opinion of influenced users, called \emph{MEO} (Sec.~\ref{subsec:prob_formulation}).

	\item We propose effective, efficient and scalable algorithms (Sec.~\ref{sec:alg}), namely -- \emph{OSIM} and \emph{EaSyIM}, that run in linear time $O(k\diameter(m+n))$ and space $O(n)$.
	\item We provide thorough theoretical analyses (Sec.~\ref{subsec:score_analysis}) to prove the effectiveness of our algorithms. We believe that such an in-depth analysis for the algorithms that estimate influence as a function of simple paths is not present in the literature.
	\item Through an in-depth empirical analysis (Sec.~\ref{sec:exp}), we show that our algorithms are the first to possess the capability of handling \emph{huge} graphs on \emph{commodity-hardware} and provide the best trade-off between \emph{running-time} and \emph{memory}-\emph{consumption}, while retaining the quality of \emph{spread}.
\end{itemize}
\ignore{
\begin{figure}[t]
\centering
	\subfloat[Time]
	{
		\scalebox{0.34}{
			\includegraphics[width = 0.99\linewidth]{motivate_time}
		}
		\label{fig:time_motivate}
	}
	\subfloat[Memory]
	{
		\scalebox{0.34}{
			\includegraphics[width = 0.99\linewidth]{motivate_memory}
		}
		\label{fig:memory_motivate}
	}
\figcaption{Efficiency and scalability comparison. (a) {\bf Time}: \emph{EaSyIM} (l=3) vs CELF++. (b) {\bf Memory}: \emph{EaSyIM} (l=3) vs TIM$^+$ ($\epsilon=0.1$).}
\label{fig:algorithm_motivate}
\end{figure}
}

The rest of the paper is organized as follows. In Section~\ref{sec:problem}, we introduce a generic opinion-aware information diffusion model and present the problem statement. Section~\ref{sec:alg} describes our algorithm and its analysis. In Section~\ref{sec:exp}, experimental results are presented. Section~\ref{sec:related} highlights the related work before Section~\ref{sec:conc} concludes.

\ignore{
Humans have been termed social animals, and possess the never ending need of interacting with other human beings. This truth has probably been one of the reasons for the unprecedented growth of online social networks over the last few years. These networks came into existence to create and maintain social relationships at various levels such as family, friends and professional colleagues. Researchers and commercial organizations have found several applications in these networks viz. personalized/community based product recommendations \cite{reco1,reco2}, advertisement targeting \cite{adTarget}, marketing \cite{market}, estimating the propagation of diseases to model epidemics \cite{disease}, mining the trending topics, predicting election results or for that matter planning election campaigns \cite{election} etc. Specifically in marketing there is an important application towards identifying a set of people who can help propagate certain information (e.g. about a newly launched product) maximally within a network. A person in a network is preferred if she has higher influence on her neighbors leading to higher spread of information. A technical formulation of this application was first proposed as {\it Influence maximization} in 2003 by Kempe et. al.~\cite{}. 

Given a network $G(V,E)$; $|V|=n,\ |E|=m$, with edge weights ($p(e) \mid e \in E$) denoting the pair-wise influence probabilities, and a budget constraint $k$ the objective of influence maximization is to select a set $S$ of $k$ seed-nodes ($|S|=k$) with the ability to maximize the spread of information over this network. Diffusion
models are used to explain and simulate the spread of information
in social networks. Kempe et al. \cite{kempe} in their seminal work proved that finding an optimal solution for the influence maximization problem is NP-Hard and were the first to prove that a simple greedy algorithm can provide the best approximation guarantees in polynomial time. They incorporated the use of two fundamental diffusion models -- {\it Independent Cascade (IC)} and {\it Linear Threshold (LT)} for information propagation. However, the algorithm proposed by them had two sources of inefficiency. The first is that it took $O(kmn)$ time to produce a solution, while the second one is that it requires an additional factor of a large number of Monte Carlo (MC) simulations ($\approx10K$) to obtain the expected value of the \emph{spread}.

Over the next decade or so, a number of work has appeared to address the first issue with \emph{CELF++} \cite{celfPlus} being the most efficient of all, but there has not been much work in improving the second. More recently, Tang et al. \cite{tim} have come up with an algorithm (TIM)\footnote{For notation and details please refer \cite{tim}.} that runs in $O((k+l)(m+n)\log {n}/\epsilon^{2})$ expected time and produces a $(1-\frac{1}{e}-\epsilon)$-approximate solution, where $\epsilon$ is a constant, with probability as high as $1-{n^{-l}}$. While this is the fastest known algorithm for influence maximization it cannot be termed \emph{scalable} as it has a high memory footprint. The worst case space complexity of TIM is $O(n^2\log{{n \choose k}}/\epsilon^{2})$, which can be very high for small values of $\epsilon$. For example, the memory footprint of TIM can be as high as 100 GB for a graph with a million nodes and close to 3 million edges (Details in Sec.~\ref{sec:exp}) which is not unreasonable today with \emph{Facebook}, \emph{Twitter} etc. have over billions of nodes and trillions of edges. {\color{red} does twitter have billion nodes}

Almost all the above work on influence maximization assumed propagation of {\it any} information is desirable. However, in real life settings concerning spread of information, it is important to distinguish between favorable and unfavorable information propagation. For example, a marketing company would not like to select an influential person as seed who has been against their line of products. In fact a more influential person in this case would be less preferred than a less influential person. Hence it is important to distinguish between positive and negative (and possibly neutral) information propagation. While Zhang et. al. have modeled this problem as every individual (node) having certain polarity towards information under consideration~\cite{ovm}, Li et. al. have suggested having links with polarities indicating positive and negative relationships (friend and foe)~\cite{plosone}. 

In this paper, we propose an efficient algorithm \emph{ASIM} for scalable polarized influence maximization which provides the best tradeoff between \emph{memory-consumption} and \emph{running-time} and is capable of handling real-world large scale networks on moderately sized machines. We argue that our algorithm can be efficiently parallelized since each MC simulation is independent of the other. Moreover, a single iteration of \emph{ASIM} takes $O(kd(m+n))$ (Details in Sec.~\ref{sec:alg}) which is faster than TIM thus effectively it can perform better on overall time\footnote{Considering the ease of availability of multiple processing units (cores) in a single machine when compared to large amount of memory (RAM). Details in Sec.~\ref{sec:exp}.} while keeping the memory footprint $150-200$ times smaller when compared to the latter.

{\bf Contributions: We can use the flow of information as shared by NegIC and OVM, specifically NegIC some of our proofs are on the same lines as NegIC. (DAG proofs)
}
}
\ignore{
Social networks have become pervasive owing to the exponential growth in their popularity. The scale at which these networks operate today is humongous -- \emph{Facebook}, \emph{Twitter} etc. have over billions of nodes and trillions of edges. This wide-spread reach paves the way for a host of applications with huge impact. The influence maximization problem with applications in \emph{viral marketing} is one such example.

The exponential growth in the popularity of social networks has paved the way for a host of applications that benefit largely from the data derived using these networks.  Under the realistic setting of a constrained budget, \emph{Viral Marketing} has been widely studied as the influence maximization problem in the literature. 

More formally, given a network $G(V,E)$; $|V|=n,\ |E|=m$, with edge weights ($p(e) \mid e \in E$) denoting the pair-wise influence probabilities, and a budget constraint $k$ the objective of influence maximization is to select a set $S$ of $k$ seed-nodes ($|S|=k$) with the ability to maximize the spread of information over this network.

Kempe et al. \cite{kempe} in their seminal work proved that finding an optimal solution for the influence maximization problem is NP-Hard and were the first to prove that a simple greedy algorithm can provide the best approximation guarantees in polynomial time. They incorporated the use of two fundamental diffusion models -- Independent Cascade (IC) and Linear Threshold (LT) for information propagation. However, the algorithm proposed by them had two sources of inefficiency. The first is that it took $O(kmn)$ time to produce a solution, while the second one is that it requires an additional factor of a large number of Monte Carlo (MC) simulations ($\approx10K$) to obtain the expected value of the \emph{spread}.

Considerable amount of work has been done to cater to the first aspect -- optimizing the running time of this greedy algorithm, with \emph{CELF++} \cite{celfPlus} being the most efficient of all, but there has not been much work in improving the second. More recently, Tang et al. \cite{tim} have come up with an algorithm (TIM)\footnote{For notation and details please refer \cite{tim}.} that runs in $O((k+l)(m+n)\log {n}/\epsilon^{2})$ expected time and produces a $(1-\frac{1}{e}-\epsilon)$-approximate solution, where $\epsilon$ is a constant, with probability as high as $1-{n^{-l}}$. While this is the fastest known algorithm for influence maximization it cannot be termed \emph{scalable} as it has a high memory footprint. The worst case space complexity of TIM is $O(n^2\log{{n \choose k}}/\epsilon^{2})$, which can be very high for small values of $\epsilon$. For example, the memory footprint of TIM can be as high as 100 GB for a graph with a million nodes and close to 3 million edges (Details in Sec.~\ref{sec:exp}). This huge requirement is tough to be honored by commodity hardware.

The rest of the paper is organized as follows. In Section~\ref{sec:problem}, we introduce a generic opinion-aware information flow model before the problem statement is presented. Section~\ref{sec:alg} describes our algorithm and its analysis. In Section~\ref{sec:exp}, experimental results are presented. Section~\ref{sec:related} highlights the related work before Section~\ref{sec:conc} concludes.
}

\section{Opinion Aware IM}
\label{sec:problem}

In this section, we first introduce the basic concepts of the IM problem and build upon them to describe the Opinion-cum-Interaction (OI) model for information diffusion in detail. The notations used in the rest of the paper are summarized in Table~\ref{tab:terminology}.

\subsection{Preliminaries}
\label{subsec:prelims}

The objective of the IM problem is to capture the dynamics of information diffusion for maximizing the spread of information in a network. To this end, we first define the notions of \emph{seed} and \emph{active} nodes in the context of fundamental (opinion-oblivious) information diffusion models, namely -- IC and LT. Next, we use these concepts to define the notion of \emph{spread} of information in a network.

\begin{definition}[Seed Node]
\label{def:seed}
	A node $v \in V$ that acts as the source of information diffusion in the graph $G(V,E)$ is called a seed node. The set of seed nodes is denoted by $S$.
\end{definition}

\begin{definition}[Active Node]
\label{def:active}
	A node $v \in V$ is deemed active if either (1) It is a seed node ($v \in S$) or (2) It receives information, under the dynamics of information diffusion models, from a previously active node $u \in V_{(a)}$. Once activated, the node $v$ is added to the set of active nodes $V_{(a)}$.
\end{definition}

Given a seed node $s \in S$ and a graph $G(V,E)$, an information diffusion model $\mathcal{I}$ defines a step-by-step process for information propagation. Having defined the notions of seed and active nodes, we now introduce the dynamics of the IC and LT models. For both the IC and LT models, the first step requires a seed node $s \in S$ to be activated and added to the set of active nodes $V_{(a)}$. Under the IC model, at any step $i$ each newly activated node $u \in V_{(a)}$ gets one independent attempt to activate each of its outgoing neighbours $v \in \Out(u)$ with a probability $p_{(u,v)}$. However, under the LT model, a node $v$ gets activated if the sum of weights $w_{(u,v)}$ of all the incoming edges $(u,v)$ originating from active nodes $\forall u \in \In(v)_{(a)}$\footnote{$X_{(a)}$, for any set $X$, denotes the set of active nodes in that set.} exceeds the activation threshold $\theta_v$ of v, i.e., $\sum_{u \in \In(v)_{(a)}} w(u,v) \geq \theta_v$. Once a node becomes active, it remains active in all the subsequent steps. This diffusion process runs for each seed-node $s \in S$ until no more activations are possible. Eventually, the cardinality of the set of active nodes $|V_{(a)}|$, barring the number of seeds $k=|S|$, constitute the spread of a given set of seed nodes. Formally,
\vspace{-0.3mm}
\begin{definition}[Spread]
\label{def:spread}
	Given an information diffusion model $\mathcal{I}$ (opinion-oblivious), the spread $\Gamma(S)$ of a set of seed nodes $S$ is defined as the number of nodes, barring the nodes in $S$, that get activated using these seed nodes. Mathematically, $\Gamma(S)=|V_{(a)}|-|S|$.
\end{definition}
\vspace{-0.3mm}

\begin{table}
\centering
\small
\scalebox{0.85}{
\begin{tabular}{|c|c|}\hline
\textbf{Item} & \textbf{Definition}\\\hline
\hline
$V$ & Set of vertices; $|V| = n$. \\\hline
$E$ & Edge set for the nodes in $V$; $|E|=m$.\\\hline
$\diameter$ & Diameter of the graph.\\\hline
$S$ & Set of seed nodes.\\\hline
$k=|S|$ & Number of seed nodes.\\\hline
$\Out(u)$ & Set of outgoing neighbours of $u$. \\\hline
$\In(v)$ & Set of incoming neighbours of $v$. \\\hline
$p_{(u,v)}$ & Influence probability of $u$ on $v$ for IC.\\\hline
$\theta_v$ & Activation threshold of $v$; $\theta_v \in [0,1]$.\\\hline
$o_v$ & Personal opinion of $v$; $o_v \in [-1,1]$.\\\hline
$w_{(u,v)}$ & Weight on the edge $u \rightarrow v$ for LT; $w_{(u,v)} \in [0,1]$.\\\hline
$\varphi_{(u,v)}$ & Interaction probability from $u$ to $v$.\\\hline
$\flow(S)$ & Spread obtained by set of seed nodes ($S$).\\\hline
$\sigma(S)$ & Expected value of spread by seeds in $S$; $\expectation (\flow(S))$.\\\hline
$V_{(a)}$ & The set of active nodes; $V_{(a)} \subseteq V$. \\\hline
$\lambda$ & Penalty parameter on negative opinion spread.\\\hline
$l$ & Maximum path length for score-assignment, $1\leq l\leq\diameter$.\\\hline
$\mathcal{L}(u\leadsto v)$ & Length of a given $u\leadsto v$ path.\\\hline
$\gamma_{v}(u)$ & Contribution of $v$ to $\sigma(u)$, using approximate scores (Sec.~\ref{subsec:score}).\\\hline
$\gamma^{*}_{v}(u)$ & Contribution of $v$ to $\sigma(u)$, using {{\kempegreedy}}\cite{kempe}.\\\hline
$\pathset_{uv}$ & The set of all $u$-$v$ paths; $\{u\leadsto v\}$\\\hline
$t_{(u,v)}$ & \#$u$-$v$ paths -- \#paths starting at $u$ and ending at $v$.\\\hline
$s^d_{(u,v)}$ & Contribution of $v$, via a $d$-length $u$-$v$ path, to $\sigma(u)$.\\\hline
$\Delta^l(u)$ & Score assigned to $u$, using all $u\leadsto v$ paths of length $\leq l$.\\\hline

\end{tabular}
}
\tabcaption{\textbf{Summary of the notations used.}}
\label{tab:terminology}
\end{table}

Given the above defined concepts, the IM problem aims at identifying a set of seed nodes capable of maximizing the \emph{expected spread} of information in a network under a fixed budget on the number of seed nodes. More formally, Given a graph $G=(V,E)$, a fundamental (opinion-oblivious) information diffusion model $\mathcal{I}$ (IC/LT) with specifics defined as above and a budget $k$, find a set of seed nodes, $S \subseteq V \mid k=|S|$, that maximizes the expected value of information spread $\sigma(S)=\expectation[\flow(S)]$ in this graph.

\subsection{Opinion-cum-Interaction (OI) Model}
\label{subsec:model}
The OI model for information diffusion serves as an extension over the IC and LT models to facilitate opinion-aware IM. The fundamental models are modified to include a second layer attributed to modelling the diffusion and change of \emph{opinion}(Def.~\ref{def:opinion}) in the network. As opposed to the IC/LT models, where a newly activated node (oblivious to its \emph{opinion}) is always considered to be contributing positively towards the information spread (Def.~\ref{def:spread}), the OI model considers the \emph{spread} (opinion-spread, Defs.~\ref{def:op_spread} and~\ref{def:eop_spread}) of information under an \emph{opinion-aware} scenario -- where the contribution of a newly activated node could as well be negative.

\ignore{
\subsubsection{Motivation}
\label{subsubsec:model_motivate}
We explained the need for opinion aware information propagation in social networks in Sec.~\ref{sec:intro}. An opinion consists of two key components: a \textit{target} and a \textit{sentiment}  on the target, where target can be any entity or aspect of the entity about which an opinion has been expressed, and sentiment is a positive, negative, or neutral sentiment, or a value (say [-1,+1]) expressing the strength/intensity of the sentiment $s$~\cite{opinion}.

Opinion belongs to an individual and it is their personal subjective preference. For example, Alice is positive (or +0.85) about Apple iPhone 5, but negative (-0.35) about Samsung Galaxy note. Hence if a viral marketer is attempting to spread awareness about an Apple product then opinion value based on similar Apple products (0.85 in our case) should be attached to the node corresponding to Alice. For Samsung products the value attached would be 0.35. It is not a important for this work that how one arrives at these numbers but we assume these are available apriori. 

Now consider Bob who is Alice's connection in a social network. Depending on how much influence Alice has on Bob, Bob will receive and possibly propagate information shared by Alice. This is the commonly used influence in prior art. However Alice can positively or negatively influence Bob depending on relationship they share. Hypothetically based on past interactions, it is observed that Bob and Alice rarely have exhibited similar preferences with respect to topics they discussed. Hence, if Alice endorses a new Apple product it is possible that Bob may actually get negatively activated even though he had a neutral opinion (0) to Apple products. Hence each edge should have a probability depicting if the node at one end of the edge  has a positive, negative or neutral influence to the node at the other end.

\ignore{
1) In a realistic setting where opinions play a significant part\\

2) opinion is personal preference towards newly occuring event and other blah blah...\\

3) Opinion possess two things: sentiment and value.\\
    Describe learning from the past using examples on similar topics so that it is trivial to understand.\\
    This is an estimate of how a user behaves to a new event using his behaviour on previous similar events.\\

4) past interactions on several unrelated facets describe how the two users behave overall..\\
    Describe cases using examples so that it is easy to understand.\\
    u to v might not be same as v to u. Explain in the example.\\
    This probability is an estimate of how any possible user-pair behave collectively.\\
}
}

The OI model can be easily tuned, with minor modifications, to work with both IC and the LT models. The specifics of OI, thus, change depending on the underlying fundamental information diffusion model. Before moving to the model specifics, it is useful to introduce the notions of \emph{opinion} ($o$) of a node and \emph{interaction} ($\varphi$) between two nodes. Usually, the objective of IM is to maximize the spread of information about a specific content viz. product, topic, person, event etc. Having said that, the opinion of a node is always defined in the context of a particular content to denote the personal preference of this node towards that content. Formally,

\vspace{-0.3mm}
\begin{definition}[Opinion]
\label{def:opinion}
	The opinion of a node $v \in V$ consists of two sub-components -- (1) an orientation $($\{negative,\ neutral,\ positive\}$)$ towards a content and (2) strength $([0,1])$ that quantifies its preference towards that content. The opinion of a node is, thus, denoted as $o_v \in [-1,1]$.
\end{definition}
\vspace{-0.3mm}

As opposed to \emph{opinion} which is associated with properties of a node alone, \emph{interaction} aims at capturing the effect of the dyadic relationship between any two nodes to better model the process of information diffusion. Formally,

\ignore{
    The interaction (directed) between two nodes $u,v \in V$, denoted as $\varphi_{(u,v)} \in [0,1]$, is defined as the probability with which the contribution of $u$ to the polarity of the opinion associated with $v$ is the same as its own polarity.
	The interaction probability (directed) between two nodes $u,v \in V$, denoted as $\varphi_{(u,v)} \in [0,1]$, is defined as a fraction of the times an information content shared by the node $u$ gets accepted by the node $v$ with the same orientation as that of $u$.
}

\vspace{-0.3mm}
\begin{definition}[Interaction]
\label{def:interaction}
	%
	The interaction probability (directed) between two nodes $u,v \in V$, denoted as $\varphi_{(u,v)} \in [0,1]$, is defined as a fraction of the times an information content shared by $u$ gets accepted by $v$ with the same orientation as that of $u$.
\end{definition}
\vspace{-0.3mm}

For example, if two nodes $u$ and $v$ agree with each other $1$ out of $5$ times in the past, then $\varphi_{(u,v)}=\varphi_{(v,u)}=1/5=0.2$.
As discussed in Sec.~\ref{sec:intro}, the opinion of a node towards a new content can be estimated from the previously held opinions by this node towards similar contents in the past. As opposed to this, the interaction probabilities (directed, $\varphi_{(u,v)}$ might not be equal to $\varphi_{(v,u)}$) between two nodes is rather generic and can be estimated\footnote{We have discussed one possible way to estimate opinion and interaction. These notions are rather generic and other approaches can be employed for their estimation as well.} by accounting for all possible interactions between these nodes in the past. Next, we describe the specifics of the OI model.


We model the underlying network as a directed graph $G=(V,E)$ with parameters $p,w,\theta,o,\varphi$, where $|V| = n$ and $|E| = m$ denote the set of vertices and edges respectively. The parameters have similar meanings as discussed earlier in this section. Oblivious to the underlying information diffusion model (IC or LT), each node $v \in V$ possesses an opinion towards a new content denoted by $o_v \in [-1,1]$. Specifically the sign of $o_v$ indicates the \emph{orientation} and the value denotes the \emph{strength}. More formally, $o_v > 0$, $o_v = 0$, and $o_v < 0$ indicate positive, neutral and negative orientation respectively. In addition, $o_v=-1$ denotes a strong negative opinion while $o_v=+1$ denotes a strong positive opinion. Moreover, each edge $(u,v) \in E$ possess $\varphi_{(u,v)} \in [0,1]$, that refers to the probability of a node $v$ acquiring the same opinion as that of $u$, considering only the contribution from node $u$. $\varphi_{(u,v)}=0$ indicates that $v$ never agrees with $u$, while $\varphi_{(v,u)}=0.5$ indicates that $u$ agrees with $v$ half of the time. The dynamics of the model are as follows.

Under the IC model, the first step involves activation of each seed node $s$, in the selected seed set $S$ ($|S|$ = $k$), with its opinion\footnote{The final opinion of the seeds is same as their initial personal opinion, $o_s'=o_s$.} $o_s, \forall s \in S$, while all other nodes remain inactive. At any step $i$, if a node $u$ (which was activated at step $i-1$) activates another node $v$ then the final opinion ($o_v'$) of $v$ is dependent upon both, its initial opinion $o_v$ and the final opinion ($o_u'$) of the node $u$. More formally, each node $u$, contributes $o_u'$ with a probability $\varphi_{(u,v)}$ and  $- o_u'$ with a probability $1 - \varphi_{(u,v)}$. The final opinion  of the node $v$ is $o'_v = \frac{o_v + (-1)^\alpha o_u'}{2}$, where $\alpha=0$ with a probability of $\varphi_{(u,v)}$ and $\alpha=1$ with a probability of $1-\varphi_{(u,v)}$.

Under the LT model, the first step witnesses the same set of initializations as done for OI under the IC model. At any step $i$, if a node $v$ gets activated (by the set of nodes $\In(v)_{(a)}$, activated at previous steps) then the final opinion $o_v'$ of $v$ is dependent upon both, its initial opinion $o_v$ and the final opinion $o_u'$ for all the nodes $u \in \In(v)_{(a)}$. The contributions made by each node to the opinion of the node $v$ is the same as discussed above for OI under the IC model. The final opinion of the node $v$ is $o'_v = (o_v + \frac{1}{|\In(v)_{(a)}|}\sum\limits_{u \in \In(v)_{(a)}}(-1)^{\alpha_{(u,v)}}o_u')/2$, where $\alpha_{(u,v)}=0$ with a probability of $\varphi_{(u,v)}$ and $\alpha_{(u,v)}=1$ with a probability of $1-\varphi_{(u,v)}$.

Once a node becomes active, it remains active with the same effective opinion in all the subsequent steps. The information propagation process runs until no more activations are possible. With these extensions, we aim to solve the IM problem under settings that closely mirror real-world scenarios. Before moving ahead, we revisit Example~\ref{example:oi1} and analyze the difference between \emph{spread} and \emph{opinion-spread} with stating the importance of the latter in real-world scenarios using Example~\ref{example:oi2}.

\begin{example}
\label{example:oi2}
	Let us consider the construction described in Example~\ref{example:oi1} with all the model-specific parameters present in Figure~\ref{fig:oi_example}. For the sake of brevity, the analysis of spread in this scenario assumes (1) the IC model and (2) the OI model using IC at the first-layer. The influence probability ($p$, of an edge) possesses the same meaning as for the fundamental models. Since $p_{AD}=0.8$, the expected spread of $A$ under the IC model $(\sigma(A))$ is $0.8$. Similarly for other nodes: $\sigma(B)=0.3628,\ \sigma(C)=0.9$ and $\sigma(D)=0$. Thus, $C$ is selected as the seed node. Since $D$ agrees with $A$ with a probability of $\varphi_{AD}$ and disagrees otherwise, the expected opinion-spread of $A$ under the OI model is expressed as $\sigma^o(A)=p_{AD}\big(\varphi_{AD}(o_D+o_A)/2 + (1-\varphi_{AD})(o_D-o_A)/2\big)=0.8\big(0.9(-0.3+0.8)/2+0.1(-0.3-0.8)/2\big)=0.136$. Similarly for other nodes: $\sigma^o(B)=-0.022564,\ \sigma^o(C)=-0.351$ and $\sigma^o(D)=0$. Thus, $A$ is selected as the seed node. An interesting observation is that, the seed identified using the IC model, i.e. $C$, would have resulted in the worst-possible opinion-spread. This clearly shows the importance of the notion of opinion-spread under the OI model, over opinion-oblivious spread using the IC model, in real-world scenarios where opinion of nodes and interaction between nodes play an important part in governing the process of information propagation.
	%
\end{example}

Next, we formally state the \emph{MEO} problem followed by its tractability analysis.

\subsection{Problem Formulation}
\label{subsec:prob_formulation}

Using the concepts described in the previous sections, we formally define the notion of spread under the opinion-aware (OI) model, called \emph{Opinion Spread}. While, under the opinion oblivious models spread can simply be stated as the total number of nodes that get activated with a given set of seed nodes, a more involved notion of spread is required under the \emph{opinion-aware} settings.

\begin{definition}[Opinion spread]
\label{def:op_spread}
Opinion spread of a seed set $S$, denoted by $\flow^{o}(S)$, is defined as the sum of final opinions of  the nodes in the activated set $V_{(a)}\setminus S$, when $S$ is the chosen seed set, i.e., $\flow^{o}(S) = \sum\limits_{v\in V_{(a)}\setminus S}o_v'$.
\end{definition}

\begin{definition}[Effective opinion spread]
\label{def:eop_spread}
    Effective opinion spread of  a seed set $S$, denoted by $\flow^{o}_{\lambda}(S)$, is defined as the weighted difference between the opinion spread of nodes with positive orientation and the  opinion spread of nodes with negative orientation, i.e., $\flow^{o}_{\lambda}(S) = \big(\sum\limits_{o_v'>0}o_v' - \lambda \sum\limits_{o_v'<0}|o_v'|\big); \forall v\in V_{(a)}\setminus S$, in the set of activated nodes $V_{(a)}\setminus S$, where $\lambda$ is the penalty on negative opinion spread.
\end{definition}

We call the problem of maximizing the effective opinion spread under the OI model as \emph{Maximizing the Effective Opinion of the Influenced Users (MEO)} problem, which is defined formally as follows.

\begin{problem}
	Given a graph $G=(V,E)$, the opinion-aware (OI) model with specifics defined as in Sec.~\ref{subsec:model} and a budget $k$, find a set of seed nodes, $S\subseteq V \mid k=|S|$, that maximizes the expected value of the effective opinion spread $\sigma^o_{\lambda}(S)=\expectation[\flow^{o}_{\lambda}(S)]$.
\end{problem}

\subsection{Properties of MEO under the OI model}
\label{subsec:model_analysis}

\subsubsection{NP-hardness}

\begin{lemma}
\label{thm:nphard}
The MEO problem is NP-hard. 
\end{lemma}

\begin{proof}
The IM problem is reducible to an instance of the MEO problem under the OI model, when $o_v=1$, $\forall v \in V$ and $\varphi_{(u,v)} = 1,$ $\forall (u,v) \in E$. It is known that any generalization of a NP-hard problem is also NP-hard. Since, the IM problem is NP-hard \cite{kempe}, the MEO problem is NP-hard as well.
\hfill{}
\end{proof}

\subsubsection{Submodularity}

A function $f(\cdot)$ is submodular if the marginal gain from adding an element to a set $S$ is at least as high as the marginal gain from
adding it to a superset of $S$. Mathematically, 
\begin{align*}
	f(S \cup \{x\}) - f(S) \geq f(T \cup \{x\}) - f(T)
\end{align*}
for all elements $x$ and all pairs of sets $(S,T)$, where $S \subseteq T$. For submodular and monotone functions, the greedy algorithm of iteratively adding the
element with the maximum marginal gain approximates the optimal solution within a factor of ($1 - 1/e$) \cite{submodular}.

\begin{lemma}
\label{thm:non_monotone_submodular}
The opinion spread as a function $\Gamma(\cdot)$ in a graph $G$, is neither monotone nor submodular.
\end{lemma}

\begin{proof}
Let us consider a bipartite graph $G(V,E)$ (Figure~\ref{fig:monotone}) containing two sets of nodes $X$ and $Y=V\setminus X$. The set $X$ contains $n_x=|X|$ nodes while the set $Y$ contains $n_y=|Y|, n_y \geq 2n_x$ nodes. Moreover, let us assume that the opinions of the nodes $\forall x_i \in X$ to be $o_{x_i}=+1$, while that of nodes $\forall y_j \in Y$ to be $o_{y_j}=0$. There exist directed edges from each node $x_i \in X$ to two consecutive nodes $y_{2i-1},y_{2i} \in Y$. Further, the influence probabilities are initialized as $p_{(u,v)}=1, \forall (u,v) \in E$ and the interaction probabilities $\varphi$ associated with all but the last two edges is $1$, which attain a value of $0$. Mathematically, $\varphi_{x_i,y_j}=1 \mid 1\leq i\leq n_x-1, j=2i-1$ or $j=2i$ while $\varphi_{x_i,y_j}=0 \mid i=n_x, j=2i-1$ or $j=2i$. An analysis of the properties of the spread function $\Gamma(\cdot)$, assuming the above construction, is presented next.

Consider a scenario where the set of seeds $S$ is initialized with a single node $x_i \mid 1 \leq i \leq n_x-1$. With this seed, both of the nodes $y_{2i-1}$ and $y_{2i}$ will get activated with final opinion $+1/2$. Hence, the effective spread $\Gamma^o(S) = 2\times(1/2) = +1$. Now if the node $x_{|X|}$ is added to the set $S$, the nodes $y_{2|X|-1}$ and $y_{2|X|}$ will also get activated. The final opinion of these newly activated nodes will be $-1/2$ each. The effective spread of the updated seed set is $\Gamma^o(S) = 1 - 1 = 0$. On adding another node $x_j \mid j\neq i$ and $1\leq j\leq n_x-1$ to the set $S$, the nodes $y_{2j-1}$ and $y_{2j}$ will become active with a final opinion of $1/2$ each. Thus, the effective spread is now $\Gamma^o(S) = 0+1 = 1$. In the above case it can be seen that the effective spread varied from $1\rightarrow0\rightarrow1$ on subsequent additions of nodes to the seed set. This clearly shows that the opinion spread function $\flow^o(S)$ is neither monotone nor submodular.
\hfill{}
\end{proof}	

\begin{figure}[t]
	\centering
	\subfloat[Submodularity]
	{
		\scalebox{0.39}{
		\begin{tikzpicture}[node distance=15mm,
		round/.style={fill=green!50!black!20,draw=green!50!black,text width=10mm,align=center,circle,thick}]
		\centering
		\node at (0,2.5) {\huge $\mathbf{o_s = 0}$};
		\node at (3.2,2.5) {\huge$\mathbf{o_t = 1}$};
		\node at (0,-7.5) {\huge Layer 1};
		\node at (3.3,-7.5) {\huge Layer 2};

		\node[round, scale=0.9, label={\huge $\mathbf{s_1}$}] (u1) {};
		\node[round, scale=0.9, label={\huge $\mathbf{s_{n_s-1}}$}] (u2) [below = 2.4cm of u1] {};
		\node[round, scale=0.9, label={\huge $\mathbf{s_{n_s}}$}] (um) [below =0.7cm of u2] {};

		\node[round, scale=0.5, minimum size=5pt, label={\LARGE $\mathbf{t_1}$}] (v1) [above right =.13cm and 2.7cm of u1 ] {};
		\node[round, scale=0.5, label={\LARGE $\mathbf{t_2}$}] (v2) [below = .6cm of v1] {};
		\node[round, scale=0.5, label={\LARGE $\mathbf{t_{2n_s-3}}$}] (v3) [below = 1.8cm of v2] {};
		\node[round, scale=0.5, label={\LARGE $\mathbf{t_{2n_s-2}}$}] (v4) [below = .5cm of v3] {};
		\node[round, scale=0.5, label={\LARGE $\mathbf{t_{2n_s-1}}$}] (v5) [below = 0.5cm of v4] {};
		\node[round,scale = 0.5, minimum size=5pt, label={\LARGE $\mathbf{t_{2n_s}}$}] (v6) [below =.5cm of v5] {};

		\draw[thick] (0,-2.4) ellipse (1.1cm and 4.4cm);
		\draw[thick] (3.3,-2.4) ellipse (1.1cm and 4.4cm);

		\draw[loosely dotted,thick, line width=2.5pt](0,-1.2) -- (0,-1.8);
		\draw[loosely dotted,thick, line width=2.5pt](3.3,-1.2) -- (3.3,-1.8);

		\path [->, -triangle 45,line width=1.4pt] (u1) edge (v1)
		edge (v2)
		(u2) edge (v3)
		(u2) edge (v4)
		(um) edge (v5)
		(um) edge (v6)
		; 
		\end{tikzpicture}
		}
		\label{fig:monotone}
	}
	\subfloat[Tractability]
	{
		\scalebox{0.35}{
			\begin{tikzpicture}[node distance=15mm,
			round/.style={fill=green!50!black!20,draw=green!50!black,minimum size=5mm,text width=10mm,align=center,circle,thick, line width=1.6pt}]
			\centering
			\node at (0,2) {\huge $\mathbf{o_x = 0}$};
			\node at (3.1,2) {\huge$\mathbf{o_y = \frac{1}{n}}$};
			\node at (6.2,2) {\huge$\mathbf{o_z = -\frac{1}{2n}}$};
			\node at (9.7,-4.5) {\huge$\mathbf{o_s = \frac{1}{n}-1}$};
			\node at (9.7,-3.5) {\huge \textbf{sink}};

			\node at (0,-6) {\huge\textbf{ Layer 1}};
			\node at (3.3,-6) {\huge \textbf{Layer 2}};
			\node at (6.6,-6) {\huge \textbf{Layer 3}};

			\node[round] (u1) {\huge $\mathbf{x_1}$};
			\node[round] (u2) [below of=u1] {\huge $\mathbf{x_2}$};
			\node[round] (um) [below =1cm of u2] {\huge $\mathbf{x_m}$};

			\node[round] (v1) [right =2cm of u1] {\huge $\mathbf{y_1}$};
			\node[round] (v2) [below of=v1] {\huge$\mathbf{y_2}$};
			\node[round] (vm) [below =1cm of v2] {\huge$\mathbf{y_n}$};

			\node[round] (x1) [right =2cm of v1] {\huge$\mathbf{z_1}$};
			\node[round] (x2) [below of=x1] {\huge$\mathbf{z_2}$};
			\node[round] (xm) [below =1cm of x2] {\huge$\mathbf{z_{m+n-2}}$};

			\node[round] (s) [rectangle, below right=.3cm and 2cm of x2] {\huge $\mathbf{s}$};

			\draw[thick] (0,-2) ellipse (1cm and 3.5cm);
			\draw[thick] (3.3 ,-2) ellipse (1cm and 3.5cm);
			\draw[thick] (6.6,-2.) ellipse (1cm and 3.5cm);

			\draw[loosely dotted,thick, line width = 1.6pt](0,-2.5) -- (0,-3);
			\draw[loosely dotted,thick, line width = 1.6pt](3.3,-2.5) -- (3.3,-3);
			\draw[loosely dotted,thick, line width = 1.6pt](6.6,-2.5) -- (6.6,-3);

			\path [->, ,-triangle 45,line width=1.6pt] (u1) edge (v1)
			edge (v2)
			edge (vm) 
			(u2) edge (v1)
			edge (vm)
			(um) edge (v1)
			edge (vm)
			(v1)edge (x1)
			edge (x2)
			edge (xm)
			(v2)edge (x1)
			edge (x2)
			edge (xm)
			(vm)edge (x1)
			edge (x2)
			edge (xm)
			(x1)edge (s)
			(x2)edge (s)
			(xm)edge (s)
			; 
			\end{tikzpicture}
		}
		\label{fig:submod}
	}
	\figcaption{Reductions for the submodularity and tractability analysis of MEO.}
	\label{fig:reductions_meo}
\end{figure}
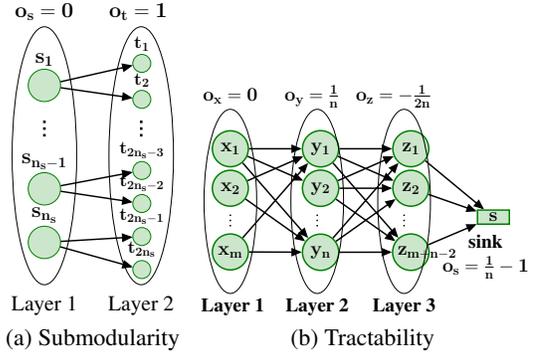

Using Lemma~\ref{thm:non_monotone_submodular}, it is evident that a greedy algorithm cannot produce a $1-1/e$ approximate solution to the \emph{MEO} problem. Next, we further prove that $\nexists$ any algorithm capable of approximating \emph{MEO} within a constant ratio.

\begin{theorem}
\label{thm:non_approximate}
Approximating MEO within a constant ratio is not possible unless $P=NP$.
\end{theorem}

\begin{proof}
We show that the classical set cover problem can be decided in polynomial time if an approximation, with a constant ratio, for MEO exists in polynomial time. To this end, we reduce the set cover problem to an instance of MEO. Given a set of elements $Q = \{ q_1,\ldots,q_n\}$ and a collection of subsets $R = \{R_1,\ldots,R_m\}$ where $R_i \subseteq Q, \forall i \in \{1,\ldots,m\}$, the set-cover decision problem returns true if $\exists C\subseteq R; |C|=k$ and $\cup C_j  = Q,\ \forall C_j \in C$. 

Now given the set-cover problem, we construct a graph as shown in Figure~\ref{fig:submod} to obtain an instance of MEO by adding three layers of nodes in addition to a sink node. For each subset $R_i \in R$, a node $x_i$ is added in the first layer of the graph with opinion $o_{x_i} = 0$. In the second  layer, for each element $q_i\in Q$, we add a node $y_i$ with opinion $o_{y_i}=\frac{1}{n}$. We add a third layer containing $m + n -2$ nodes denoted by $z_i$ with opinion $o_{z_i} =-\frac{1}{2n} $. Along with this, we add a sink node $s$ with opinion $o_s = -1+\frac{1}{n}$. Now, a directed edge  $(x_i,y_j)$ is  added iff $q_i \in R_j$. For each node $y_i$ in the second layer, edges $(y_i,z_j), \forall j \in \{1,2,\ldots,m+n-2\}$ are added in the graph. To finalize this construction, we add the edge $(z_i,s), \forall i \in \{1,2,\ldots,m+n-2 \}$ . 

For the IC model, all the edges $(u,v) \in E$ are assigned $p_{(u,v)} = 1$ and $ \varphi_{(u,v)} = 1$. Moreover, the LT model uses the same activation threshold for each node $\theta_u = 1$. It can be seen that for both these models, the seeds should be chosen from the first layer (any of the $x_i$'s), because they will activate $y_j$'s which in turn will activate any of the $z_i$'s and thus, $s$ would always be activated.

We instantiate MEO with $\lambda = 1$. Without loss of generality, assuming $\exists$ a set-cover of size $k$, $C = \{x_1,x_2,\ldots,x_k \}$, then choosing these nodes as seeds ensures the maximum spread. The final opinion of $y_i$, $o_{y_i}' = (0+\frac{1}{n})/2=\frac{1}{2n}$. Similarly, the final opinion of $z_i$, $o_{z_i}' = \frac{1}{2n}-\frac{1}{2n}=0$ and the final opinion of $s$, $o_{s}'=(0-1+\frac{1}{n})/2=-\frac{1}{2}+\frac{1}{2n}$. Hence, the spread = $n\frac{1}{2n} + 0  - \frac{1}{2} + \frac{1}{2n} = \frac{1}{2n}>0$. When a set cover of size $k$ does not exist, the maximum spread achieved is $|Q'|\frac{1}{2n} + 0   -\frac{1}{2} + \frac{1}{2n}  \leq (n-1)\frac{1}{2n} - \frac{1}{2} + \frac{1}{2n} =0$, where $Q', |Q'|\leq n-1$ is the maximum number of elements covered by the $k$ chosen sets. Hence, the maximum spread achieved when the set-cover does not exist is $0$.

If there is a polynomial time algorithm which approximates MEO within a constant ratio, then we can decide the set-cover problem in polynomial time using the above mentioned reduction. In other words, if an approximate algorithm gives a spread $\leq0$ on the reduced graph, then a set-cover does not exist while a set-cover exists if the spread $>0$. This renders the set-cover problem decidable in polynomial time,  which means P=NP. Therefore, approximating MEO within a constant ratio is NP-hard.
\hfill{}
\end{proof}

\ignore{
\begin{theorem}
 EaSyIM is NP-hard and approximating it within finite ratio   is NP-hard
\end{theorem}

\begin{proof}
It can be seen that when $o_v=1, \  \forall v \in V$ and $w_{(u,v)} = 1\forall (u.v) \in E$, the EaSyIM model reduces to the Influence maximization problem in LT or IC model whichever is being chosen by EaSyIM. Since Influence maximization is NP-hard, any generalization (EaSyIM) of a NP-hard problem is also NP-hard.

In order to prove the second part,  we show that the classical set cover problem can be decided in polynomial time if EaSyIM can be approximated in polynomial time. For this, we reduce the set cover problem to an instance of EaSyIM. According to the set-cover problem, given a set of elements $U = \{ e_1,e_2,...,e_n\}$ and a collection of subsets $S = \{S_1,..., S_m\}$ where $S_i \subset U\  \forall i \in \{1,2,...,m\}$, it returns true if $\exists C\subset S, |C|=k$ and $\cup C_i  = U\  \forall\  C_i \in C$. 

Now given this instance of set-cover, we generate the graph to get the instance of EaSyIM.
For each element $e_i\in U$, we add a vertex $u_i$ with orientation $o_{u_i} = 0$ and for each subset $S_i$, we add a vertex $v_i$  with orientation $o_{v_i} = \frac{1}{m}$ in the graph. Along with this, we add a sink node $s$ with orientation $o_s = -1$. Now, a directed edge  $(u_i,v_j)$ is  added iff $e_i \in S_j$. Along with this, we add $(v_i,s)$ edges $\forall i \in \{1,2,...,m \}$ . 
Incase of Independent Cascade model, All edges are assigned the probability $p = 1$ and $w = 1$.
Incase of Linear threshold model also, we use the same activation threshold for each node $\theta_u = 1$.

It can be seen that for both the models, the seeds should be chosen from any of the $u_i$ because they will eventually activate $v_j$ and $q$ would always be activated.

If $\exists$ a set cover of size $k$, for the given problem then we run an approximation of EaSyIM with $\lambda = 0$. Without loss of generality, assume the set cover $C = \{u_1,u_2,...,u_k \}$, then choosing these as the seed nodes ensures the maximum spread. The orientation of $v_i $, $o_{v_i} = \frac{1}{m}$ and the final orientation of $s$, $o_s' = -1 + \frac{1}{m}$. So the spread = $m\frac{1}{m} -1 + \frac{1}{m}$ = $\frac{1}{m}$. Assume there is no set cover of size $k$. Now the maximum spread  = $-1 + \frac{1}{m}  + |C|\frac{1}{m}   \leq -1 + \frac{1}{m}  + (m-1)\frac{1}{m} =0 $, where $C$ is the maximum number of elements covered by $k$ chosen sets. So the maximum spread when there is no set cover is 0. 

If there is a polynomial algorithm which gives a finite ratio of the EaSyIM, then we can decide the Set cover in polynomial time using the above mentioned reduction. If running approximate EaSyIM on reduced graph, gives the spread = 0, then the set cover does not exist and if the spread $>0$, then the set cover exists. So set-cover is decidable in polynomial time,  which means P=NP. So approximating EaSyIM is NP-hard.
\hfill{}
\end{proof}
}

\section{Algorithm}
\label{sec:alg}

\subsection{Outline of our Algorithm}
\ignore{
\begin{figure}[t]
\centering
\includegraphics[width = 0.6\linewidth]{images/overview_paper}
\figcaption{\textbf{Overview of our algorithm.}}
\label{fig:overview}
\vspace{4mm}
\end{figure}
}
To solve the \emph{MEO} problem efficiently, we propose a technique as outlined in Algorithm~\ref{algo:seed_select}. We first assign a score (line $4$), to each node $u \in V$ of the graph $G(V,E)$, by aggregating the contributions of all paths starting at $u$. We prove that if the score assignment is correct (we will specify exactly what is meant by ``correct'' later), then the expected value of the influence spread achieved using our algorithm is the same as that achieved using the {{\kempegreedy}} gold standard \cite{kempe}. The correctness of our score assignment algorithms, both for the opinion-oblivious (EaSyIM) and the opinion-aware (OSIM) case, holds perfectly for trees. 
Moreover, in a conclusive discussion, we show that the error introduced in case of graphs is small as well. Through detailed analyses and experiments we show that our results do not deviate much from that of the {\kempegreedy} algorithm.

On completion of the score assignment step, the node possessing the maximum score is selected as the seed node (lines $5$--$9$). In order to ensure the sets of nodes activated by each selected seed node to be disjoint; the last step of the seed-selection algorithm (line 11) keeps track of all the previously activated nodes as a set ($V_{(a)}$). This enables the score assignment step to discount the contributions of all the previously activated nodes in subsequent iterations. The above process continues until $k$ seeds are selected.


Next, we describe the score-assignment algorithms -- (1) \emph{EaSyIM} and its extension (2) \emph{OSIM} in detail.
The algorithms are fundamentally similar for both, \emph{opinion-oblivious} and \emph{opinion-aware}, cases.

\subsection{Score Assignment}
\label{subsec:score}

The score assignment step, similar to ASIM \cite{asim}, leverages the idea that the probability of a node $v$ to get activated by a seed node $u$ is dependent upon the number of simple paths from $u$ to $v$ in $G$. Thus, a simple function of the number of paths from a node $u$ to all other nodes $v \in V \setminus \{u\}$ can be used to assign a score to $u$. This score-assignment is further used to rank all the nodes $v \in V$ in an order determining their expected spread $\sigma(v)$\ignore{$\expectation[\flow(\{v\})]$}. To this end, the scores assigned to a node $u$ ($\Delta^l(u)$), is computed by aggregating the contributions of all $u\leadsto v$ paths of length at most $l$\footnote{$l$ is the maximum path length considered for score assignment, where $l\leq \diameter$. $\diameter$ is the diameter of the graph.} ($\mathcal{L}(u\leadsto v) \leq l$). We start with a description of the score assignment step for the opinion-oblivious (EaSyIM) case and then follow it up with the opinion-aware (OSIM) case. The following explanations of the score assignment algorithms assume the IC model of information diffusion. Their extensions to the \emph{linear threshold} (LT) and the \emph{weighted cascade} (WC) models are discussed in Sec.~\ref{subsec:extensions}.

\subsubsection{EaSyIM}
\label{subsubsec:easyim}

As mentioned in Sec.~\ref{subsec:score}, a function of the number of directed simple paths from a node $u$ to another node $v$ can be used to estimate the influence of the former on the latter. However, direct application of this notion to design solutions for the influence maximization problem poses two inherent challenges of maintaining -- 1) Scalability and Efficiency, 2) Correctness guarantees. Since, counting the number of $s$-$t$ paths is shown to be in $\#P$-Complete \cite{paths_valiant}, computing the expected influence of a node is $\#P$-hard for both the IC \cite{pmia} and the LT \cite{ldag} models. This further shows that it is impossible to come up with a score-assignment algorithm capable of correctly selecting a seed-node with maximum influence.

\setlength{\textfloatsep}{5pt}
\begin{algorithm}[t]
\caption{{\ourgreedy}}
\label{algo:seed_select}
\begin{algorithmic}[1]
{\scriptsize
\REQUIRE Graph $G = (V,E)$,\#seeds $k=|S|$, $l$
\ENSURE Seed set $S$
\STATE $S,V_{(a)} \leftarrow \emptyset$
\FOR{i = 1 to k} 
\STATE $max,\ maxId \leftarrow 0 $
\STATE $\Delta^l \leftarrow AssignScore (G(V\setminus V_{(a)},E),o,p,\varphi,l)$ 
\FOR{each $u \in V$}
\IF{$\Delta^l(u) > max$}
\STATE $max \leftarrow \Delta^l(u),\ maxId \leftarrow u$
\ENDIF
\ENDFOR
\STATE $S \leftarrow S \cup \{ maxId \}$
\STATE Update $V_{(a)}$ with nodes activated by the newly selected seed-node ($maxId$)
\ENDFOR
}
\end{algorithmic}
\afterpage{\global\setlength{\textfloatsep}{\oldtextfloatsep}}
\end{algorithm}

\setlength{\textfloatsep}{5pt}
\begin{algorithm}[t]
\caption{AssignScore}\label{algo:assignScore}
\begin{algorithmic}[1]
{\scriptsize
\REQUIRE Graph $G = (V,E)$, $o$, $p$, $\varphi$, $l$
\ENSURE $\Delta^l$
\IF{$\mathlarger{\mathlarger{\neg}}\ Opinion$}
\STATE $\Delta^l \leftarrow EaSyIM(G(V,E),p,l)$ //Opinion-Oblivious
\ELSE
\STATE $\Delta^l \leftarrow  OSIM(G(V,E),o,p,\varphi,l)$ //Opinion-Aware
\ENDIF
}
\end{algorithmic}
\afterpage{\global\setlength{\textfloatsep}{\oldtextfloatsep}}
\end{algorithm}

Although counting $s$-$t$ paths is $\#P$-Complete, there exists a polynomial time algorithm to count all possible walks of length at most $l$ between all node-pairs ($\forall(u,v) \in V$) in $G$ that takes $O(n^3\log l)$ time and consumes $O(n^2)$ memory. The $Path$-$Union$ ($PU$) algorithm (Algorithm~\ref{algo:exact_im}) extends this algorithm by initializing the adjacency matrix of the graph $M$ with the pair-wise influence probabilities (line $2$). Since the contributions of these paths cannot be simply aggregated, we define a new operator $\otimes$ for matrix multiplication (line $4$). Under this operator, the multiplication of $i^{th}$ row ($M[i][:]$) with $j^{th}$ column ($M[:][j]$) is defined using Eq.~\ref{eq:union_mult}. This algorithm has an inherent problem that it counts cyclic paths as well, which eventually serves as a source of error while computing the influence of a node. We try to reduce the impact of this error for each node, by removing the contributions of the walks that pass through the node itself (lines $5$--$7$). Finally, the score of each node $u$ is computed as indicated in line $10$.

\moveup
\moveup
\begin{align}
\label{eq:union_mult}
	M[i][j] &= M[i][:] \otimes M[:][j] = \bigcup_{k=1}^n M[i][k] \times M[k][j].
\end{align}
\moveups

Owing to its huge time and space complexity, the $PU$ algorithm cannot be used in practice. The \emph{EaSyIM} algorithm (Algorithm~\ref{algo:easy_im}) describes a score-assignment method, with efforts to tackle this scalability bottleneck. Paths of length $l$ from a node $u$ can be calculated as the sum of all paths of length $l-1$ from its neighbors. $\Delta^l(u)\ (\forall u \in V)$ is defined as the weighted sum of the number of paths of length at most $l$ starting from $u$ and is computed as indicated in line $5$ of Algorithm~\ref{algo:easy_im}. The weight for each path is defined as the product of probabilities $p_{(u,v)}$ of the edges composing that path. EaSyIM score of a node $u$ tries to mimic closely the expected value of the spread when $u$ is chosen as a seed node. 

\ignore{
\FOR{each $u \in V$}
	\STATE $\Delta^0(u) \leftarrow 0$
	\FOR{each $v \in V$}
		\STATE $M[u][v]\leftarrow p_{(u,v)}$
	\ENDFOR
\ENDFOR
}

For every iteration of \emph{EaSyIM}, the outgoing neighbours $\Out(u)$ of each node $u \in V$ are visited to accumulate the contribution of paths of different lengths ($l$). The parameter $l$ can be as large as the diameter ($\diameter$) of the graph. Thus, the total time taken by the score-assignment step using the \emph{EaSyIM} algorithm is $O(\diameter(m+n))$, in the worst-case. Moreover, the total time taken for selecting $k$ seeds using our modified greedy algorithm ({\ourgreedy}) is $O(k\diameter(m+n))$. The space complexity of this algorithm is $O(n)$, as the only additional overhead is the storage of a score ($O(1)$) at each node. A detailed analysis on the correctness of $PU$ and \emph{EaSyIM} algorithms is present in Sec.~\ref{subsubsec:easyim_analysis}.

\setlength{\textfloatsep}{4pt}
\begin{algorithm}[t]
\caption{$Path$-$Union (PU)$}\label{algo:exact_im}
\begin{algorithmic}[1]
{\scriptsize
\REQUIRE Graph $G = (V,E)$, $n=|V|$, $p$, $l$
\ENSURE $\Delta^l$
\STATE Matrix $M^{n\times n} \leftarrow 0$, $PU \leftarrow I^{n\times n}$
\STATE $\Delta^0(u) \leftarrow 0$ ($\forall u \in V$), $M[u][v]\leftarrow p_{(u,v)}$ ($\forall u,v \in V$)
\FOR{each $i \in \{ 1, \ldots ,  l\}$}
	\STATE $PU \leftarrow PU \otimes M$
	\FOR{each $v \in V$}
		\STATE $PU[v][v] \leftarrow 0$
	\ENDFOR
	\FOR{each $u \in V$}
		\FOR{each $v \in V$}
			\STATE $\Delta^i(u)\leftarrow  \Delta^{i-1}(u) + PU[u][v]$
		\ENDFOR
	\ENDFOR
\ENDFOR
\STATE return $\Delta^l$
}
\end{algorithmic}
\afterpage{\global\setlength{\textfloatsep}{\oldtextfloatsep}}
\end{algorithm}

\ignore{
\FOR{each $v \in V$}
\STATE $\Delta^0(v) \leftarrow 0$
\ENDFOR
}

\subsubsection{OSIM}
\label{subsubsec:osim}

We now describe the \emph{OSIM} algorithm (Algorithm~\ref{alg:OSIM}), to assign a score to each node of a graph for the opinion-aware case. This algorithm extends \emph{EaSyIM} by accommodating for the change in opinion (Sec.~\ref{subsec:model}) during the information propagation process. Since, the change in opinion is captured using a second layer over and above the activation step of a node, we incorporate the use of the following intermediate terms $\alpha_i$, $or_i$ and $sc_i$ for each node. At each iteration, these intermediate terms contain only the contributions for a given path length $i \leq l$. For a given node $u$, the term $or_i(u)$ contains the weighted sum of the initial opinions of all nodes, reachable via paths of length $i$ starting at $u$ (line $6$). In other words, $or_i(u)$ constitutes the contribution of nodes reachable by $u$ via paths of length $i$ to its score ($\Delta^i(u)$), without considering the change in opinion during information propagation. Similarly, the term $\alpha_i(u)$ computes the weighted sum of the interaction probabilities ($\varphi$) associated with all the paths of length $i$, starting at $u$ (line $7$). For a given path of length $i$ ($\mathcal{L}(u\leadsto v)=i$), $sc^v_i(u)$ contains the contributions of all nodes in this path to the change in opinion of node $v$ (line $8$). The term $sc_i(u)$ further contains the aggregation of $sc_i^v(u)$ for all $i$ length paths starting at $u$ (line $10$).

\setlength{\textfloatsep}{3pt}
\begin{algorithm}[t]
\caption{EaSyIM}\label{algo:easy_im}
\begin{algorithmic}[1]
{\scriptsize
\REQUIRE Graph $G = (V,E)$, $p$, $l$
\ENSURE $\Delta^l$
\STATE $\Delta^i(v) \leftarrow 0$ ($\forall i \leq l,\ \forall v \in V$)
\FOR{each $i \in \{1,\ldots,l\}$}
    \FOR{each $u \in V$}
        \FOR{each $v \in \Out(u)$}
           \STATE  $\Delta^i(u)\leftarrow \Delta^i(u) +  p_{(u,v)}(1 + \Delta^{i-1}(v))$ 
        \ENDFOR
    \ENDFOR
\ENDFOR
\STATE return $\Delta^l$
}
\end{algorithmic}
\afterpage{\global\setlength{\textfloatsep}{\oldtextfloatsep}}
\end{algorithm}

Similar to \emph{EaSyIM}, the weight for each path is defined as the product of probabilities $p_{(u,v)}$ of the edges composing that path. Moreover, the weight, for term $\alpha_i$, additionally incorporates the interaction probabilities ($\varphi_{(u,v)}$) as well. In line $11$, the score for a node $u$ ($\Delta^i(u)$) is iteratively updated with the aggregation of all the intermediate terms ($or_i(u),\alpha_i(u),$ and $sc_i(u)$). Finally, $\Delta^l(u)$ contains the score of a node $u$ with the contributions of all paths of length at most $l$, starting at $u$. Note that, since the loop invariants in lines $2,3$ and $5$ of \emph{OSIM} is exactly the same as that of \emph{EaSyIM} (lines $2,3$ and $4$), its time and space complexity analysis is exactly the same as \emph{EaSyIM}.

\ignore{
\FOR{each $v \in V$}
	\STATE $\alpha_0 (u) \leftarrow 1$, $or_0(u) \leftarrow o_u$, $sc_0(u), \Delta^0(u) \leftarrow 0$
\ENDFOR
}

\subsection{Extensions}
\label{subsec:extensions}

The extension to the WC model is rather trivial, as instead of assuming a fixed value of $p$ (usually $p=0.1$ for IC), we assign $\forall (u,v) \in E, p_{(u,v)}=1/|\In(v)|$. All the previously presented algorithms can readily accommodate this change. However, the changes required for the LT model are much more involved. 

Extension to the LT model is difficult when considering its classical definition where each node possesses a threshold, while it is quite intuitive using the \emph{live-edge} model. Kempe et al. proved the equivalence between the LT and the live-edge model in \cite{kempe}. Similar to the IC and WC models, under the live-edge model each edge $(u,v) \in E$ possesses a probability with which $u$ activates its neigbour $v$. Moreover given an instance of the graph, it has an additional constraint of each node possessing one (live) incoming edge. The expected value of spread is then calculated using the spreads achieved for each of these graph instances. Thus, by associating influence probabilites with each edge and generating various graph instances satisfying the above constraint, our algorithms get extended to the \emph{live-edge} model and hence, the LT model without undergoing any change. Since each node can possess only a single incoming edge, the error introduced in score-assignment owing to the non-disjoint nature of the paths gets completely removed (a unique $u$--$v$ path $\forall u,v \in V$), which in turn facilitates stronger theoretical analysis for the LT model when compared to the IC and WC models.

\ignore{
[ToDo extension for LT] Looking at LT model from the classical definition where each node keeps  a threshold is sliightly difficult but if we look at it from thee view point of live edge model, every edge has a parobability to get activated similar to the IC/WC model.

It has been shown by Kempe et al. that the LT model  is equivalent to the live edge model of information propagation. Also, live edge model is very similar to the IC/WC models along with a constraint that each node doesnot have more than one incoming edge. The algorithm for LT model remains the same as that of IC model. The benefit in this case is that the contribution of each path is independent and needs to be summed rather than union. So the error analysis for LT model is even stronger than the IC model.
}
%

\setlength{\textfloatsep}{2pt}
\begin{algorithm}[!t]
\caption{OSIM}\label{alg:OSIM}
\begin{algorithmic}[1]
{\scriptsize
\REQUIRE Graph $G = (V,E)$, $o$, $p$, $\varphi$, $l$
\ENSURE $\Delta^l$
\STATE $\alpha_0 (u) \leftarrow 1$, $or_0(u) \leftarrow o_u$, $sc_0(u), \Delta^0(u) \leftarrow 0$ ($\forall u \in V$)
\FOR{each $i \in \{ 1, \ldots ,  l\}$}
    \FOR{each $u \in V$}
	\STATE $ \alpha_i (v),or_i (v),\ sc_i (v) \leftarrow 0$
        \FOR{each $v \in \Out(u)$}
		\STATE $or_i(u) \leftarrow or_i(u) + p_{(u,v)}or_{i-1}(v)$
		\STATE  $\alpha_i(u) \leftarrow \alpha_i(u) +  p_{(u,v)}\alpha_{i-1}(v)(2\varphi_{(u,v)} - 1)/2 $
		\STATE $sc_{i}(u) \leftarrow sc_i(u)+ p_{(u,v)}sc_{i-1}(v) $
        \ENDFOR
	\STATE $ sc_i(u) \leftarrow sc_i(u)+  o_u\alpha_i(u) $
		\STATE $\Delta^i(u) \leftarrow \Delta^{i-1}(u) +   \frac{or_{i}(u)+sc_{i}(u)+o_u\alpha_{i}(u)}{2}$
    \ENDFOR
\ENDFOR

\STATE return $\Delta^l$
}
\end{algorithmic}
\afterpage{\global\setlength{\textfloatsep}{\oldtextfloatsep}}
\end{algorithm}

Next, we present a detailed analysis for the \emph{EaSyIM} and the \emph{OSIM} algorithms for score-assignment, assuming the IC model of information diffusion. Sec.~\ref{subsubsec:discussion} provides discussions on methods to extend these analyses to the WC and LT models as well.

\subsection{Analysis of Score Assignment}
\label{subsec:score_analysis}

As a first step towards analyzing our algorithms, we derive a relation between the expected value of spread achieved using a set of seed-nodes and any possible partition of this set.

\begin{lemma}
\label{thm:ICLT}
$\sigma(A\cup B) =\sigma(A) + \sigma_{\sigma(A)}(B)$
\end{lemma}

\begin{proof}
Using Kempe's ~\cite{kempe} analysis of the IC model it can be seen that,
\begin{align}
	\label{ic}
	 \sigma(&A\cup B) = \sum_{X} P(X) \sigma_X(A\cup B) \nonumber\\
    	&= \sum_{X} P(X) \big( \sigma_X(A) + \sigma_X(B) - \sigma_X(A\cap B) \big) \nonumber\\
	&= \sum_{X} P(X) \sigma_X(A)  + \sum_{X} P(X) \big(\sigma_X(B)-\sigma_X(A\cap B) \big) \nonumber\\
	&= \sigma(A) + \sigma_{\sigma(A)}(B).
\end{align}
A similar analysis can be done for the LT model using the reduction to the \emph{live-edge} model \cite{kempe}. Under the live-edge model each node possesses a single (live) incoming edge, thus, a node $v$ is activated via a path either from a node in $A$ or in $B$. 
\begin{align}
	\label{lt}
	\sigma(A\cup B) &= \sum_{X} P(X) \sigma_X(A\cup B) \nonumber\\
 	   &= \sum_{X} P(X) \big( \sigma_X(A) + \sigma_X(B) \big) \nonumber\\
	&= \sum_{X} P(X) \sigma_X(A)  + \sum_{X} P(X) \sigma_X(B)\nonumber
\end{align}
Since, $\sigma(A)$ and $\sigma(B)$ are disjoint,%
\begin{align}
	\sigma(A\cup B) &= \sigma(A) + \sigma_{\sigma(A)}(B).
\end{align}
\hfill{}
\end{proof}	


The next result uses Lemma~\ref{thm:ICLT} to prove that the {\ourgreedy} algorithm produces a $1-1/e$ approximate solution to the IM problem, under the condition that the score assigned to each node captures the expected value of spread using that node.

\begin{lemma}
Given a graph $G(V,E)$, if score-assignment is correct, i.e., the score assigned to each node $\Delta^l(u)$ captures the expected value of its spread $\Delta^l(u) =\sigma_{\sigma(S)}(\{u\})$ where $S$ is the set of seed nodes, then $\sigma (S_k) = \sigma(S_k^*)$ where $S_k$ is the set of seed nodes selected using the {\ourgreedy} algorithm while $S_k^*$ is the set selected using {\kempegreedy} algorithm.
\label{thm:greedy}
\end{lemma}

\begin{proof}
If $k = 1$,\ $\sigma (S_k) = \sigma(S_k^*)$ because {\ourgreedy}$(G,k,l)$ chooses the node with the maximum score. Since the score assigned to each node captures the expected value of its spread, the seed-node selected by the former is the same as that selected by the {\kempegreedy} algorithm. Let us assume that $\sigma (S_k) = \sigma(S_k^*)$ holds for $k=i$. Now for $k = i+1$, 
\begin{align*}
\sigma (S_{i+1}^*) &= \sigma(S_{i+1}^*) - \sigma(S_{i}^*) + \sigma(S_{i}^*)
\end{align*}
Paritioning $S_{i+1}$ as $S_i \cup \{u_{i+1}\}$, Lemma~\ref{thm:ICLT} yields the following,
\begin{align*}
                   &= \sigma_{\sigma(S^*_i)}(\{u_{i+1}\}) + \sigma(S_i^*)\\
                   &= \sigma_{\sigma(S_i)}(\{u_{i+1}\}) + \sigma(S_i) \\
	        &= \Delta^l(\{u_{i+1}\}) + \sigma(S_i)\\
                   &= \sigma(S_{i+1}).
\end{align*}

Since $\sigma(S_{i+1}) = \sigma(S_{i+1}^*)$, by the principle of mathematical induction $\sigma(S_k) = \sigma(S_k^*)$ holds $\forall k$.
\hfill{}
\end{proof}	

\begin{conclusion}
Given a correct score assignment algorithm, the {\ourgreedy} algorithm produces a $1-1/e$ approximate solution to the IM problem.
\label{conc:greedy}
\end{conclusion}

As discussed in Sec.~\ref{subsec:score}, it is impossible to devise a correct scoring algorithm, i.e., a score-assignment that captures the expected value of spread for each node, for graphs. Thus, the {\ourgreedy} algorithm cannot produce a $1-1/e$ approximate solution to the IM problem. Next, we show a detailed analysis of both EaSyIM and OSIM score assignment strategies.

\subsubsection{EaSyIM}
\label{subsubsec:easyim_analysis}

In this section, we first state the causes for the introduction of errors in $PU$ \& \emph{EaSyIM}, and follow it up by a concrete analysis on the exact quantification of these errors.

We first show the amount of error introduced by $PU$ for DAGs\footnote{Directed acyclic graphs would be abbreviated as DAGs in the rest of the paper.}. Although, the $PU$ algorithm can enumerate the exact number of paths for DAGs, errors might still creep in while assigning scores, owing to the presence of non-disjoint paths between any pair of nodes as shown in Figure~\ref{fig:non_disjoint_paths}. More specifically errors are introduced in the score-assignment of a node when none of the paths starting from that node to any other node in the graph are disjoint.

\begin{lemma}
\label{thm:path_union_dags}
Given a DAG, $G(V,E)$ and a set of paths, $\pathset_{uw}$, s.t. $(w,v) \in E$, if the contribution of $w$ in the score of $u$ is correct then the maximum possible relative error introduced by $v$ in the score of $u$ using the $PU$ algorithm is $\epsilon_1^{DAG}=\sum\limits_{w \in \In(v)}\big((p_{(w,v)}-1)A_1\big)$, where $A_1=\sum\limits_{\rho \in \pathset_{uw}}\prod\limits_{e \in \rho}p_e$.
\end{lemma}

\begin{proof}
The exact contribution of a node $v$ to the score of another node $u$ is
\begin{align}
\label{eq:exact_contri}
	\gamma^*_v(u)= \sum_{w \in \In(u)}\Big(p_{(w,v)}\bigcup_{\rho \in \pathset_{uw}} \prod_{e\in \rho}p_e \Big).
\end{align}
Similarly, the contribution of $v$ to the score of $u$ using the $PU$ algorithm is
\begin{align}
\label{eq:approx_contri}
	\gamma_v(u)=\bigcup_{\rho \in \pathset_{uv}} \prod_{e\in \rho}p_e.
\end{align}
Using the standard inclusion-exclusion principle,
\begin{flalign*}
\label{eq:exact_inc_exc}
	\bigcup_{\rho \in \pathset_{uv}} \prod_{e\in \rho}p_e = \bigg(& \sum_{\rho \in \pathset_{uv}}\prod_{e\in \rho}p_e - \sum_{\rho' \in \pathset_{uv}}\sum_{\rho \in \pathset_{uv}\setminus\{\rho'\}} \prod_{e\in \rho} p_e \prod_{e' \in \rho'}p_{e'}\nonumber\\
	&+ \ldots + (-1)^{|\pathset_{uv}|-1}\prod_{\rho \in \pathset_{uv}} \prod_{e \in \rho} p_e \bigg).
\end{flalign*}
Thus, Eq.~\ref{eq:exact_contri} and~\ref{eq:approx_contri} can be written as follows.
\begin{flalign*}
\gamma^*_v(u)=\sum_{w \in \In(u)}\Big(p_{(w,v)}\big(A_1 - A_2 + \ldots + (-1)^{t-1}A_{t^*} \big)\Big).	
\end{flalign*}
\begin{flalign*}
\gamma_v(u)=B_1 - B_2 + \ldots + (-1)^{t-1}B_t.	
\end{flalign*}
where $t^*=|\pathset_{uw}|$, $t=|\pathset_{uv}|$ and $A_i, i\leq t^*$, $B_j,j\leq t$ stand for the larger terms in the expansion of $\bigcup_{\rho \in \pathset_{uw}}\prod_{e\in \rho} p_e$ and $\bigcup_{\rho \in \pathset_{uv}}\prod_{e\in \rho} p_e$ respectively.\\
\\
Assuming $A_2<<A_1$ and henceforth $B_2<<B_1$, all the higher order terms can be neglected (significance of this assumption is described in Sec.~\ref{subsubsec:discussion}). Moreover, the relationship between $A_1$ and $B_1$ is defined as follows.
\begin{align*}
\sum_{w \in \In(v)} p_{(w,v)}A_1 &= \sum_{w \in \In(v)} p_{(w,v)}\big(\sum_{\rho \in \pathset{uw}}\prod_{e \in \rho} p_e\big)\\
	&=\sum_{\rho \in \pathset{uv}}\prod_{e \in \rho} p_e = B_1.
\end{align*}
Neglecting the summation over the pair-wise product of the paths between different intermediary nodes $w_i,w_j, \forall(i,j) \mid w_i, w_j \in \In(v)$, we devise a relationship between $B_2$ and $A_2$.
\begin{align*}
B_2&=\sum_{\rho \in \pathset_{uv}}\sum_{\rho'  \in (\pathset_{uv}\setminus\{\rho\})}\prod_{e \in \rho} p_e \prod_{e' \in \rho'} p_e'\\ 
&\geq \sum_{w \in \In(v)} p_{(w,v)}^2 \sum_{\rho \in \pathset_{uw}}\sum_{\rho'  \in (\pathset_{uw}\setminus\{\rho\})}\prod_{e \in \rho} p_e \prod_{e' \in \rho'} p_e'\\
&= \sum_{w \in \In(v)} p_{(w,v)}^2 A_2.
\end{align*}
The relative error $\epsilon^{DAG}_1$, is now defined as,
\begin{align}
\epsilon^{DAG}_1 &= \frac{|\gamma^*_v(u) - \gamma_v(u)|}{\gamma^*_v(u)} \approx\frac{B_2 - \sum\limits_{w \in \In(v)} p_{(w,v)}A_2}{\sum\limits_{w \in \In(v)} p_{(w,v)}(A_1-A_2)} \nonumber\\
Since&,\ A_2<<A_1,\ A_2<A_1^2\nonumber\\
	&\approx\frac{\sum\limits_{w \in \In(v)}p_{(w,v)}(p_{(w,v)}-1)A_2}{\sum\limits_{w \in \In(v)}p_{(w,v)}A_1}\nonumber\\
	&\leq\sum\limits_{w \in \In(v)}\big((p_{(w,v)}-1)A_1\big).
\end{align}
\ignore{
\begin{flalign}
	\Rightarrow \sum_{\rho in P} \prod_{e \in \rho} p_e  - \sum_{w \in In(u)}\left(p_{(w,v)}\big(   \sum_{\rho' \in P}\sum_{\rho \in P\setminus\{\rho'\}} \prod_{e\in \rho} p_e \prod_{e' \in \rho'}p_{e'} \big) \right)\\
               &\geq \sum_{\rho in P} \prod_{e \in \rho} p_e  - \sum_{w \in In(u)}\left(p_{(w,v)}\big(   \sum_{\rho' \in P}\sum_{\rho \in P\setminus\{\rho'\}} \prod_{e\in \rho} p_e \prod_{e' \in \rho'}p_{e'} \big) \right)
\end{flalign}
}
\hfill{}
\end{proof}

\ignore{
\begin{theorem}
In  a DAG,  $\sigma(S) =\sigma(S^{*}) $, for the IC model of information diffusion, where $S^*$ is the set of  seeds selected by kempe's algorithm.
\end{theorem}

\begin{proof}
We compute the probability of activation of each node $v$ on choosing $u$ as the seed and  denote it as $p^u_a(v)$
\begin{align*}
\expectation[\flow(u)] &= \sum_{v\in V}p^u_{a}(v) \nonumber\\
                     	& =   \sum_{v\in V}  \cup_{w \in \Out(u)} p_{(u,w)}p_a^w(v) \\
                        &= \sum_{v\in V}\delta_v\\
		& = score[u]
\end{align*}
Here $p_a^u(u) = 1$
Since $score[u] = \expectation[\flow(u)]$, using theorem ~\ref{thm:Algo1},$\sigma(S) = \sigma(S^*)$
\hfill{}
\end{proof}

\begin{lemma}
\label{lem:easyim}
For a node $u \in V$, the contribution of $v$, present at a distance of $d$ from $u$ is $1 - (1-s)^t$ wheere $t$ is the number of paths that start from $u$ and reach $v$ and $s$ = score of a path of length $d$.
\begin{proof}
From algo~\ref{alg:exact_im}, it can be seen that total contribution of $v$ in the expected spread of $u$ is $\cup_{i=0}^t s$
\begin{align*}
\cup_{i=0}^t s &= n s - {t \choose 2}s^2 + ...  -(-s)^{t} \\
 &= 1 -  \sum _{i=0}^t {t\choose i}(-s)^{t}\\
&= 1 - (1 - s)^t
\end{align*}
\end{proof}
\end{lemma}
It can be seen that given two nodes, $u$ and $v$. If the contribution of node $w$ in the expected spread on choosing $u$ as seed $>$ contribution of $v$  is chosen as the seed.
$**Add \delta, \delta^* to notation Add s^d_{(u,v)} and t_{(u,v)}**$
\begin{lemma}
\label{lem:relative}
If $\delta_u(w) > \delta_v(w)$, then $\delta^*_u(w) > \delta^*_v(w)$ if
\begin{proof}
If $\delta_u(w) > \delta_v(w)$, then $1 - (1-s)^t > 1 - (1-s)^t$
\end{proof}
\end{lemma}
}
\begin{lemma}
\label{thm:easyim}
Given a DAG, $G(V,E)$ and a set of paths, $\pathset_{uv}$, the maximum possible relative error (w.r.t $PU$) introduced by $v$ in the score of $u$ using the \emph{EaSyIM} algorithm is $\epsilon_2^{DAG}\leq B_1=\sum\limits_{\rho \in \pathset_{uv}}\prod\limits_{e \in \rho}p_e$.
\end{lemma}

\begin{proof}
Using the same nomenclature as Lemma~\ref{thm:path_union_dags}, the contribution of $v$ in the score of $u$ using the $PU$ and the \emph{EaSyIM} algorithm is stated as follows
\begin{align}
	\gamma_v^{PU}(u)&=\sum_{i=1}^{t} (-1)^{i-1}B_i 	\approx B_1 - B_2.\\
	\gamma_v^{EaSyIM}(u)&=B_1.
\end{align}
The relative error (w.r.t $PU$) $\epsilon^{DAG}_2$, can now be stated as follows
\begin{align}
	\epsilon^{DAG}_2=\frac{B_2}{B_1 - B_2} \approx \frac{B_2}{B_1} \leq B_1.
\end{align}
\hfill{}
\end{proof}

Using Lemma~\ref{thm:path_union_dags} and~\ref{thm:easyim} the total error (worst-case) introduced by \emph{EaSyIM} for DAGs is $\epsilon_1^{DAG}+\epsilon_2^{DAG}\leq\sum\limits_{w \in \In(v)}(2p_{(w,v)}-1)A_1$.

Moving ahead, we analyze the errors introduced by $PU$ and \emph{EaSyIM} on graphs. The error due to cycles that end up at a node $u$ itself are discounted using the $PU$ algorithm (lines $5$--$7$). The largest amount of error would be introduced by cycles of length $3$, owing to the exponential decrease of the influence probabilities with length of a path. Next, we analyze the errors introduced by other cyclic paths as shown in Figure~\ref{fig:cycle}.

\begin{lemma}
\label{thm:path_union_graphs}
Given a graph $G(V,E)$ and a set of cyclic paths $\pathset_{ww}$, s.t. $(u,w) \in E$, the maximum possible relative error, assuming just the effect of cycles, introduced by $w$ in the score of $u$ using approximate score-assignment algorithms ($PU$ and \emph{EaSyIM}) is $\epsilon^{cycle}=\sum\limits_{\rho \in \pathset_{ww}}(\prod\limits_{e \in \rho}p_e)/|\rho|$, where $|\rho| \leq \diameter$.
\end{lemma}

\begin{proof}
The contribution of a node $w$ in the score of another node $u$, considering the effect of cycles exclusively, is dependent upon all possible cyclic paths that contain $w$. Moreover, since the combined contribution of the participating nodes in each cycle $\rho$ should only be considered once to the score of $u$, we scale $\gamma_w(u)$ by $1/|\rho|$. Mathematically,
\begin{align*}
\gamma_w(u) = \sum_{\rho \in \pathset_{ww}}p_{(u,w)}(1+1/|\rho|)\prod_{e \in \rho}p_e.
\end{align*}
However, the contribution using the {\kempegreedy} algorithm is
\begin{align*}
\gamma^*_w(u) = p_{(u,w)}.
\end{align*}
Now, the relative error $\epsilon^{cycle}$, is
\begin{align}
\label{eq:cycle}
\epsilon^{cycle} = \frac{|\gamma^*_w(u) - \gamma_w(u)|}{\gamma^*_w(u)} &= \frac{\sum\limits_{\rho \in \pathset_{ww}}p_{(u,w)}(1/|\rho|)\prod_{e \in \rho}p_e}{p_{(u,w)}}\nonumber\\
		&= \sum_{\rho \in \pathset_{ww}}\Big(\prod_{e \in \rho}p_e\Big)/|\rho|.
\end{align}
\hfill{}
\end{proof}

\ignore{
\begin{figure}
\centering
\scalebox{0.6}{
\begin{tikzpicture}[node distance=15mm,
round/.style={fill=green!50!black!20,draw=green!50!black,minimum size=5mm,text width=10mm,align=center,circle,thick}]
\centering
\node at (0,-0.3) {\large $\mathbf{0}$};
\node at (1,-.3) {\large $\mathbf{\gamma_w^*(v)}$};
\node at (3.7,-.3) {\large $\mathbf{\gamma_w^*(u)}$};
\node at (2.7,-1.4) {\large $\mathbf{\epsilon_2}$};
\node at (4.5,-.5) {\large $\mathbf{\epsilon_1}$};
\node at (4.7,-1.4) {\large $\mathbf{\gamma_w(v)}$};
\node at (5.5,-.4) {\large $\mathbf{\gamma_w(u)}$};

\draw[-latex](0,0) to(5.6,0);
\draw[-](0,-.1) to(0,.2);
\draw[-](1,-.1) to(1,.2);
\draw[-](1,-1.1) to(1,-.8);
\draw[-](1,-1) to(4.5,-1);
\draw[-](3.7,-.1) to(3.7,.2);
\draw[-](4.5,-1.1) to(4.5,-.8);
\draw[-](3.7,-.8) to(3.7,-.5);
\draw[-](3.7,-.7) to(5,-.7);
\draw[-](5,-.8) to(5,-.5);

\end{tikzpicture}
}
\figcaption{Error in score}
\label{fig:easyim_relative}
\end{figure}
}
\begin{theorem}
\label{thm:easyim_relative}
Given the expected spread of $u$ is greater when compared to the spread of $v$ using the {\kempegreedy} algorithm i.e., $\sigma^*(u) > \sigma^*(v)$, the corresponding expected spreads achieved using seeds selected by any approximate score-assignment algorithm preserve this relationship i.e., $\sigma(u) > \sigma(v)$, if $\frac{\epsilon_v}{\sigma^*(v)} - \frac{\epsilon_u}{\sigma^*(u)} \leq \frac{\sigma^*(u)-\sigma^*(v)}{\sigma^*(v)}$, where $\epsilon_u=\sigma(u) - \sigma^*(u)$ and $\epsilon_v=\sigma(v) - \sigma^*(v)$.
\end{theorem}

\begin{proof}
Let $\epsilon_u$ and $\epsilon_v$ be the errors introduced in $u$ and $v$ respectively  using any approximate score-assignment algorithm. Given $\sigma^*(u) > \sigma^*(v)$, for $\sigma(u)$ and $\sigma(v)$ to violate this relationship the relative error introduced in $v$ ($\frac{\epsilon_v}{\sigma^*(v)}$) should exceed the relative error introduced in $u$ ($\frac{\epsilon_u}{\sigma^*(u)}$) by the amount of gap between $u$ and $v$ relative to $v$. Figure~\ref{fig:easyim_relative} portrays this observation. Mathematically,
\begin{align}
\label{eq:easyim_relative}
&\frac{\epsilon_v}{\sigma^*(v)} - \frac{\epsilon_u}{\sigma^*(u)} \leq \frac{\sigma^*(u) - \sigma^*(v)}{\sigma^*(v)}.\ Since,\ \sigma^*(u)>\sigma^*(v)\nonumber\\
&\Rightarrow \epsilon_v - \epsilon_u \leq \sigma^*(u) - \sigma^*(v) \Rightarrow \sigma(v) \leq \sigma(u).
\end{align}
\ignore{
The first part of the proof is as follows: \textbf{Fight}
\begin{align}
\gamma_w(u) > \gamma_w(v) &\Rightarrow \gamma^*_w(u) + \epsilon_u > \gamma^*_w(v) + \epsilon_v\nonumber\\
		&=\gamma^*_w(u) - \gamma^*_w(v) > \epsilon_v - \epsilon_u\nonumber\\
		&=\frac{\gamma^*_w(u) - \gamma^*_w(v)}{\gamma^*_w(v)} > \frac{\epsilon_v}{\gamma^*_w(v)} - \frac{\epsilon_u}{\gamma^*_w(v)}
\end{align}
Now we cannot move further.
}

Thus, if the above condition holds then the errors introduced by any approximate score-assignment algorithm preserves the relative ordering between $\sigma^*$ and $\sigma$.
\hfill{}
\end{proof}

\subsubsection{Discussion}
\label{subsubsec:discussion}

Eq.~\ref{eq:easyim_relative} (theorem~\ref{thm:easyim_relative}) presents the condition for any approximate score-assignment algorithm to maintain the same relative ordering of nodes as produced by the greedy-gold standard\footnote{{\kempegreedy} ranks the nodes in the order of their expected influence.}. This condition evaluates to $\epsilon_1 - \epsilon_2 \leq \frac{\sigma^*(u)}{\sigma^*(v)}-1 $, where $\epsilon_1 = \frac{\epsilon_v}{\sigma^*(v)}$ and $\epsilon_2=\frac{\epsilon_u}{\sigma^*(u)}$ are the relative errors incurred in the scores of $v$ and $u$ respectively. This condition could further be simplified, by putting $\epsilon_1 - \epsilon_2 = \epsilon$, to $\sigma^*(v) \leq \frac{\sigma^*(u)}{1+\epsilon} $. Let us now analyze the above mentioned condition and identify the cases where it gets violated. Firstly, when $\epsilon <<1$, i.e., $1 + \epsilon \approx 1$, the condition always holds as it is given that $\sigma^*(u) > \sigma^*(v)$. Now, when $\epsilon$ is large, we can rewrite the condition as $\epsilon \leq \frac{\sigma^*(u)}{\sigma^*(v)} -1$. A large $\epsilon$ means that the difference of relative errors in the two chosen nodes is high. Given this, if $\sigma^*(v)$ and $\sigma^*(u)$ are close then choosing any one of $u$ and $v$ as a seed-node would not impact the expected spread much, thus, producing a correct ordering is no longer important. On the other hand, when $\sigma^*(v)$ and $\sigma^*(u)$ are far apart, we can safely assume $\sigma^*(v)$ to be equal to $\sigma^*(u)$ in the worst case\footnote{The chance of any score-assignment algorithm to distort the relative ordering by introducing errors would be the highest when $\sigma^*(u)$=$\sigma^*(v)$.}. Finally, the only case for which the condition in Eq.~\ref{eq:easyim_relative} gets violated is when $\epsilon > 0$. To better explain the previous result, we provide an in-depth analysis of lemma's~\ref{thm:path_union_dags},~\ref{thm:easyim}, and~\ref{thm:path_union_graphs} to get an expression for the relative error. Please note that, unlike above, all the analyses in the following discussion are based on the score of a node $v$ w.r.t contribution of another node $w$.

Using lemma~\ref{thm:path_union_dags} and~\ref{thm:easyim}, it is evident that the relative error, for DAGs, introduced by $v$ in the score of $u$ using the \emph{EaSyIM} algorithm is $\leq\sum_{w \in \In(v)}(2p{(w,v)}-1)A_1$. To quantify this error, we evaluate $A_1$ under different information-diffusion models. Considering the average degree of the underlying graph to be $\eta$, the expected value of the total number of $u$-$v$ paths of a given length ($l$) can be upper-bounded by $\eta^{l-1}$. Under the IC model, the contribution of a path of length $l$ is $\prod p_e = p^l$ where $p = 0.1$ which simplifies $A_1$ to $\sum_{i=2}^{l-1}\eta^{i-1}p^i$. It can be seen that this error is not as large as it looks, since the number of paths usually do not increase exponentially (w.r.t $\eta$) with the increase in length. Let us consider a graph with small value of $\eta$, i.e., $\eta p < 1$ and analyze the growth of $A_1$ with the addition of paths of increasing length. With the addition of paths of a given length $i$, $A_1$ increases by $(\eta p)^{i-1}p$, which in turn decreases exponentially since $\eta p <1$. Thus, we can observe that $A_1$ grows sub-logarithmically. On the contrary, when $\eta$ is large it dominates the growth of $A_1$ and thus, the error increases with the increase in path-length. However, the path-length is upper-bounded by the diameter $\diameter$ of the graph, which is usually small for large $\eta$ \cite{diameter_degree}. Hence, the overall error introduced by the \emph{EasyIM} algorithm is upper-bounded by $\eta(2p-1)A_1$, which is small in case of DAGs. Similarly, lemma~\ref{thm:path_union_graphs} computes the error introduced by $v$ in the score of $u$ due to cycles by the $PU$ algorithm. Eq.~\ref{eq:cycle} can be re-written as $\sum_{i=2}^l(\eta^{i-1}p^i)/i$, which is similar to the error term evaluated for DAGs. Using similar analysis as above, this term is also small. Moreover, \emph{EaSyIM} incurs an additional error owing to the cycles of varying lengths that contain $u$. It is evident that the error due to these cycles is similar to that of the $PU$ algorithm. Thus, the total error introduced by \emph{EaSyIM} is not large for graphs as well.

\begin{figure}[t]
	\centering
	\subfloat[Non-Disjoint Paths]
	{
		\scalebox{0.27}{
			\begin{tikzpicture}[node distance=15mm,
			round/.style={fill=green!50!black!20,draw=green!50!black,minimum size=5mm,text width=10mm,align=center,circle,thick}]
			\centering

			\node[round,scale=1.2] (u) {\huge $\mathbf{u}$};
			\node[round,scale=1.2] (w) [below = 1.5cm of u]{\huge $\mathbf{w_i}$};
			\node[round,scale=1.2] (v) [below = 0.6cm of w]{\huge $\mathbf{v}$};
			\node[round, scale=1.2] (w1) [left = 2.5cm of w]{\huge $\mathbf{w_1}$};
			\node[round, scale=1.2] (w2) [right = 2.5cm of w]{\huge $\mathbf{w_{|\In(v)|}}$};

			\draw [snake arrow](u) to [ bend right=15] (w);
		
			\draw [snake arrow](u) to [ bend left=15] (w);

			\draw [snake arrow](u) to [ bend right=5] (w2);		
		
			\draw [snake arrow](u) to [ bend left=25] (w2);		

			\draw [snake arrow](u) to [ bend right=25] (w1);		

			\draw [snake arrow](u) to [ bend left=5] (w1);		



			\draw[loosely dotted,thick, line width=2.5pt](-.2,-1.4) -- (.2,-1.4);
			\draw[loosely dotted,thick, line width=2.5pt](-2.2,-3.2) -- (-1.6,-3.2);
			\draw[loosely dotted,thick, line width=2.5pt](1.8,-3.2) -- (2.4,-3.2);

			\draw[-latex,-triangle 45, line width=1.4pt](w1) to (v);
			\draw[-latex,-triangle 45, line width=1.4pt](w) to (v);
			\draw[-latex,-triangle 45, line width=1.4pt](w2) to (v); 			
	  
			\end{tikzpicture}
		}
		\label{fig:non_disjoint_paths}
	}
	\subfloat[Cycles]
	{
		\centering
		\scalebox{0.35}{
			\begin{tikzpicture}[node distance=15mm,
			round/.style={fill=green!50!black!20,draw=green!50!black,minimum size=5mm,text width=10mm,align=center,circle,thick}]
			\centering
			\node at (1.4,-0.8) {\huge $\mathbf{p_{(u,w)}}$};
			\node at (4.0,-0.8) {\huge $\mathbf{p_{(w,v)}}$};
			\node at (3.9,1.5) {\huge $\mathbf{p_{(v,w)}}$};

			\node[round] (u) {\huge $\mathbf{u}$};
			\node[round] (w) [right = 1.25cm of u]{\huge $\mathbf{w}$};
			\node[round] (v) [right = 1.25cm of w]{\huge $\mathbf{v}$};

			\draw[-latex,-triangle 45, line width=1.4pt](u) to (w);
			\draw[-latex,-triangle 45,line width=1.4pt](w) to (v);
			\draw[-latex,-triangle 45,line width=1.4pt](v) to[bend right = 60] (w);
				  
			\end{tikzpicture}
		}
		\label{fig:cycle}
	}
	\subfloat[Relative Error]
	{
		\scalebox{0.4}{
		\begin{tikzpicture}[node distance=15mm,
		round/.style={fill=green!50!black!20,draw=green!50!black,minimum size=5mm,text width=10mm,align=center,circle,thick}]
		\centering
		\node at (0,-0.2) {\huge $\mathbf{0}$};
		\node at (1,1.2) {\huge $\mathbf{\gamma_w^*(v)}$};
		\node at (3.7,1.2) {\huge $\mathbf{\gamma{\huge_w^*(u)}}$};
		\node at (2.7,-1.4) {\huge $\mathbf{\epsilon_2}$};
		\node at (4.4,-.2) {\huge $\mathbf{\epsilon_1}$};
		\node at (4.7,-1.4) {\huge $\mathbf{\gamma_w(v)}$};
		\node at (5.85,-.15) {\huge $\mathbf{\gamma_w(u)}$};

		\draw[-latex,-triangle 45, line width=1.4pt](0,.5) to(5.6,.5);
		\draw[-,line width=1.4pt](0,.2) to(0,.8);
		\draw[-,line width=1.4pt](1,.2) to(1,.8);
		\draw[-,line width=1.4pt](1,-1.2) to(1,-.8);
		\draw[-,line width=1.4pt](1,-1) to(4.5,-1);
		\draw[-,line width=1.4pt](3.7,.2) to(3.7,.8);
		\draw[-,line width=1.4pt](4.5,-1.2) to(4.5,-.8);
		\draw[-,line width=1.4pt](3.7,-.7) to(3.7,-.3);
		\draw[-,line width=1.4pt](3.7,-.5) to(5,-.5);
		\draw[-,line width=1.4pt](5,-.7) to(5,-.3);	  

		\end{tikzpicture}
		}
		\label{fig:easyim_relative}
	}
	\figcaption{\textbf{Error analysis for the score-assignment step.}}
\end{figure}
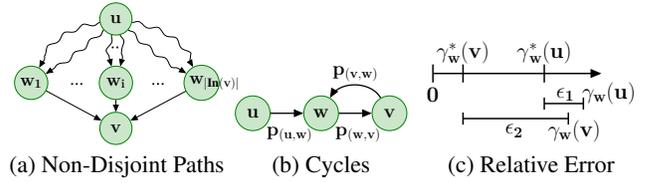

Under the WC model, we can safely assume that $\eta p=1$ as the probability associated with each edge $e=(u,v)$ is $1/$in-degree$(v)$$\approx1/\eta$. The analysis for the IC model holds for this case as well and thus, the error is small. Our results are even stronger for the LT model, as \emph{EaSyIM} always maintains the same ordering of nodes as produced by the {\kempegreedy} algorithm for DAGs. This clearly follows from the live-edge model, where each node possesses at most one incoming edge and thus, in our analysis all the $\cup$'s can be replaced by $\sum$'s which in turn eradicates all the sources of errors for DAGs (and other simpler structures as well). One can also observe that using a similar analysis, the error introduced by \emph{EaSyIM} for tree-like structures is $0$ for the IC and WC models as well. Moreover, the errors introduced for graphs are due to the cycles alone, and following the similar analysis as above (edge probabilities are $1/$in-degree), the overall error remains small.

Now we revisit the expression $\epsilon >0$. Since $\epsilon_1$ is not close to $\epsilon_2$, we can conclude, owing to the above analysis that the error increases significantly with length only if the number of paths increase exponentially, that there exist a large number of longer paths contributing to $\epsilon_1$. Let there be $t_1$ paths of length $l_1$ which contribute towards $\epsilon_1$ and $t_2$ paths of  length $l_2$ which contribute towards $\epsilon_2$. Assuming that the other set of paths introduce an equal amount of error in both $u$ and $v$, the expression $\epsilon >0$ reduces to $t_1 p^{l_1} - t_2p^{l_2} > 0$ for the IC model. Thus, the condition in Eq.~\ref{eq:easyim_relative} gets violated if the length of the paths contributing to $v$, i.e., $l_2<l_1<l_2-\log_p(\frac{t_1}{t_2})$. Considering an alternate analysis, we can see that if there are $t$ paths of length $l$, with an individual contribution of $p^l$, ending at a node $v$, then the correct score would be $\cup p^l = 1 - (1-p^l)^t$. Since $p^l<<1$ this expression reduces to $\approx 1 - (1-tp^l) = tp^l$. This is the same as the score assigned using the \emph{EaSyIM} algorithm.

In summary, we identify the conditions where the error introduced by \emph{EaSyIM} does not distort the relative ordering of nodes, as produced by the greedy gold standard. For the cases, where this relative ordering cannot be maintained we show that the error introduced is small, and thus the overall impact to the spread of a given set of seed-nodes is also small. Before analyzing \emph{OSIM}, we state the conclusions drawn on the results presented in this section.

\begin{conclusion}
Score-assignment using the \emph{EaSyIM} algorithm captures the expected spread of a node for tree-like structures, under the IC, WC and LT models of information diffusion. Thus, the {\ourgreedy} algorithm produces a $1-1/e$ approximate solution to the IM problem for tree-like structures.
\end{conclusion}

\begin{conclusion}
Score-assignment using the \emph{EaSyIM} algorithm captures the expected spread of a node for DAGs under the LT model. Thus, the {\ourgreedy} algorithm produces a $1-1/e$ approximate solution to the IM problem for DAGs under the LT model.
\end{conclusion}






\subsubsection{OSIM}

Since the analysis over generic graphs is quite complex, we analyze \emph{OSIM} on DAGs with fixed value of $\varphi_{(u,v)}$. Moreover, since \emph{OSIM} is fundamentally similar to EaSyIM and the analysis involving opinions is complex, thus, we reduce (without introducing any limitation in the algorithm) the analysis effort with results on a single path and not for all possible paths from a node. These results can further be easily extended.

Both the below mentioned results talk about the expected spread of a single node. Since the overall function is neither submodular nor monotone, we cannot employ the use of a greedy algorithm to achieve a constant approximation. (Proof in Sec.~\ref{sec:problem}).

The following notation would assume these definitions for the rest of this section. $\delta_j(\{A\})$ is the \emph{Dirac}\footnote{The Dirac measure $\delta_j(\{A\})$ is defined for any set $A$, where $\delta_j(\{A\})=1$ if $j \in A$ and $0$ otherwise.} measure, $\psi_{i} = \frac{2\varphi_{(u_{i,i+1})}-1}{2}$ and $\thatsymbol_{j} = p_{(u_j,u_{j+1})}$.



\ignore{
\begin{align*}
\sigma_l^o = \sum_{i=1}^{l} \bigg(\Big( \prod_{j=1}^i \thatsymbol_{j-1} \Big) \Big( \sum_{j=0}^{i} \frac{o_{u_j}}{2}(1+\delta_j{\{0\}})\prod_{k=1}^{i-j} \frac{(2\psi_{i-k}-1)}{2}  \Big) \bigg)
\end{align*}
}

\begin{lemma}
\label{lem:osim1}
Given a path of length $l$, consisting of nodes $u_0, u_1,$ $\ldots,u_l$, with the penalty parameter on opinion spread, $\lambda = 1$, the effective opinion spread under the $OI$ model of information-diffusion, considering this path, on choosing $u_0$ as the seed node is $\sigma^o(\{u_0\}) = \sum\limits_{i=1}^{l} \bigg(\Big( \prod\limits_{j=1}^i \thatsymbol_{j-1} \Big) \sum\limits_{j=0}^{i} \Big(\frac{o_{u_j}}{2}\big(1+\delta_j(\{0\})\big)\prod\limits_{k=1}^{i-j} \psi_{i-k} \Big) \bigg)$.
\end{lemma}

\begin{proof}
It is evident that the contribution of a path of length $l$ in the expected opinion spread of a seed node $u_0$, is the sum of expected effective opinions of nodes in that path. Now we try to solve the following recursive relation to get a closed expression for the expected opinion of the nodes.
\begin{align*}
o_{u_i}' &= \frac{o_{u_i}}{2} + \psi_{i-1} o_{u_{i-1}}'\\
o_{u_{i-1}}' &= \frac{o_{u_{i-1}}}{2} + \psi_{i-2} o_{u_{i-2}}'\\
&.\\
&.\\
&.\\ 
o_{u_1}' & = \frac{o_{u_1}}{2} + \psi_{0} o_{u_{0}}'
\end{align*}
On solving the telescopic sum and substituting $o_{u_0}' = o_{u_0}$
\begin{align*}
o_{u_i}' &= \sum_{j=1}^{i}\Big( \frac{ o_{u_j}}{2} \prod_{k=1}^{i-j} \psi_{i-k}\Big) + o_{u_0} \prod_{k=1}^{i} \psi_{i-k}\\
&=  \sum_{j=0}^{i}\Big( \frac{ (1+\delta_j(\{0\}))  o_{u_j}}{2} \prod_{k=1}^{i-j} \psi_{i-k}\Big).
\end{align*}
\begin{align*}
\sigma^o(\{u_0\}) &= \sum_{i=1}^{l} \bigg(o_{u_i}'\prod_{j=1}^i \thatsymbol_{j-1}\bigg)\\
&= \sum\limits_{i=1}^{l} \bigg(\Big( \prod\limits_{j=1}^i \thatsymbol_{j-1} \Big) \sum\limits_{j=0}^{i} \Big(\frac{o_{u_j}}{2}\big(1+\delta_j(\{0\})\big)\prod\limits_{k=1}^{i-j} \psi_{i-k} \Big) \bigg).
\end{align*}

\ignore{
\begin{align*}
OIspread &= \sum_{i=1}^{d} \left(o_{u_i}'\prod_{j=1}^i p_{(u_{j-1},u_j)}\right)\\
& = \sum_{i=1}^{d} \left(\left( \prod_{j=1}^i p_{(u_{j-1},u_j)} \right) \left( \sum_{j=0}^{i} \frac{(1+\delta_j(\{0\}))  o_{u_j}}{2}\prod_{k=1}^{i-j} \frac{(2\varphi_{(u_{i-k,i-k+1})}-1)}{2}   \right) \right)
\end{align*}
}
\hfill{}
\end{proof}

\begin{lemma}
\label{lem:osim2}
Given a path of length $l$ and $\lambda = 1$, the score assigned to a node $u$ using the \emph{OSIM} algorithm captures the effective opinion spread, considering the contribution of this path, on choosing $u$ as the seed node. Mathematically, $\Delta^l(u) = \sigma^o(\{u\})$.
\end{lemma}

\begin{proof}
In order to prove this, we first obtain an expression of the score assigned to a node as a function of $u_i,\ i\in \{0,1,\ldots, l\}$. From the line $6$ of Algorithm~\ref{alg:OSIM}, we evaluate the expression for $or_i(u_k)$ as follows:
$or_{i}(u_k) = or_{i-1}(u_{k+1}) p_{(u_k,u_{k+1})} = or_{i-1}(u_{k+1}) \thatsymbol_{k} $.
This expression can be written recursively as follows:

\begin{align}
\label{eq:or_recursive}
or_{i}(u_k) &= or_{i-1}(u_{k+1}) \thatsymbol_{k}\nonumber\\
or_{i-1}(u_{k+1}) &= or_{i-2}(u_{k+2}) \thatsymbol_{k+1}\nonumber\\
&.\nonumber\\
&.\nonumber\\
&.\nonumber\\ 
or_{1}(u_{k+i-1}) &= or_{0}(u_{k+i} ) \thatsymbol_{k+i-1}
\end{align}
On solving the telescopic sum and putting $or_0(u_k) = o_{u_k}$, we get
\begin{align*}
or_i(u_k) & = o_{u_{k+i}}\prod_{j=0}^{i-1}\thatsymbol_{k+j}.
\end{align*}
We evaluate $\alpha$ using the line $7$ of Algorithm~\ref{alg:OSIM}, and obtain a similar set of recursive expressions as in Eq.~\ref{eq:or_recursive}.
\ignore{
\begin{align*}
\alpha_{i}(u_k) &= \alpha_{i-1}(u_{k+1}) \thatsymbol_{k}\psi_k\\
\alpha_{i-1}(u_{k+1}) &= \alpha_{i-2}(u_{k+2}) \thatsymbol_{k+1}\psi_{k+1}\\
&.\\
&.\\
&.\\ 
\alpha_{1}(u_{k+i-1}) &= \alpha_{0}(u_{k+i} ) \thatsymbol_{k+i-1}\psi_{k+i-1}
\end{align*}
}
On solving the telescopic sum and putting $\alpha_0(u_k) = 1$, we obtain the following:
\begin{align*}
\alpha_i(u_k) & = \prod_{j=0}^{i-1}\thatsymbol_{k+j}\psi_{k+j}.
\end{align*}
Now using line $10$ of Algorithm~\ref{alg:OSIM}, we obtain an expression of $sc_{i}(u_k)$ defined as follows:
\begin{align*}
sc_{i}(u_k) &= sc_{i-1}(u_{k+1}) \thatsymbol_{k} + o_{u_k} \alpha_i(u_k).
\end{align*}
Putting the evaluated value of $\alpha$ in the above expression, we again obtain a similar set of recursive expressions as in Eq.~\ref{eq:or_recursive}.
\ignore{
\begin{align*}
sc_{i}(u_k) =& sc_{i-1}(u_{k+1}) \thatsymbol_{k} + o_{u_k} \prod_{j=0}^{i-1}\thatsymbol_{k+j}\psi_{k+j} \\
sc_{i-1}(u_{k+1}) =& sc_{i-2}(u_{k+2}) \thatsymbol_{k+1} + o_{u_{k+1}} \prod_{j=0}^{i-2}\thatsymbol_{k+j+1}\psi_{k+j+1} \\
&.\\
&.\\
&.\\ 
sc_{1}(u_{k+i-1}) =& o_{u_{k+i-1}}\thatsymbol_{k+i-1} \psi_{k+i-1}
\end{align*}
\vspace{0.5mm}
}
On solving the telescopic sum, we get:
\begin{align*}
sc_i(u_k) & = \sum_{j=0}^{i-1}\bigg (o_{u_{k+j}}\big(\prod_{r=0}^{j-1}\thatsymbol_{k+r}\big) \bigg)\bigg( \prod_{s=0}^{i-j-1} \thatsymbol_{k+s+j}(\psi_{k+s+j})\bigg)\\
& = \prod_{j=k}^{i+k-1}  \left( \sum_{j=0}^{i-1}o_{u_{k+j}} \prod_{s=0}^{i-j-1} (\psi_{k+s+j})\right).
\end{align*}
Now we evaluate $sc_i(u_0) + or_i(u_0) + o_{u_0}\alpha_{i}(u_0)$ as:
\begin{align*}
& \prod_{j=0}^{i-1} \thatsymbol_{j} \bigg( \sum_{j=0}^{i-1}o_{u_{j}} \prod_{s=0}^{i-j-1} (\psi_{s+j})\bigg) + \prod_{j=0}^{i-1} \thatsymbol_{j} o_{u_i} + o_{u_0} \prod_{j=0}^{i-1}\thatsymbol_{j}\psi_j\\
& = \bigg( \prod_{j=1}^i\thatsymbol_{j-1}\bigg) \bigg(    \Big( \sum_{j=0}^{i}o_{u_{j}} \prod_{s=0}^{i-j-1} (\psi_{s+j})\Big) + o_{u_0} \prod_{j=0}^{i-1}\thatsymbol_{j}\psi_j  \bigg).
\end{align*}

From line $11$ of Algorithm~\ref{alg:OSIM}, it is evident that $\Delta^l(u_0) = \sum\limits_{i=0}^l\bigg(\frac{or_i(u_0)}{2}+ \frac{sc_i(u_0)}{2} + o_{u_0} \frac{\alpha_i(u_0)}{2}\bigg)$ which means:
\begin{align*}
\Delta^l(u_0) & = \sum_{i=1}^{l} \bigg(\Big( \prod_{j=1}^i \thatsymbol_{j-1} \Big) \Big( \sum_{j=1}^{i} \frac{o_{u_j}}{2}\prod_{k=1}^{i-j} \psi_{i-k} + o_{u_0} \prod_{k=1}^{i} \psi_{i-k} \Big) \bigg) \\
& = \sigma^o(\{u_0\}).
\end{align*}
\hfill{}
\end{proof}



\section{Experiments}
\label{sec:exp}
All the simulations were done using the Boost graph library \cite{boost} in C++ on an Intel(R) Xeon(R) 20-core machine with 2.4 GHz CPU and 100 GB RAM running Linux Ubuntu 12.04. We present results on real (large) graphs, taken from the arXiv \cite{arxiv} and SNAP \cite{snap} repositories, as described in Table~\ref{tab:dataset}. In addition to these, a snapshot of the \emph{Twitter} network crawled from Twitter.com in July 2009 was downloaded from a publicly available source \cite{twitter}. We consider a mix of both directed and undirected graphs, however to ensure uniformity the undirected graphs were made directed by considering, for each edge, the arcs in both the directions. As a conventional practice, the \emph{spread}, and hence the \emph{opinion-spread} as well, is calculated as an average over $10K$ Monte Carlo (MC) simulations\footnote{These instances were run in parallel on 20-cores for \emph{OSIM}, \emph{EaSyIM} and Modified-{\kempegreedy}. However, for a fair comparison with other techniques we report the total time taken.}. Next, we state the \emph{information-diffusion} models and the \emph{algorithms} used for obtaining the results described in the rest of this section.

{\bf Information-Diffusion Models}: We incorporate the use of three fundamental diffusion models, namely -- IC, WC and LT \cite{kempe}, for the \emph{opinion-oblivious} case. The IC model is used with a uniform probability $p_{(u,v)} = 0.1$, assigned to all the edges ($\forall (u,v) \in E$) of the graph $G$. On the other hand, the WC model associates a probability of $p_{(u,v)} = \frac{1}{|\In(v)|}$ with all the edges, where $|\In(v)|$ denotes the in-degree of $v$. As opposed to the IC and WC models, the LT model requires an additional parameter, apart from the weights $w_{(u,v)} = \frac{1}{|\In(v)|}$ associated with all the edges, on all the nodes of the graph specifying the node threshold $\theta_v = rand(0,1),\ \forall v \in V$. Note that the $rand(0,1)$ function generates numbers randomly and uniformly between $0$ and $1$. Moreover, all these parameter assignments follow the conventional practice in the literature \cite{kempe, celf, celfPlus, tim}. For the \emph{opinion-aware} case, we employ the use of two diffusion models, namely -- OC \cite{ovm} and OI (Sec.~\ref{subsec:model}). This choice is driven by our thorough analytical comparisons in Sec.~\ref{sec:intro} and later in Sec.~\ref{sec:related} with other related models. The details about dataset preparation using the opinion-aware model parameters -- \emph{opinion} ($o$) and \emph{interaction} ($\varphi$) are present in Sec.~\ref{subsec:exp_opinion}.

{\bf Algorithms}: We portray the effectiveness of \emph{OSIM} by comparing with the Modified-{\kempegreedy} (Appendix~\ref{app_alg:mod_greedy}) algorithm which, being tuned to maximize the \emph{effective opinion spread} (Def.~\ref{def:eop_spread}), serves as a baseline to evaluate the \emph{quality of spread} for the \emph{opinion-aware} scenario. We compare \emph{EaSyIM} for effectiveness, efficiency and scalability with a representative set of state-of-the-art algorithms and heuristics, as established by us in Secs.~\ref{sec:intro} and~\ref{sec:related}, namely -- CELF++, TIM$^+$, IRIE and SIMPATH. For all these algorithms, we adopt the C++ implementations made available by their authors.

{\bf Parameters}: The value of spread reported by the Modified-{\kempegreedy}, CELF++ \cite{celfPlus}, \emph{OSIM} and \emph{EaSyIM} algorithms is an average over $10K$ MC simulations. Thus, we report the total time taken to complete all the MC simulations for these algorithms. Unless otherwise stated, we set $l=3$ and $\lambda=1$ for the score-assignment step of \emph{OSIM} and \emph{EaSyIM}. For the TIM$^+$ \cite{tim} algorithm we use $\epsilon=0.1$. We set the IRIE parameters $\alpha$ and $\theta$ to $0.7$ and $1/320$, respectively, and SIMPATH's parameters $\eta$ and $l$ to $10^{-3}$ and $4$, respectively. All the parameters have been set according to the recommendations by the authors of \cite{tim, simpath, irie}.

\begin{figure*}[t]
\centering
	\subfloat[Twitter]
	{
		\scalebox{0.26}{
			\includegraphics[width = 0.99\linewidth]{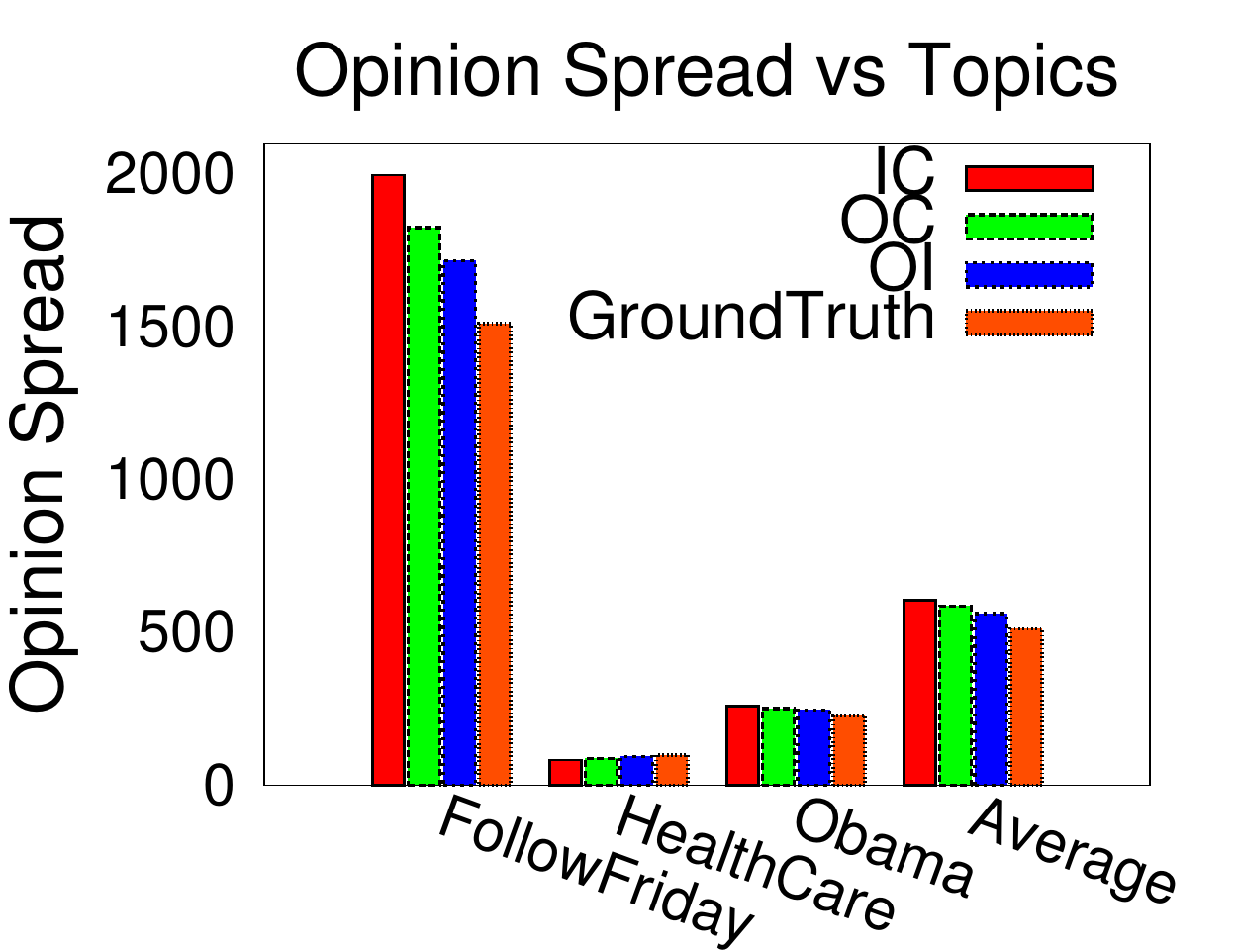}
		}
		\label{fig:twitter_eop_flow}
	}
	\subfloat[Twitter]
	{
		\scalebox{0.23}{
			\includegraphics[width = 0.99\linewidth]{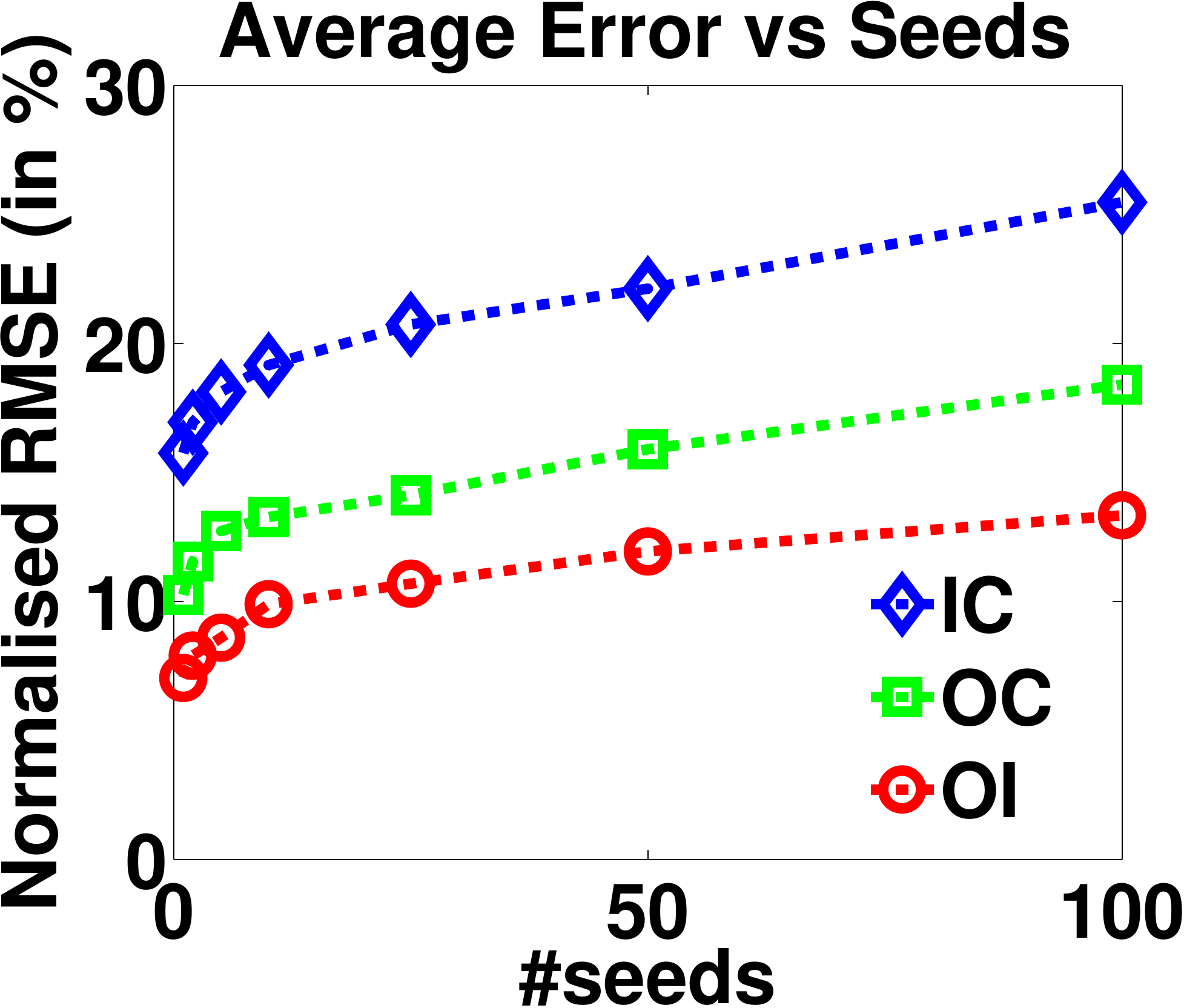}
		}
		\label{fig:twitter_eop_flow_seeds}
	}
	\subfloat[Twitter]
	{
		\scalebox{0.23}{
			\includegraphics[width = 0.99\linewidth]{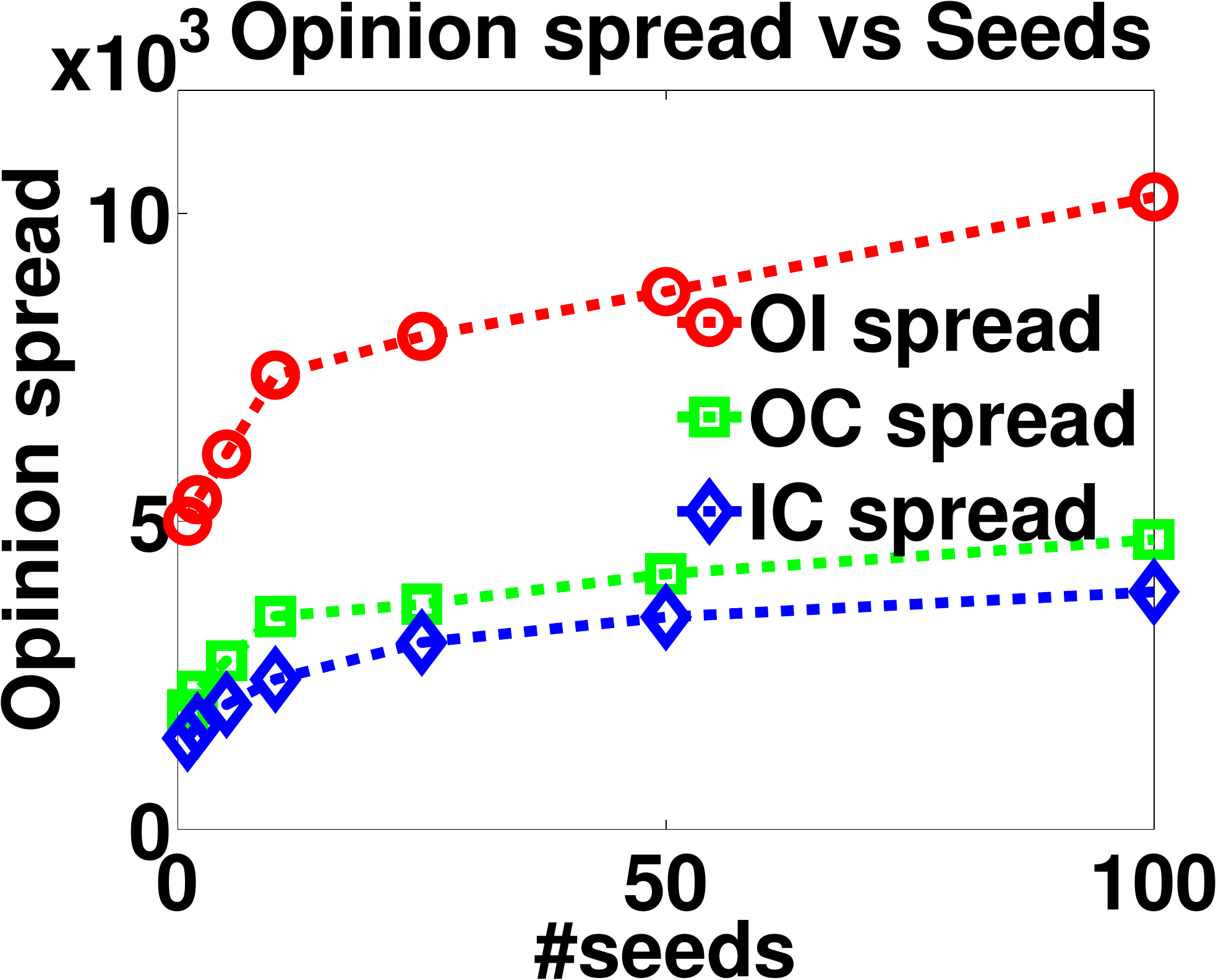}
		}
		\label{fig:twitter_eop_spread}
	}
	\subfloat[PAKDD]
	{
		\scalebox{0.23}{
			\includegraphics[width = 0.99\linewidth]{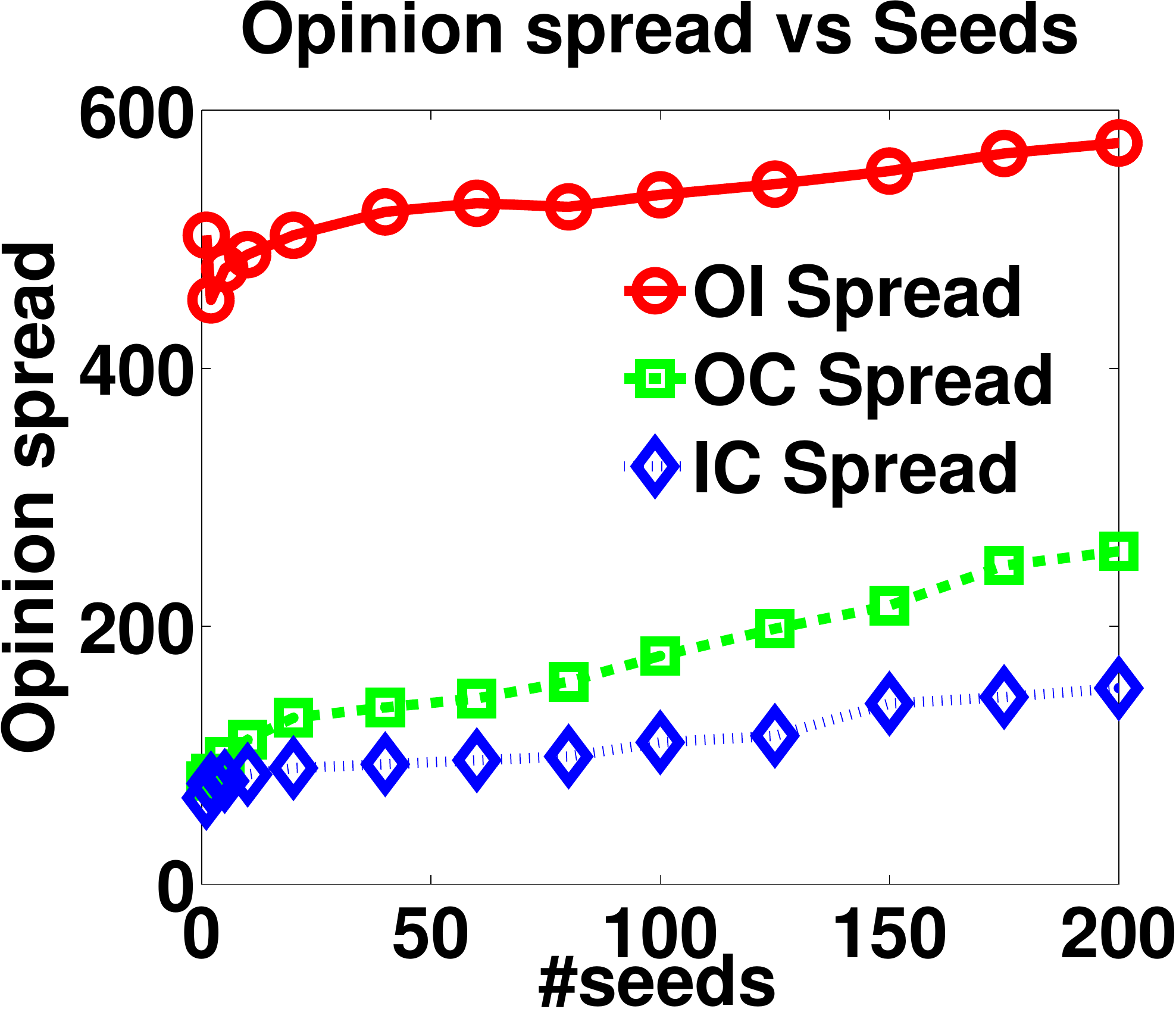}
		}
		\label{fig:real_world_eop_spread}
	}
	\\
	\vspace{-3.5mm}	
	\subfloat[NetHEPT and HepPh]
	{
		\scalebox{0.23}{
			\includegraphics[width = 0.99\linewidth]{motivate_objective}
		}
		\label{fig:nethept_hepph_lambda}
	}
	\subfloat[NetHEPT]
	{
		\scalebox{0.22}{
			\includegraphics[width = 0.99\linewidth]{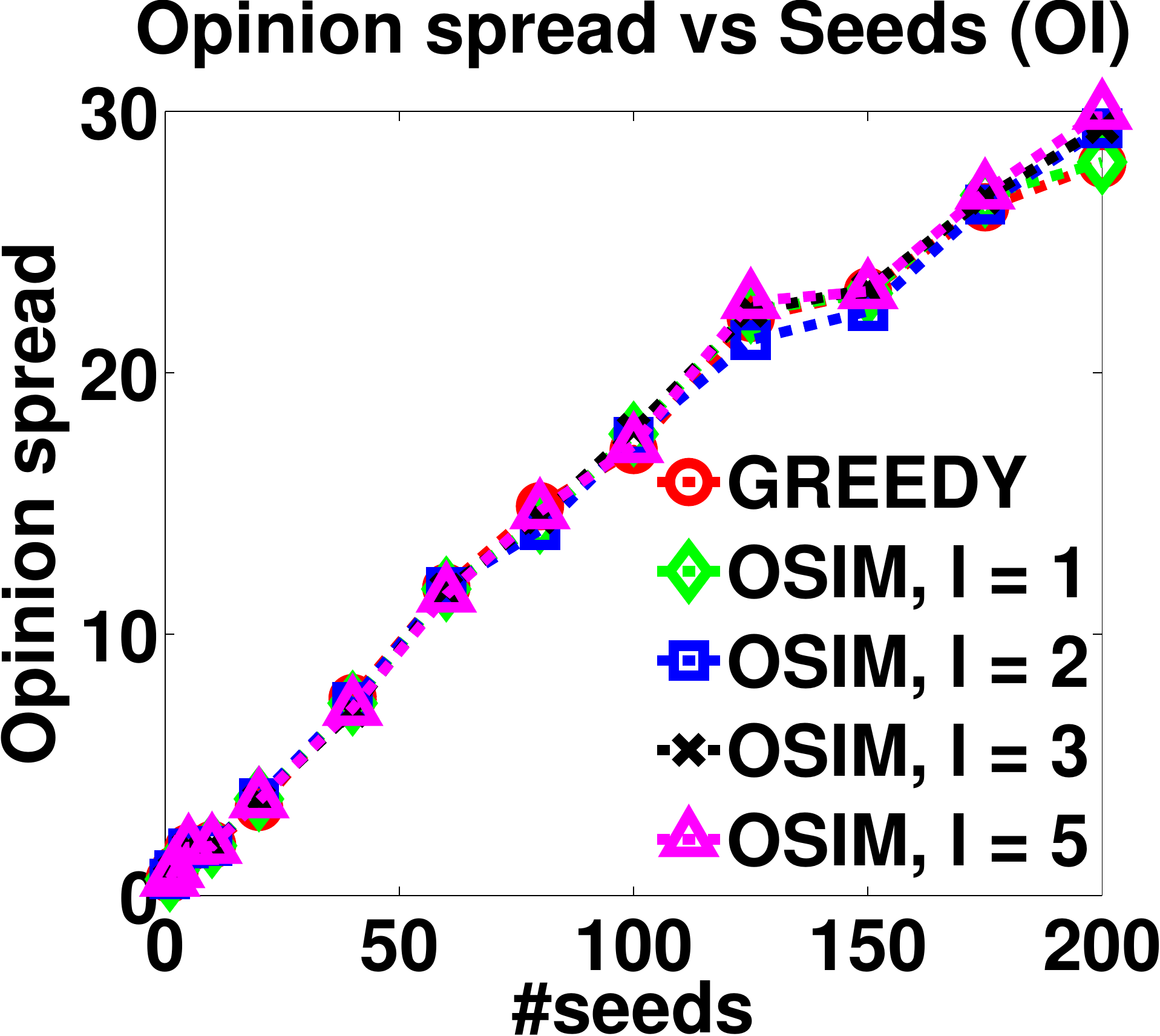}
		}
		\label{fig:eop_nethept_oi}
	}
	\subfloat[NetHEPT]
	{
		\scalebox{0.23}{
			\includegraphics[trim=0cm 0cm 10cm 0cm, width = 0.99\linewidth]{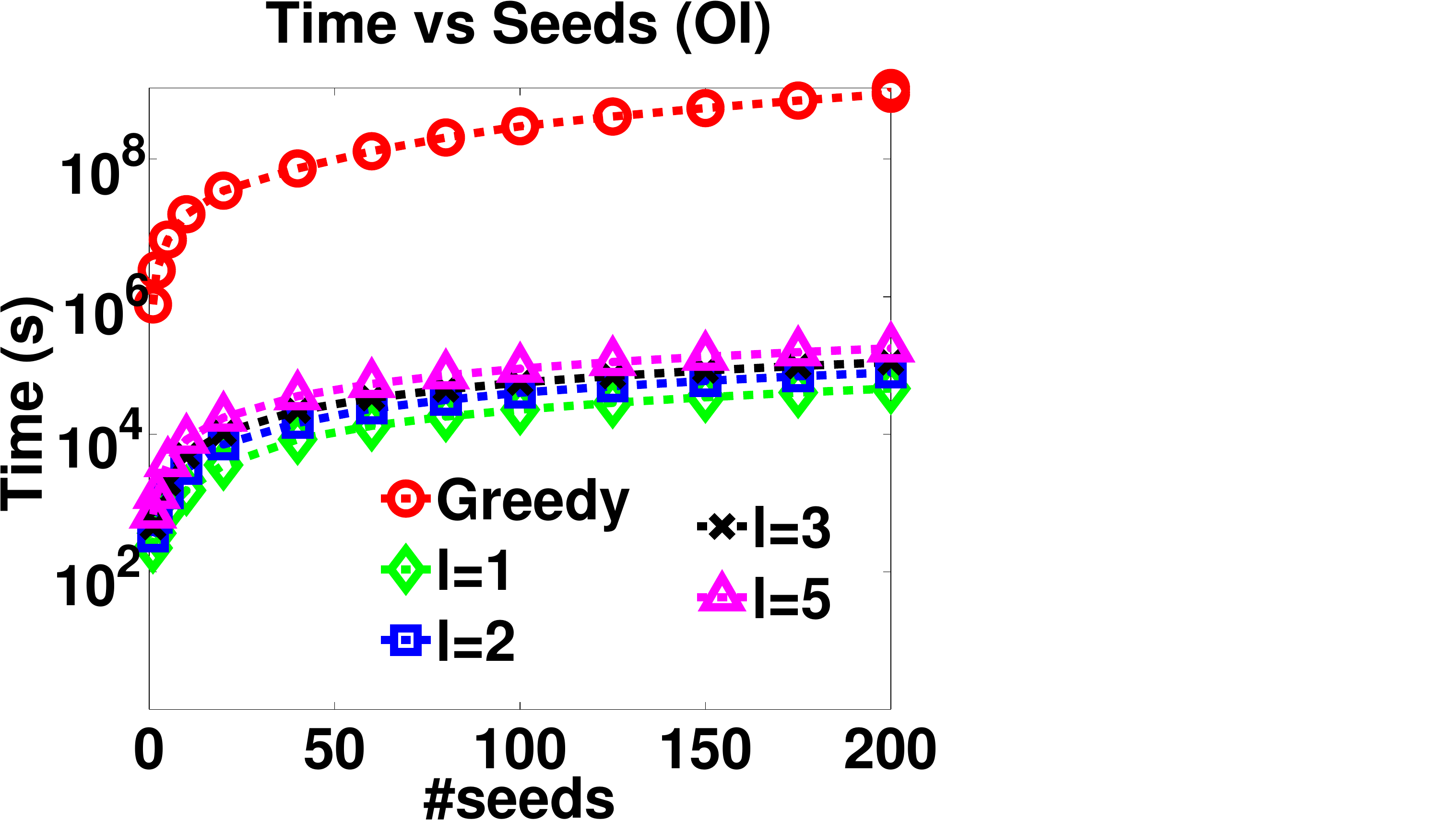}
		}
		\label{fig:time_nethept_oi}
	}
	\subfloat[Medium Datasets]
	{
		\scalebox{0.25}{
			\includegraphics[width = 0.99\linewidth, trim=0cm 0cm 0cm 0cm]{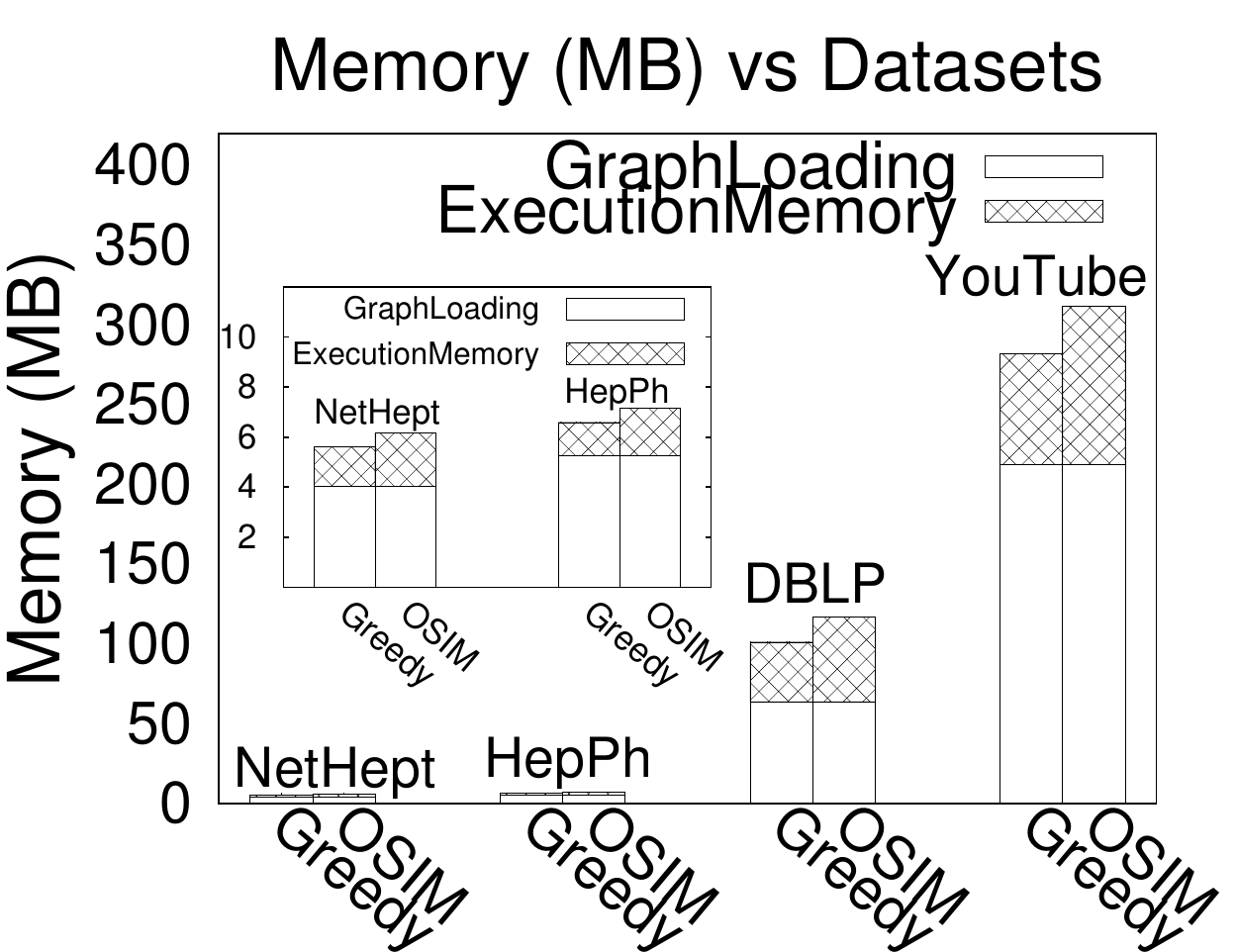}
		}
		\label{fig:memory_all_oi}
	}
\figcaption{\textbf{Comparison of average opinion-spread under OI, OC and IC with ground-truth for different topic-graphs on Twitter data with (a) k=$50$ seeds and (b) varying seeds. Opinion-spread ($\lambda=1$) comparison of OI with OC and IC on (c) Twitter data and (d) PAKDD data (churn analysis). (e) Opinion-spread comparison of $\lambda=1$ with $\lambda=0$ on NetHept and HepPh. (f) Growth of opinion-spread ($\lambda=1$) with $l$ and $k$ on NetHEPT (OI). (g) Growth-rate of running time with $l$ and $k$ on NetHEPT (OI). (h) Memory consumed for $k=100$ on Medium Datasets.}}
\label{fig:opinion_aware}
\end{figure*}

{\bf Comments}: Note that we do not perform scalability comparisons of \emph{OSIM} with the OVM algorithm \cite{ovm} as (1) Asymptotically, the time complexity of \emph{OSIM} $\big(O(k\diameter(m+n))$, where $\diameter$ is a small constant$\big)$ is better when compared to OVM $\big(O(k^2(m+n))\big)$, (2) It is rather difficult to extend OVM to work with the OI model and (3) OVM is designed to work with just the LT model at the first-layer while OSIM works with both IC and LT models. Moreover, since the time and space complexity analysis of \emph{OSIM} is exactly the same as that of \emph{EaSyIM}, the scalability comparisons performed for the latter serve for the former as well.
\ignore{

In the first experiment, fig~\ref{fig:opinion} we run the Kempe's~\cite{kempe} greedy algorithm neglecting sentiment which is the classical IC model of information propagation. On analyzing the opinion spread achieved by using these nodes in opinion oriented setting, spread of opinion achieved  is very less. On the other hand on running the greedy algorithm on the OI model, the spread achieved is much better.  This clearly justifies the fact that non-sentiment aware models clearly undermine the prevalence  of opinion and its propagation in the community. OI model clearly bridges the gap  between the real community and the IC model of information propagation.
\begin{figure}
\centering
\includegraphics[width= \linewidth]{images/opinion}
\caption{Opinion spread vs k for IC and OI model}
\label{fig:opinion}
\end{figure}
}

\subsection{Opinion-Aware}
\label{subsec:exp_opinion}

\subsubsection{Real Data: Twitter}
\label{subsubsec:twitter}
To motivate the importance of the OI model in real-world scenarios, we employed the use of data extracted from the \emph{Twitter} social network. This data has two aspects, namely -- (1) the underlying (directed) social network, and (2) the tweets by users of this network. A snapshot of the \emph{Twitter} network, crawled in June $2009$ containing $41.6M$ users (nodes) and $1.5B$ connections (edges), available publicly from \cite{twitter} is used as the underlying graph. A collection of $476M$ tweets gathered from a subset ($17M$ users) of the users ($41.6M$) in the underlying graph, crawled during the period of June $2009$ to December $2009$, available publicly from \cite{snap} constitutes the second part of our dataset. This dataset contains the following information for each tweet, namely -- (1) user-id, (2) time-stamp, and (3) the original text message or the tweet.

As discussed in previous sections, the motive of IM, and hence MEO, is to maximize the spread of information about a content, which can be a product, person, event and many more. For instance, in Twitter, the diffusion of information about a topic corresponds to the retweeting of or replying to a tweet or even tweeting about the topic. In fact, the temporal sequence of tweets by users, about a topic, guided by their connections in the network, defines the spread of information about that topic. Thus, as a first step towards analyzing information diffusion in real-world settings, there is a need to extract \emph{topic-focussed subgraphs}.

To this end, we pre-processed our dataset to retain all the tweets that possess at least one \emph{hashtag} (\#), where \emph{hashtags} constitute our \emph{topics}. Thus, we were left with $48.5M$ tweets corresponding to $1.6M$ users. Moreover, we projected the underlying twitter network on this identified subset of $1.6M$ users. This gave us a graph of $1.6M$ users with $75M$ edges among them. We denote this graph with the term \emph{background graph} in the rest of this section. Moving ahead, we describe the procedure for construction of topic-focussed subgraphs. A \emph{topic-focused subgraph} evolves by including nodes and edges when users tweet (or retweet) on the same topic, thus (re)activating the edges between them. Following this procedure, we scan the tweets in the increasing order of their time-stamps, and the users associated with the tweets are added as nodes incrementally. Moreover, a directed edge is constructed from a node $A$ to another node $B$ if this edge (directed) exists in the \emph{background graph} and vice versa. All the nodes with an in-degree of $0$ are considered to be originators or seeds of information. We stop growing a topic-graph if no new seed-nodes were added for a significant amount of elapsed time and start growing a new topic graph. The threshold on the elapsed time was learnt from the data. We look at the average frequency of tweets and identify a time difference as threshold that deviates significantly from the expected. This procedure resulted in $\approx10$--$12$ topic-focussed subgraphs per \emph{hashtag}. Note that the construction of topic-focused subgraphs requires just a single scan of the \emph{background graph}.

\begin{table}
\centering
\small
\scalebox{0.75}{
\begin{tabular}{|c||c|c|c|c|c|}\hline
\textbf{Dataset} & \textbf{n} & \textbf{m} & \textbf{Type} & \textbf{Avg. Degree} & \textbf{90-\%ile Diameter}\\\hline %
\hline
NetHEPT & 15K & 62K & Undirected & 4.1 & 8.8\\\hline
HepPh & 12K & 237K & Undirected & 19.75 & 5.8\\\hline
DBLP & 317K & 2.1M & Undirected & 6.63 & 8\\\hline
YouTube & 1.13M & 5.98M & Undirected & 5.29 & 6.5\\\hline
SocLiveJournal & 4.85M & 69M & Directed & 14.23 & 6.5\\\hline
Orkut & 3.07M & 234.2M & Undirected & 76.29 & 4.8\\\hline
Twitter & 41.6M & 1.5B & Directed & 36.06 & 5.1\\\hline
Friendster & 65.6M & 3.6B & Undirected & 54.88 & 5.8\\\hline
\end{tabular}
}
\tabcaption{\textbf{Datasets.}}
\label{tab:dataset}
\end{table}

With all the data preparation steps in place, next, we use popular sentiment analysis APIs \cite{text_processing,empath} to extract \emph{opinions} from the tweets. These sources learn a hierarchical classifier, which first determines whether a tweet is neutral or not. Neutral tweets are assigned an \emph{opinion} score of $0$. If the tweet was not neutral, then another classifier determines the probability of this tweet towards the positive class. Finally, we map this probability value ($[0,1]$) to our \emph{opinion} scores ($[-1,1]$). Note that the error introduced by the sentiment classifier would equally effect all the following results and thus, it can be safely ignored for all the analsyes.

Next, we show our analyses using topic-focussed subgraphs extracted from high frequency \emph{hashtags} (top-$100$), each of which possesses at least $50K$ tweets. Unless otherwise stated, the following portrayed results are averaged over all the topic-subgraphs and the \emph{normalized root-mean-square} error is used as the quality metric. We first show the results of opinion estimation for nodes on unseen topics. The opinions extracted from the tweets of a user, may include the effects of (1) her personal opinion, (2) the influence of her social contacts and (3) external factors. Modelling of external factors is relatively hard from the data, and thus, we consider just the first two factors. As discussed in Sec.~\ref{sec:intro}, we estimate the opinion of a node for a given topic, by considering the opinions of this node for similar topics. Since our topics possess a temporal aspect as well, we estimate the opinion of a node for a topic $A$, by performing a weighted average of the opinion of this node on all the related topics in the past. In addition, since we also possess the \emph{true} value of the opinion of nodes for the topic $A$, as output by our classifier, we portray the quality of this estimation procedure from our data. We were able to estimate the opinions of the \emph{seed-nodes} with an error\footnote{The error in estimation can be on the either side of the true value.} of $3.43\%$ and the opinions of all other nodes with an error of $8.57\%$. Note that the error on non-seed nodes is higher when compared to that on the seed nodes. This is indicative of the fact that the tweets of the seed-nodes indeed express their personal opinion, however the tweets of other nodes additionally include the effect of the opinions of their network. Moreover, this also proves that there is a need for models that: (1) consider the \emph{change of opinion} during information diffusion, and (2) consider the effect of \emph{interaction} between two nodes.

Having computed the \emph{opinion} of each node, the \emph{interaction} associated with an edge between two nodes (directed) is calculated as the fraction of the times they agree with each other across the subgraphs corresponding to all the topics in the past and not just those corresponding to the related topics. Note that both the parameters for the OI model, namely -- \emph{opinion} and \emph{interaction}, can be easily computed during the topic-subgraph construction step and do not incur any additional cost for their estimation. As mentioned earlier in this section that topic-subgraph construction is a required step for analyzing information diffusion, thus, the parameter estimation cost for the OI model is amortized constant.

Next, we show that the \emph{opinion-spread} in real-world follows the OI model. For this experiment, we consider the top-$100$ extracted topic-subgraphs and use the identified seeds or originators of information from the real-world, as the seed nodes for the information diffusion process to measure the \emph{opinion spread}. For each topic subgraph we calculate the opinion-spread using the opinions extracted from their tweets, which serves as our \emph{ground truth}. Similarly, we use the estimated parameters on these topic subgraphs and obtain the \emph{opinion spread} under the OI, OC and the IC models. Although, depending upon the distribution of opinions and interactions, the opinion-spread achieved using all the three models can fall on the either side of the ground-truth, our results in Figure~\ref{fig:twitter_eop_flow} show that the opinion-spread obtained under the OI model is always closest to the ground-truth opinion-spread and thus, possesses the least error. Moreover, Figure~\ref{fig:twitter_eop_flow_seeds} shows a similar analysis on the average opinion-spread using the OI, OC and IC models with varying seeds, and it is evident that the \emph{opinion-spread} under the OI model possesses the least error.

\begin{figure*}
\centering
	\subfloat[NetHEPT]
	{
		\scalebox{0.185}{
			\includegraphics[width = 0.99\linewidth]{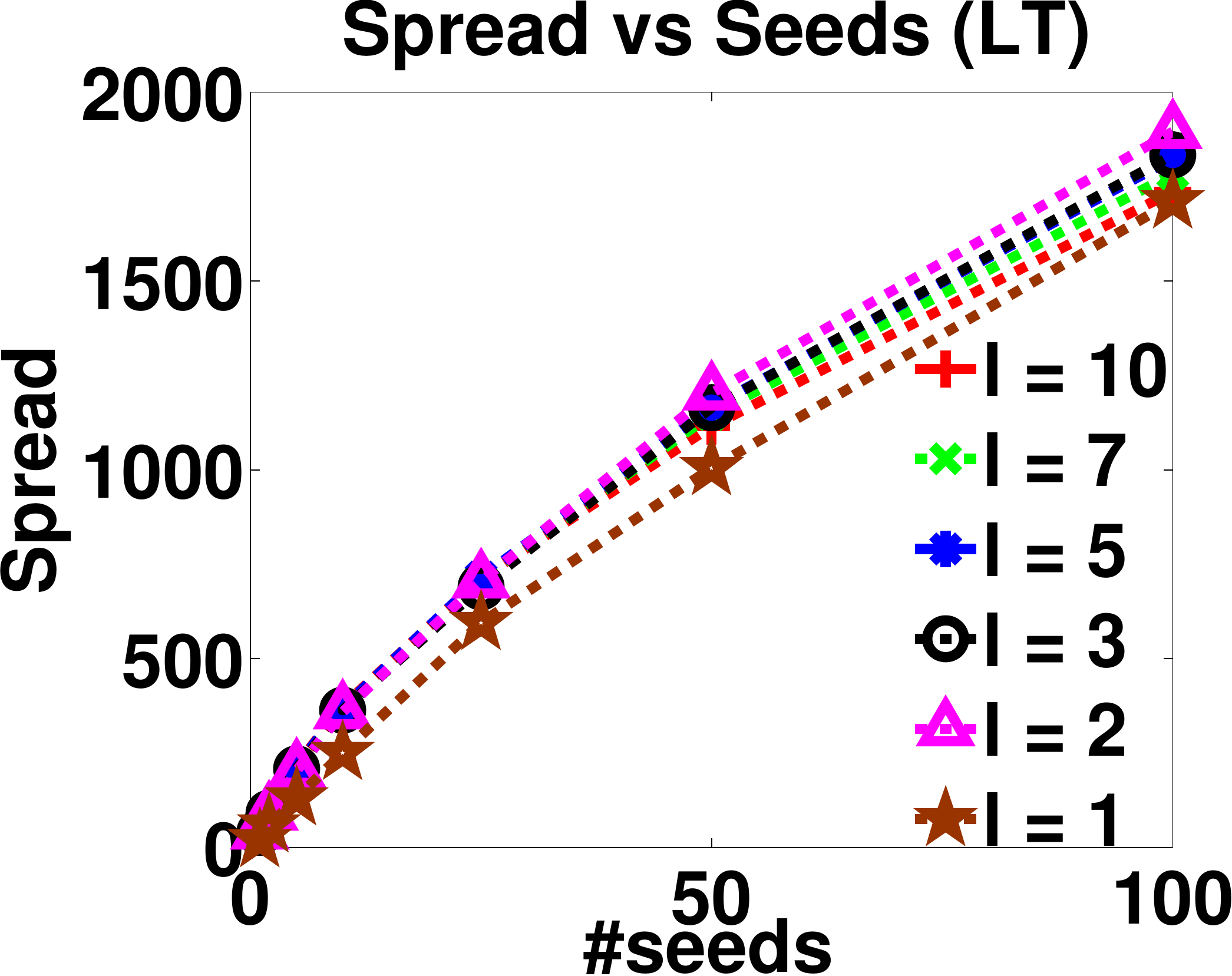}
		}
		\label{fig:spread_nethept_lt}
	}
	\subfloat[DBLP]
	{
		\scalebox{0.185}{
			\includegraphics[width = 0.99\linewidth]{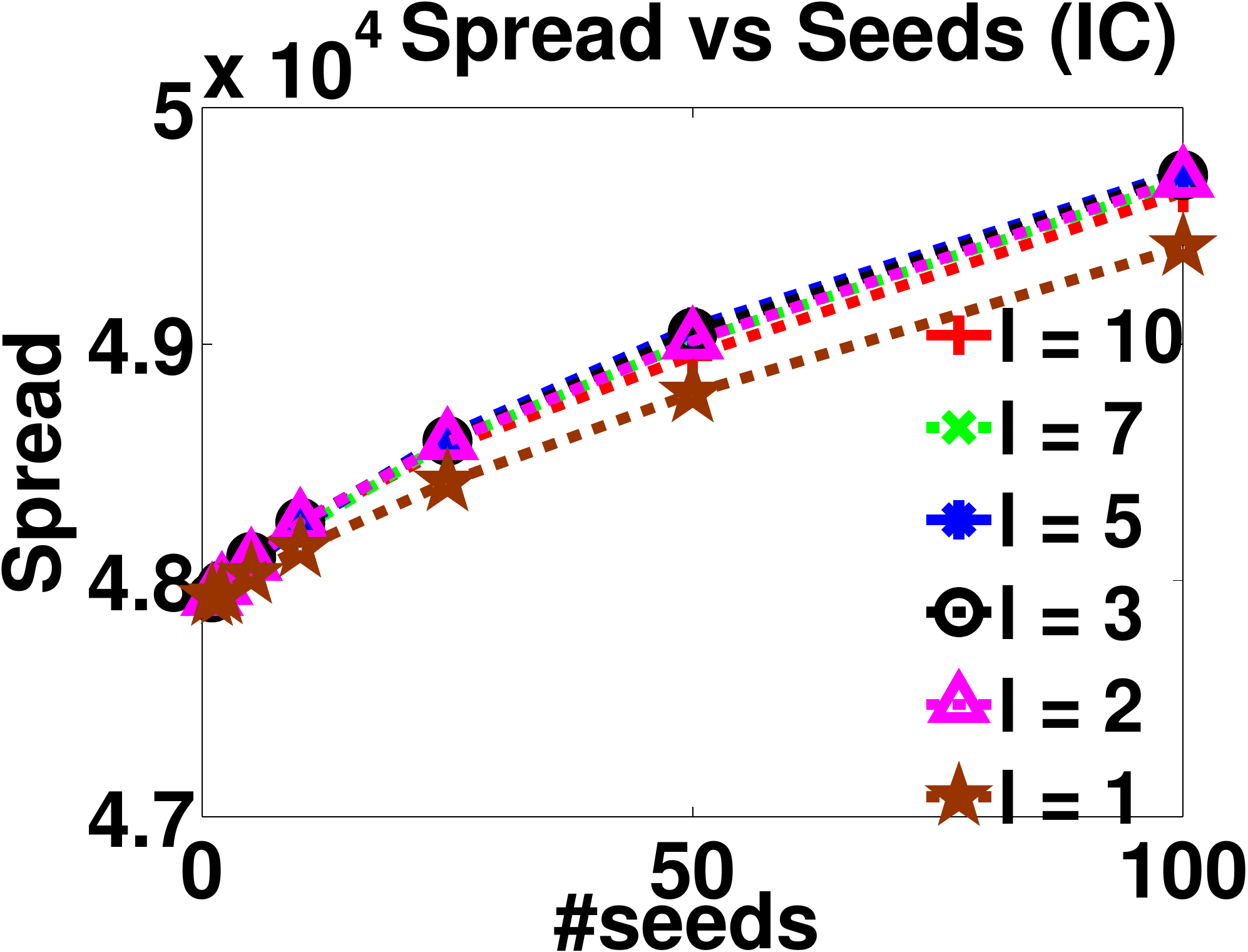}
		}
		\label{fig:spread_dblp_ic}
	}
	\subfloat[YouTube]
	{
		\scalebox{0.185}{
			\includegraphics[width = 0.99\linewidth]{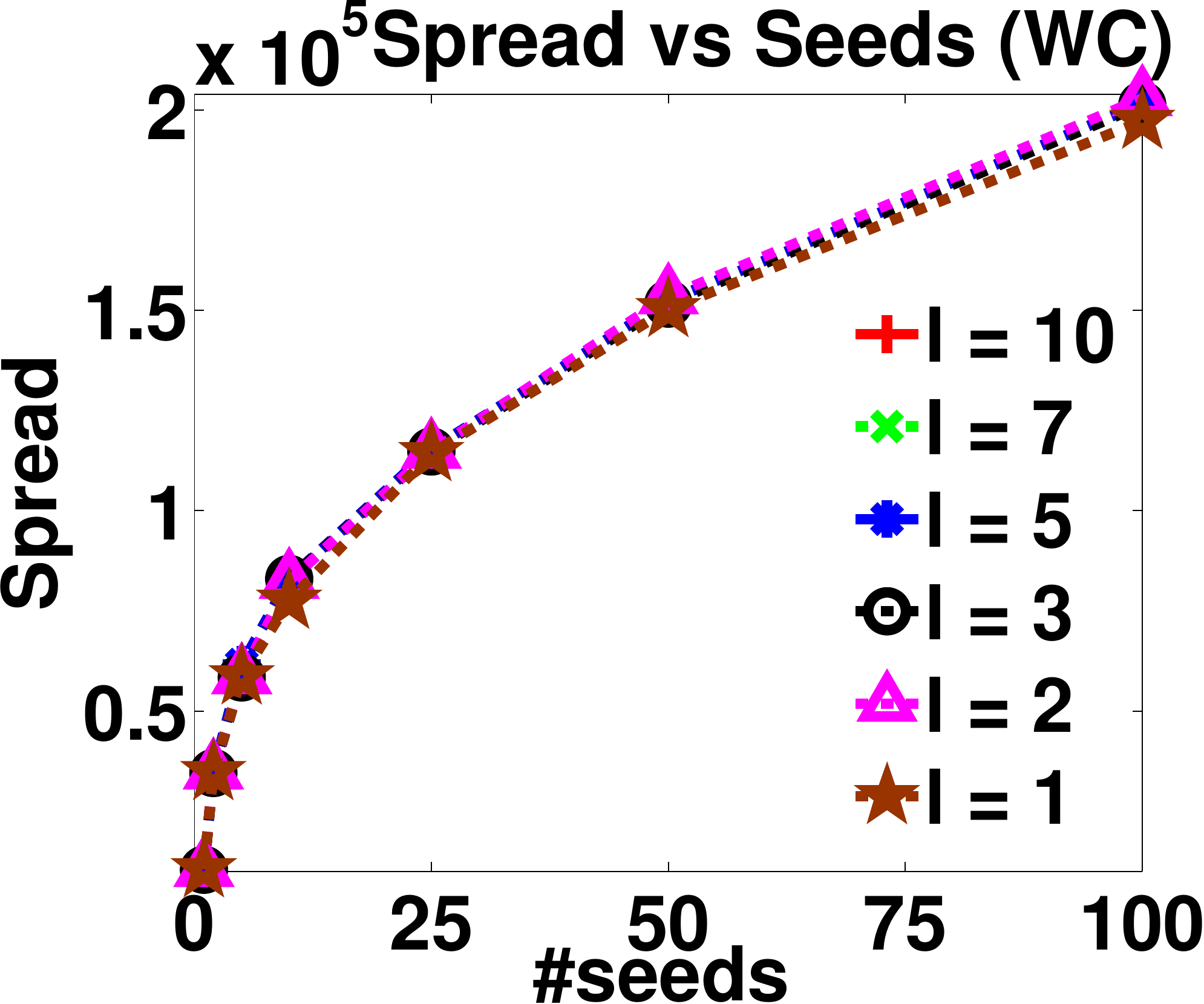}
		}
		\label{fig:spread_youtube_wc}
	}
	\subfloat[HepPh]
	{
		\scalebox{0.185}{
			\includegraphics[width = 0.99\linewidth]{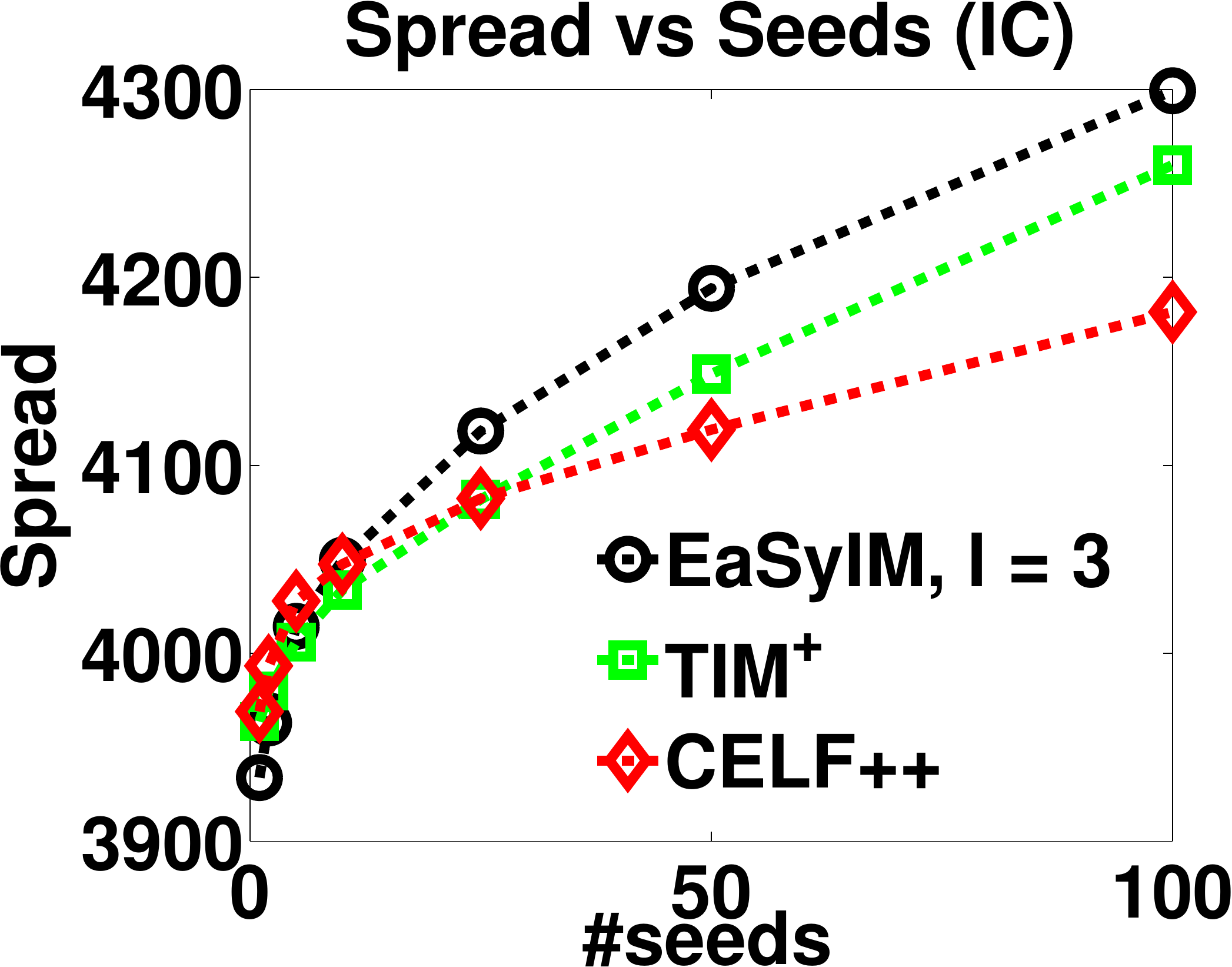}
		}
		\label{fig:spread_hepph_ic_compare}
	}
	\subfloat[DBLP]
	{
		\scalebox{0.185}{
			\includegraphics[width = 0.99\linewidth]{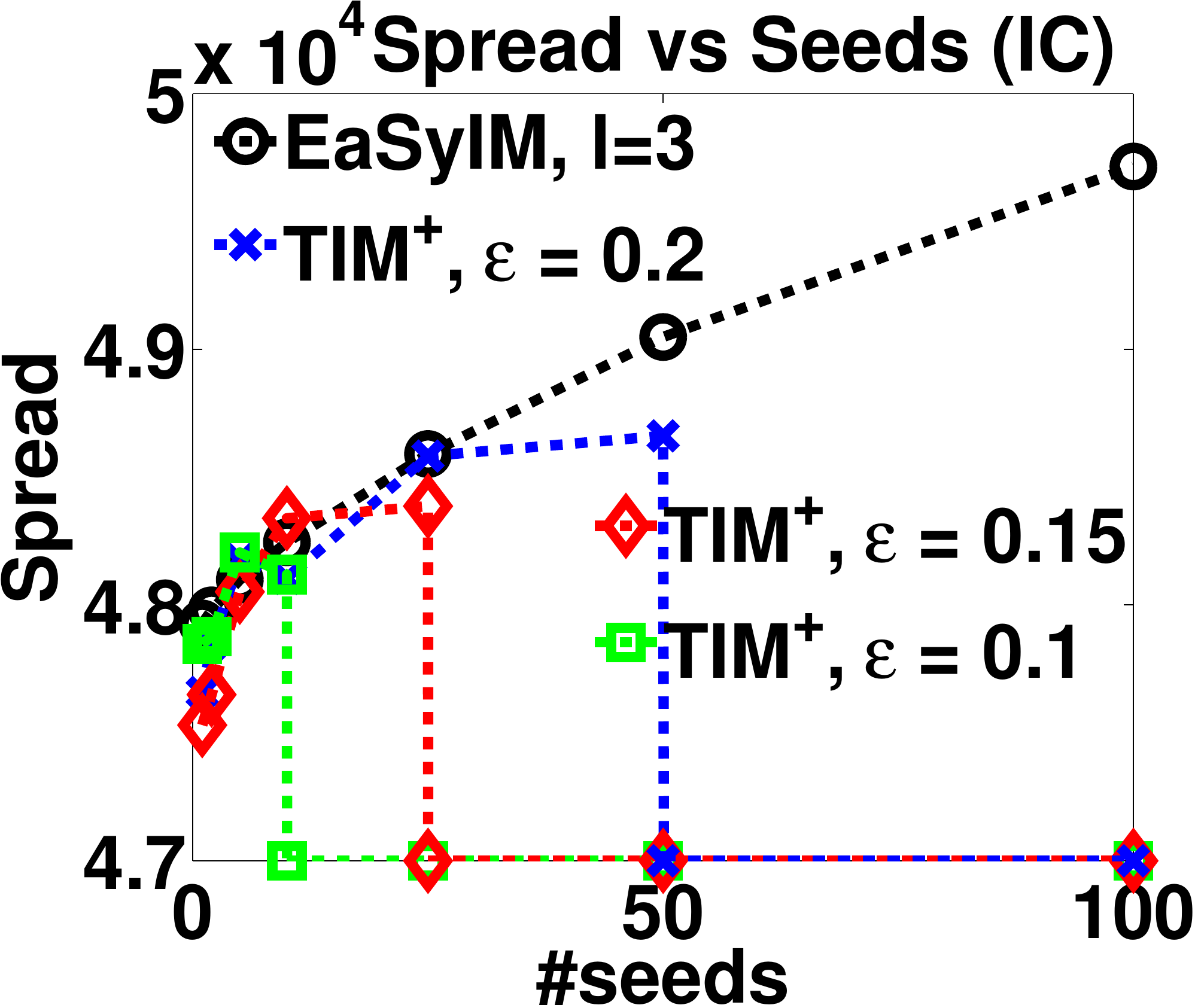}
		}
		\label{fig:spread_dblp_ic_compare}
	}
	\\
	\vspace{-2.5mm}
	\subfloat[NetHEPT]
	{
		\scalebox{0.185}{
			\includegraphics[width = 0.99\linewidth]{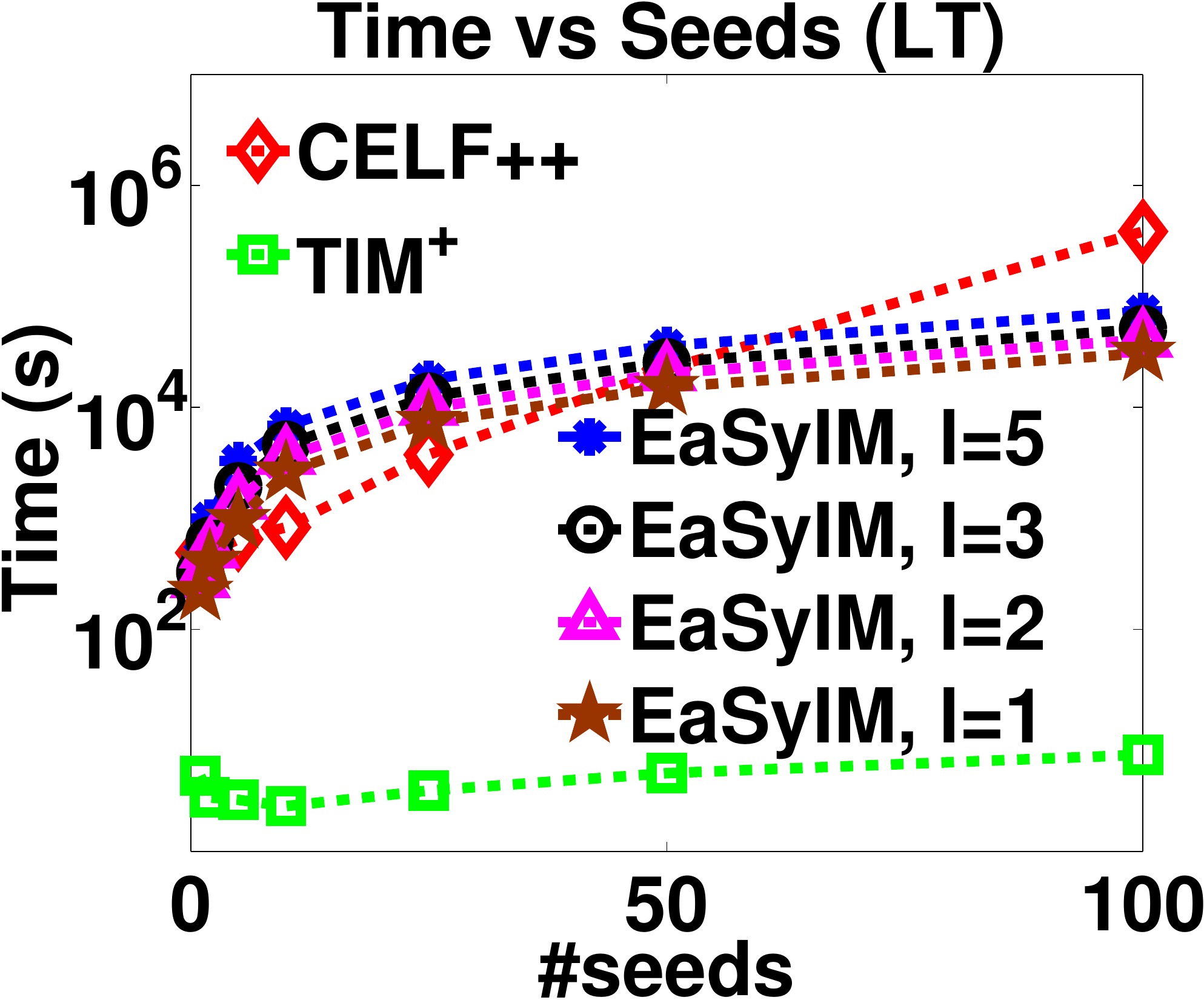}
		}
		\label{fig:time_nethept_lt_compare}
	}
	\subfloat[DBLP]
	{
		\scalebox{0.185}{
			\includegraphics[width = 0.99\linewidth]{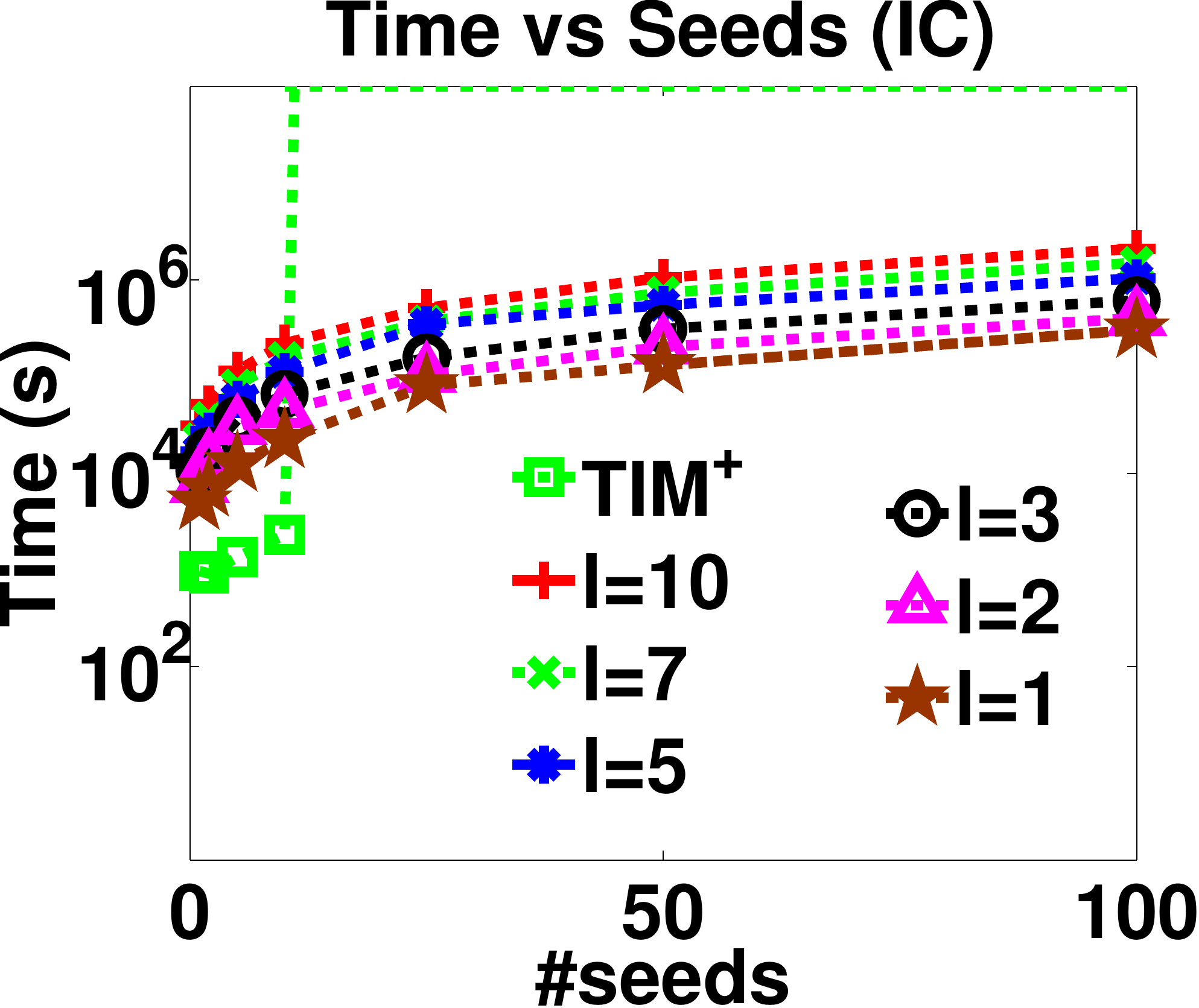}
		}
		\label{fig:time_dblp_ic_compare}
	}
	\subfloat[YouTube]
	{
		\scalebox{0.185}{
			\includegraphics[width = 0.99\linewidth]{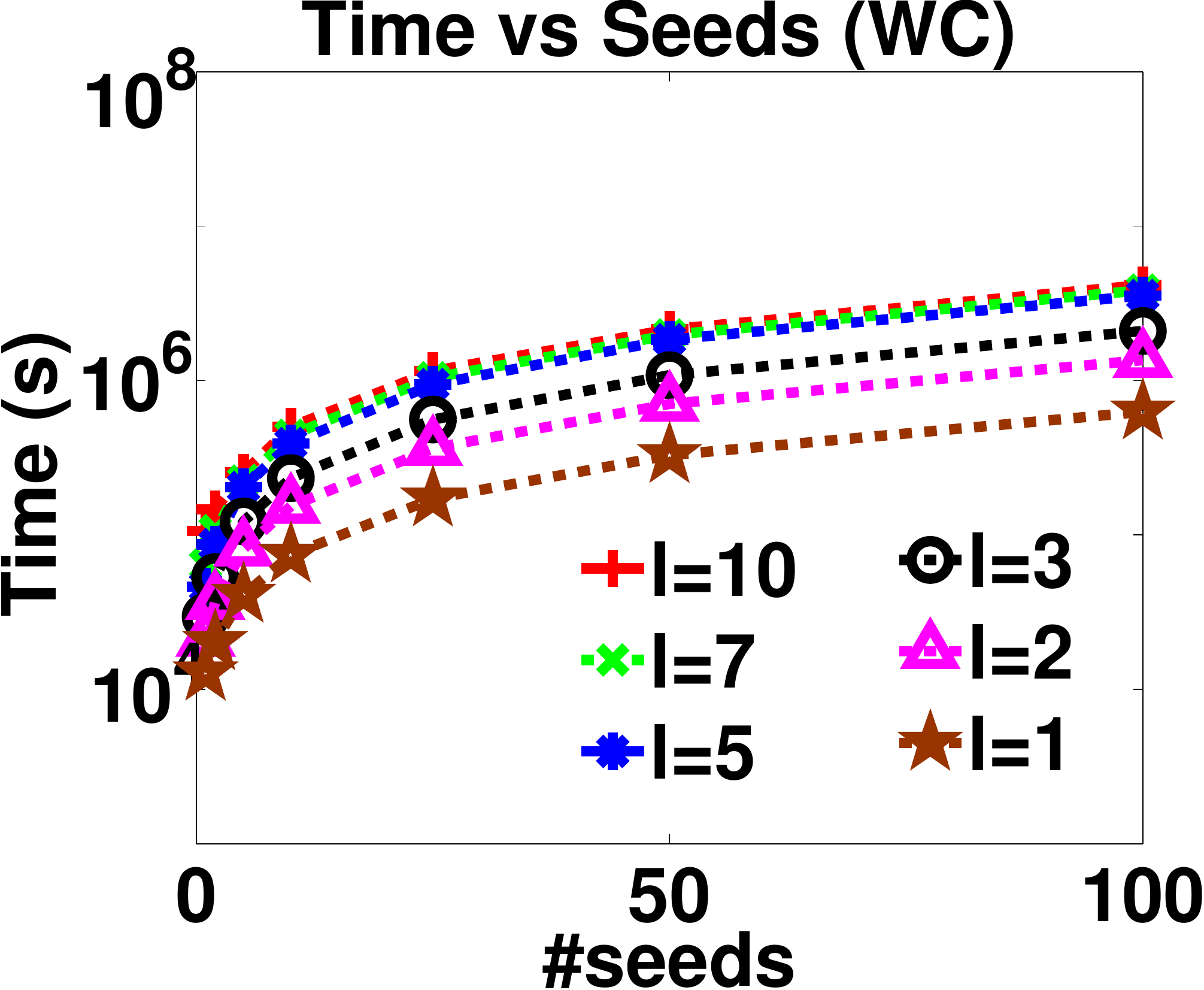}
		}
		\label{fig:time_youtube_wc_compare}
	}
	\subfloat[NetHEPT and DBLP]
	{
		\scalebox{0.185}{
			\includegraphics[width = 0.99\linewidth]{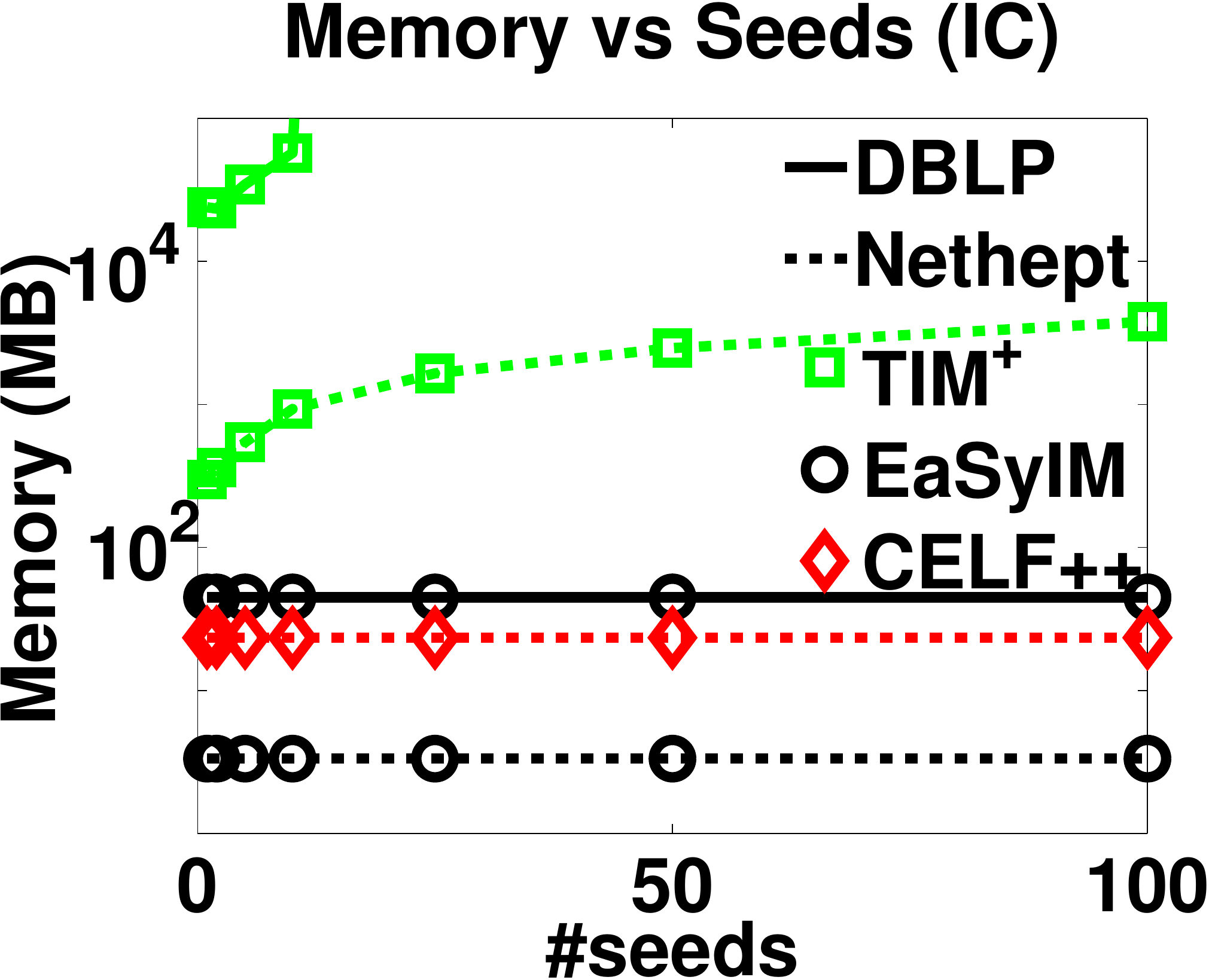}
		}
		\label{fig:memory_growth_nethept_dblp_compare}
	}
	\subfloat[Medium Datasets]
	{
		\scalebox{0.21}{
			\includegraphics[width = 0.99\linewidth]{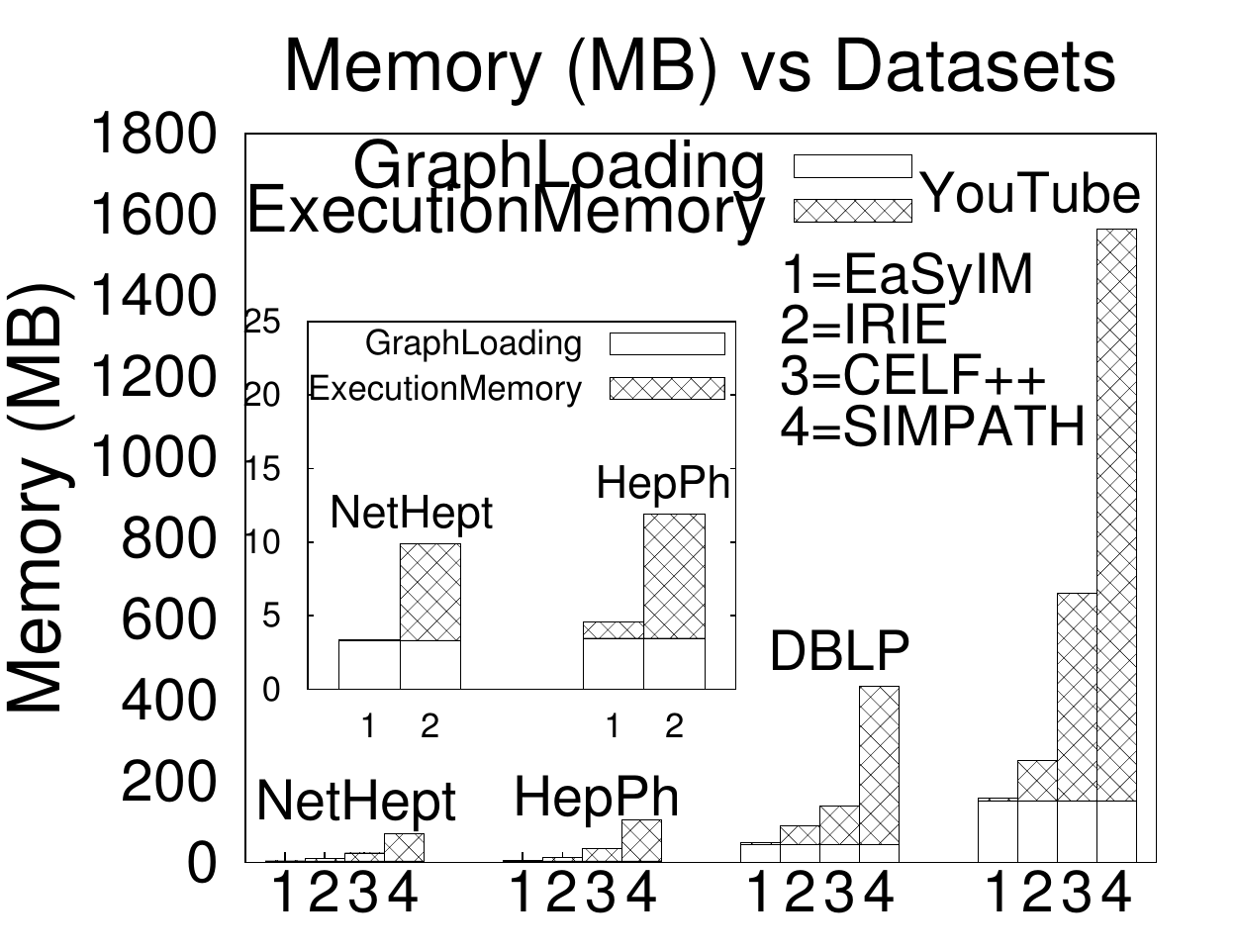}
		}
		\label{fig:memory_all_all_100_compare}
	}
\figcaption{\textbf{Growth of spread with $l$ and $k$ on (a) NetHEPT (LT), (b) DBLP (IC) and (c) YouTube (WC). Spread comparison of \emph{EaSyIM} with TIM$^+$ and CELF++ on (e) HepPh (IC), (f) DBLP (IC). Running time comparison of \emph{EaSyIM} with CELF++ and TIM$^+$ on (h) NetHEPT (LT), (i) DBLP (IC) and (j) YouTube (WC). Memory consumption comparison of \emph{EaSyIM} with CELF++ and TIM$^+$ on (m) NetHEPT and DBLP (IC), and with all algorithms for $k=100$ on (n) Medium Datasets (IC).}}
\label{fig:opinion_oblivious}
\end{figure*}

Having empirically proved that the OI model closely mirrors the \emph{opinion-spread} in real world scenarios, finally, we show our results on opinion-aware IM. As mentioned earlier in this section, we show the average opinion-spread achieved for different topics by performing the information diffusion process on the entire \emph{background graph}. Figure~\ref{fig:twitter_eop_spread} shows that the opinion-spread achieved using the seeds selected by the OI model is much better when compared to that of the OC and the IC model. Apart from the results on quality, \emph{OSIM} with $l=1$, took $497$ minutes and $1.72GB$ of RAM to identify $100$ seeds on the \emph{background graph}.

\ignore{
We first show the need of constructing topic-focussed subgraphs..

	[Reason] Since the motive of the IM and hence the MEO problem 

Total topic-graphs se we get $1.7M$ nodes. We consider top-100 (by frequency topics). How many of these users are covered? (Stats)\\
Calculate opinion using sentiment analysis APIs..\\
(1) Show that opinions can be estimated using similar topics. Prove this hypothesis using just the seed nodes as this is not possible with other nodes..\\
Calculate interaction...

(2) Show that there is a need for interaction. Since the estimated opinions doesn't match with the opinion extracted from tweets for non-seed nodes or non-originators. This is because there are extraneous factors 
opinion-spread achieved using the seeds selected by the OI model is much better when compared to that of the OC and the IC model. This proves the importance of the OI model and the notions of \emph{opinion} and \emph{interaction} in solving real-world problems.

Now given the topic-focussed subgraphs and background graph, we show what is the objective of Influence Maximization in this case and how do we perform the same. Show how can we construct opinion and interaction.

Next, we show the construction of background graph..
	We first construct a subgraph of the underlying twitter network using the $1.7M$ users for whom we have the tweet information and also are a part of at least one of the topic-focussed subgraphs (Using the nodes that have at least 1 tweet with "\#"-tag we create a subgraph of the underlying social network).. 

Thus, the processed underlying social network now contains $1.7M$ nodes and $75M$ edges.  
}

\subsubsection{Real Data: Analyzing Customer Churn}
\label{subsubsec:analyzing_churn}
We used the PAKDD 2012 data mining competition dataset, available publicly from \cite{churn_data}, to motivate the applicability of \emph{MEO} and the OI model in solving real-world problems. The dataset contains profiles of $\approx 100K$ customers, including billing information, usage data, service requests and complaints, of a large telecommunications company for a period of $1$ year from January $1$, $2011$ to December $31$, $2011$. The data also contains information about the \emph{churn} (termination of services) date for each customer, where $17K$ customers are churners and the rest are non-churners. Using this data, the objective of the challenge was to predict which customers are liable to churn. We extend this problem of modelling \emph{customer-churn} further to perform a novel analysis using \emph{opinion-aware} IM. Since the main objective of this paper is not churn-prediction, to make the data preparation step for our analysis simpler we work on a balanced subset of the original data containing $34K$ customers with equal number of churners and non-churners. 

Owing to the increasing popularity of \emph{label propagation} \cite{label_prop} and similarity based learning \cite{similarity_learning}, we build upon the hypothesis that customers with \emph{similar} attributes possess similar \emph{churn} behavior. Under this hypothesis, we induce a graph, of \emph{customers}, using \emph{attribute-value} similarity and a similarity \emph{threshold} to construct edges between two customers. The graph obtained thus, consists of $34K$ nodes (customers) and $1.5M$ edges (relationships), with churners assigned a label of $-1$ and non-churners that of $1$. Using this data, we build a model using label-propagation to predict the susceptibility of a customer towards churn. After convergence, the value ($[-1,1]$) at each node represents the \emph{affinity}, or in other words the \emph{opinion}, of a customer to churn with $-1$ and $1$ denoting strong affinity towards churn and non-churn respectively. In this way, we estimate the \emph{opinion} ($o$) parameter of the OI model. The attribute-value similarity defines the influence-proability ($p$) associated with each edge, while the \emph{interaction} probability ($\varphi$) for each edge is generated using the $rand(0,1)$ function described above.

Acquainted with the likelihood of customers to churn combined with the availability of a churn-prediction model, a service provider would like to identify potential targets, under a marketing buget, capable of maintaining its reputation, thus, indirectly preventing the \emph{cascades} of churn. This task boils down to identifying seeds that maximize the difference of positive and negative opinions, or in other words maximize the effective opinion (MEO). Figure~\ref{fig:real_world_eop_spread} shows that the opinion-spread achieved using the seeds selected by the OI model is much better when compared to that of the OC and the IC model. This proves the importance of the OI model and the notions of \emph{opinion} and \emph{interaction} in solving real-world problems.

\subsubsection{Other Real Datasets}
Having highlighted the importance of \emph{opinion}, \emph{interaction}, the OI model and the MEO problem using two real-world datasets in the previous sections, here we present an in-depth analysis of \emph{opinion-aware} IM on the benchmark datasets (Table~\ref{tab:dataset}) used for the evaluation of classical IM. Since these datasets do not inherently possess the properties of \emph{opinion} and \emph{interaction} with nodes and edges of the graph respectively, we annotate them as follows. The \emph{opinions} are generated by two methods, namely -- (a) $o\sim rand(-1,1)$, where the opinion of each node is generated uniformly at random in the range $[-1,1]$ and, (b) $o\sim\mathcal{N}(0,1)$, where the generated opinions follow the standard normal distribution, while the \emph{interaction} probability for each edge is only generated uniformly and randomly using the function $\varphi\sim rand(0,1)$. The reported results are averaged over $3$ different instances of the generated data. For additional results please see Appendix~\ref{app_res:opinion_aware}.

\textbf{Quality}: The first set of results show the importance of the choice of optimization-objective. We compare the \emph{effective opinion-spread} ($\lambda=1$) with the \emph{opinion-spread} ($\lambda=0$) for the NetHEPT dataset using \emph{OSIM} under the OI model. It is evident from Figure~\ref{fig:nethept_hepph_lambda} that $\lambda=1$ outperforms $\lambda=0$ and hence, maximizing the effective opinion-spread is better. With this, we fix $\lambda=1$ for a comparison of the effective opinion-spread obtained using \emph{OSIM} (with varying $l$) with Modified-{\kempegreedy} for different datasets. Figure~\ref{fig:eop_nethept_oi} presents the results with varying seeds for NetHEPT under the OI model ($o\sim\mathcal{N}(0,1)$). It can be seen that the reported opinion spread improves with increase in $l$, however, it starts to dip when $l \to \diameter$ owing to a significant increase in the number of cyclic paths, which in turn causes the error in the assigned scores to increase as discussed in Sec.~\ref{subsubsec:discussion}. Using these results, we conclude that $l=3$ serves as the best choice for \emph{OSIM}. Moreover, \emph{OSIM} closely mirrors Modified-{\kempegreedy} with respect to the quality of obtained \emph{opinion-spread}.

\textbf{Efficiency}: In continuation to the comparison on quality, here we compare \emph{OSIM} with Modified-{\kempegreedy} on its running-time for the NetHEPT dataset. Figure~\ref{fig:time_nethept_oi} shows that \emph{OSIM} is at least $10^3$ times more efficient when compared to Modified-{\kempegreedy} with the gain going as high as $10^5$ for some cases. Moreover, this figure (with y-axis in log-scale) also shows that the running-time of \emph{OSIM} grows linearly with increase in $l$.

\textbf{Scalability}: Figure~\ref{fig:memory_all_oi} shows a comparison of the memory-consumed by \emph{OSIM} with Modified-{\kempegreedy} for all the $4$ datasets discussed above. It is important to note that the memory consumption of these algorithms is independent of the parameters, namely -- path-length ($l$), number of seeds ($k$) and the number of MC simulations. Both the algorithms scale linearly in the size of the graph and require constant-factor overheads to solve the MEO problem. This is indicated by the stacked bars, which show that these algorithms require only a small amount (constant-factor) of memory over and above the memory required to load the graph.

\begin{table}[t]
\centering
\scalebox{0.7}{
\begin{tabular}{|c|c|c|c||c|c|c|}
\hline
\multirow{2}{*}{\bf Dataset} & \multicolumn{3}{|c||}{\bf Running Time (min)} &
\multicolumn{3}{c|}{\bf Memory (MB)} \\
\cline{2-7}
& \bf TIM$^+$ & \bf EaSyIM (l=1)& \bf Gain & \bf TIM$^+$ & \bf EaSyIM (l=1)& \bf Gain \\
\hline
\hline
DBLP& 783.1 & 2183 & 0.36x & 35234.75 & 46.5 & \bf 758x \\\hline
YouTube& NA & 5089.5 & $\infinity$ & NA & 158.3 & $\infinity$ \\\hline
socLive & NA & 15433.33 & \bf $\infinity$ & NA & 974.94 & $\infinity$ \\\hline
\end{tabular}
}
\tabcaption{{\bf Comparing \emph{EaSyIM} with TIM$^+$, $k=50$, $\epsilon=0.1$.}}
\label{tab:resultsTIM}
\end{table}

As stated earlier, in addition to the scalability and efficiency analysis for \emph{OSIM}, next, we present a more involved analysis for \emph{EaSyIM} which serves the purpose for the former as well.

\subsection{Opinion-Oblivious}
\textbf{Quality}: The first set of results portray the effect of the parameter $l$ on the \emph{spread} obtained using \emph{EaSyIM}. Figures~\ref{fig:spread_nethept_lt},~\ref{fig:spread_dblp_ic} and~\ref{fig:spread_youtube_wc} present the results with varying seeds on the NetHEPT, DBLP and the YouTube datasets under the LT, IC and WC models respectively. It can be seen that the reported spread follows a similar pattern as that of \emph{OSIM}, i.e., it improves with increase in $l$, however, it starts to dip when $l \to \diameter$. Using these results, we conclude that $l=5$ serves as the best choice for \emph{EaSyIM}, with $l=3$ being marginally worse when compared to $l=5$. Moreover, owing to $l=3$ being $\approx2$ times more efficient when compared to $l=5$ (detailed analysis on efficiency in the following paragraph), we fix $l=3$ for a comparison of the spread obtained using \emph{EaSyIM} with other algorithms. Figures~\ref{fig:spread_hepph_ic_compare} and~\ref{fig:spread_dblp_ic_compare} present the spreads obtained for the HepPh and DBLP datasets under the IC model. Note that the sudden drop in the spread of TIM$^+$ in Figure~\ref{fig:spread_dblp_ic_compare} indicates that it crashed on our machine after that point. These results show that \emph{EaSyIM} mirrors the state-of-the-art techniques closely for majority of the datasets, while being \emph{even better} at certain instances, with respect to the quality of the obtained \emph{spread}.

\textbf{Efficiency}: Figures~\ref{fig:time_nethept_lt_compare},~\ref{fig:time_dblp_ic_compare} and~\ref{fig:time_youtube_wc_compare} show a comparison of the running-times of \emph{EaSyIM} with CELF++ and TIM$^{+}$ for varying seeds on the NetHEPT, DBLP and the YouTube datasets under the LT, IC and WC models respectively. These figures show that \emph{EaSyIM} grows linearly (y-axis in log-scale) with the parameters $l$ and $k$, therefore \emph{EaSyIM} is $\approx2$ times faster for $l=3$ when compared to $l=5$. Owing to the marginal improvement in the spread from $l=3$ to $l=5$ (quality analysis in the previous paragraph), we choose $l=3$ for \emph{EaSyIM} as it provides the best tradeoff between efficiency and quality. Moreover, it can be seen that neither CELF++ nor TIM$^+$ scale well with the increase in graph size. CELF++ was not able to complete on the DBLP and the YouTube datasets even after running for $7$ days, while TIM$^+$ crashed on the DBLP dataset for $k>10,25,50$ with $\epsilon=0.1,0.15,0.2$ respectively owing to its huge memory requirement. In addition, it is evident from Tables~\ref{tab:resultsTIM} and~\ref{tab:resultsCELFPP} that \emph{EaSyIM} is $10$--$15$ times more efficient while consuming $3$-$4$ times less memory when compared to CELF++, and requires $8$--$10$ times more time to run while its memory-footprint is $\approx 500$ times smaller when compared to TIM$^+$. We argue that since \emph{EaSyIM} can be efficiently parallelized owing to the independence of the MC simulations, its lack of efficiency when compared to TIM$^+$ can be easily mitigated by running it in parallel on $8$ cores\footnote{Considering the ease of availability of multiple cores in a single machine when compared to large amount of RAM.} while ensuring the memory gain to be the same.


\textbf{Scalability}: Figure~\ref{fig:memory_growth_nethept_dblp_compare} portrays the growth of the memory-footprint of \emph{EaSyIM}, CELF++, and TIM$^+$ with varying seeds for the NetHEPT and the DBLP datasets. It is evident that the memory-consumption of all other techniques, barring TIM$^+$ (which grows at a much faster rate), grows linearly with the number of seeds. It is also clear that \emph{EaSyIM} possesses the smallest memory-footprint, being $\approx4$ and $\approx500$ times smaller when compared to CELF++ and TIM$^+$ respectively. Moreover, TIM$^+$ ($\epsilon=0.1$) crashed for $k>10$ on the DBLP dataset due to going out of memory on our machine. To further this, our experiments revealed that even after relaxing $\epsilon$ to $0.15$ and $0.2$ it crashed for $k>25$ and $k>50$ respectively. This analysis unravels the obvious problem with TIM$^{+}$, i.e., though being efficient it requires an exorbitantly high amount of main-memory and thus, cannot be termed \emph{scalable}. Figure~\ref{fig:memory_all_all_100_compare} shows the additional amount of memory required, over and above the memory required to load the graph, by each of our benchmarking algorithms for $4$ different datasets. It is clear that \emph{EaSyIM} possesses the least overhead while SIMPATH possesses the highest. Note that we omit the bar corresponding to TIM$^+$, as it hinders the comparison with other techniques owing to its memory consumption being too high. This shows that \emph{EaSyIM} possesses the capability of scaling to very large graphs. For additional results please see Appendix~\ref{app_res:opinion_aware}.

\begin{table}[t]
\centering
\scalebox{0.7}{
\begin{tabular}{|c|c|c|c||c|c|c|}
\hline
\multirow{2}{*}{\bf Dataset} & \multicolumn{3}{|c||}{\bf Running Time (min)} &
\multicolumn{3}{c|}{\bf Memory (MB)} \\
\cline{2-7}
& \bf CELF++ & \bf EaSyIM (l=1) & \bf Gain & \bf CELF++ & \bf EaSyIM (l=1) & \bf Gain \\
\hline
\hline
NetHEPT & 5352.25 & 118 & \bf 45.35x & 23.26 & 3.39 & 6.86x \\\hline
HepPh & 9746.74 & 230 & \bf 41x & 24.60 & 3.47 & 7.08x \\\hline
DBLP & NA & 5071.67 & \bf\infinity & NA & 44.73 & \bf \infinity \\\hline
\end{tabular}
}
\tabcaption{{\bf Comparing \emph{EaSyIM} with CELF++, $k=100$ and $l=1$.}}
\label{tab:resultsCELFPP}
\end{table}

In summary, both \emph{OSIM} and \emph{EaSyIM} provide the best trade-off between \emph{memory-consumption} and \emph{running-time}, and possess the capability to perform IM on large graphs using moderately sized machines or even a laptop.

\section{Related Work}
\label{sec:related}
The IM problem has been studied extensively over the past decade, thus, it is rather difficult to write a complete literature review in one section. A detailed comparison of the OI model with both IC-N \cite{negIC} and OC \cite{ovm}, the only other models for \emph{opinion-aware} IM, has been done qualitatively (Sec.~\ref{sec:intro}) and empirically (Sec.~\ref{sec:exp}). Here, we overview the existing works that overlap with our problem.

{\bf(1) \emph{Topic-Aware and Competitive IM}}: The motive of topic-aware IM \cite{topic_aware_im1, topic_aware_im2, topic_aware_im3} is to design strategies for maximizing the collective \emph{spread} of information under a given topic-distribution. It allows different influence probabilities for each topic (assuming a joint distribution of topics), without considering the opinions (of a node) towards these topics. On the other hand, the objective of competitive IM \cite{competitive_im1, competitive_im2, competitive_im3} is to maximize the \emph{spread} of information, about a specific content, in the presence of competitors. It is rather intuitive that competitive IM is similar to topic-aware IM, thus, both of them suffer with same set of deficiencies, i.e., (a) the absence of opinions and negatively activated nodes, and (b) the incapability of capturing their change as information propagates. In addition, another limitation for competitive IM lies with the information-diffusion process, where each node (strictly) propagates information about a particular content and any newly activated node (by the former) also (strictly) follows the same content. In summary, both of them consider maximization of the number of active nodes (opinion-oblivious notion of \emph{spread}), rendering the spread function to be submodular and thus, are closer to the \emph{classical} IM problem.

{\bf(2) \emph{IM in signed networks}}: IM in signed networks \cite{friend_foe_voter_model} considers the presence of an opinion, from a binary valued set containing exactly opposite opinions, at each node with the notion of \emph{friends} -- capable of activating a node with their own opinion and \emph{foes} -- capable of doing the opposite. The major limitation in this model is the stringent constraint on the parameters to attain values from a \emph{discrete binary} valued set. On the contrary, OI allows \emph{opinions} and \emph{interactions} to come from a real-valued domain. Moreover, the \emph{interaction} between two nodes is given a \emph{probabilistic} aspect which is missing in the signed network voter model. In fact, the signed network voter model is rather a \emph{special-case} of the OI model.

{\bf(3) \emph{Classical IM}}: Kempe et al. \cite{kempe} proved that finding an optimal solution for the IM problem is NP-Hard, and that a simple greedy ({\kempegreedy}) algorithm provides the best approximation guarantees in polynomial time. It required $O(kmnr)$ time to produce a solution, where $r$ is the number of MC simulations (usually $r\approx10K$). This high time complexity renders the {\kempegreedy} algorithm inapplicable to the networks of today. Since then, a host of works have introduced optimizations in {\kempegreedy} with CELF++ \cite{celfPlus} being the most efficient of all within this class of algorithms. Note that these optimizations never improved the asymptotic time complexity of the algorithm. In pursuit of better time complexity, researchers have recently resorted to techniques that use \emph{sampling} with \emph{memoization} \cite{opt1,tim,skim,imm}, to portray superior efficiency while retaining approximation guarantees. Among them, TIM$^{+}$ \cite{tim}, that runs in $O((k+l)(m+n)\log {n}/\epsilon^{2})$ expected time and produces a $(1-\frac{1}{e}-\epsilon)$-approximate solution, where $\epsilon$ is a constant, with probability as high as $1-{n^{-l}}$, and its improvement IMM \cite{imm} are the most efficient. Despite their superior efficiency these algorithms lacked \emph{scalability}, owing to their exorbitantly high memory footprint. The worst case space complexity of TIM$^{+}$ is $O(n^2\log{{n \choose k}}/\epsilon^{2})$, which can be very high for small values of $\epsilon$. In addition to the above mentioned techniques, literature has witnessed many heuristic algorithms. Among them, IRIE and SIMPATH \cite{irie,simpath} are considered state-of-the-art heuristics for the IC and LT models respectively. Although, these techniques build upon a similar idea as that of \emph{EaSyIM}, its algorithm design and analysis is very different from them with an additional advantage of running in linear space and time. Moreover, as shown in Sec.~\ref{sec:exp}, both of these techniques are not capable of scaling to larger datasets.
\ignore{
{\bf(1) \emph{Topic-Aware IM}} \cite{topic_aware_im1, topic_aware_im2, topic_aware_im3}: The motive of topic-aware IM is to design strategies for maximizing the collective \emph{spread} of information under a given topic-distribution. The most notable difference between \emph{topic-aware} and \emph{opinion-aware} IM is that the former allows different influence probabilities for each topic (assuming a joint distribution of topics), without considering the opinions (of a node) towards these topics. Moreover, neither does the former possesses a notion of \emph{negatively} activated nodes nor does it considers the \emph{change of opinions} as information propagates, hence rendering the \emph{spread} function submodular.
\\
{\bf(2) \emph{Competitive IM}} \cite{competitive_im1, competitive_im2, competitive_im3}: The objective of competitive IM is to maximize the \emph{spread} of information, about a specific content, in the presence of competitors. It is rather intuitive that competitive IM is similar to topic-aware IM, thus, it suffers with the same set of deficiencies as topic-aware IM, i.e., (a) the absence of opinions and (b) the incapability of capturing their change as information propagates. In addition, another limitation lies with the information-diffusion process, where each node (strictly) propagates information about a particular content and any newly activated node (by the former) also (strictly) follows the same content.

In summary, both topic-aware and competitive IM consider maximization of the number of active nodes (opinion-oblivious notion of \emph{spread}), where multiple topics are being propagated (for the former) or multiple competing sources are propagating information (for the latter). Moreover, the absence of a second layer capable of capturing the diffusion and change of opinion renders these problems closer to the \emph{classical} IM problem.
\\
{\bf(3) \emph{IM in signed networks}}: IM in signed networks \cite{friend_foe_voter_model} considers the presence of an opinion, from a binary valued set containing exactly opposite opinions, at each node with the notion of \emph{friends} -- capable of activating a node with their own opinion and \emph{foes} -- capable of doing the opposite. The major limitation in this model is the stringent constraint on all the parameters attaining values from a \emph{discrete} valued \emph{binary} set. On the contrary, the OI model allows \emph{opinions} and \emph{interactions} to come from a real-valued domain. Moreover, the \emph{interaction} between two nodes is given a \emph{probabilistic} aspect which is missing from the signed network voter model. In fact, the signed network voter model is rather a \emph{special-case} of the OI model. In summary, \emph{Opinion-Aware} IM under the OI model addresses most of the limitations discussed above and is, thus, generic and closer to the real-world scenarios.
\\
{\bf(4) \emph{Classical IM}}: Kempe et al. \cite{kempe} proved that finding an optimal solution for the IM problem is NP-Hard and were the first to prove that a simple greedy ({\kempegreedy}) algorithm can provide the best approximation guarantees in polynomial time. It required $O(kmnr)$ time to produce a solution, where $r$ is the number of MC simulations (usually, $r\approx10K$). This high time complexity renders the {\kempegreedy} algorithm inapplicable to the networks of today. Since then, a host of works have introduced optimizations in {\kempegreedy} with CELF++ \cite{celfPlus} being the most efficient of all within this class of algorithms. Note that these optimizations never improved the asymptotic time complexity of the algorithm. In pursuit of better time complexity, researchers have recently resorted to techniques that employ the use of \emph{sampling} and \emph{memoization} \cite{opt1,tim,skim,imm}, to portray superior efficiency while retaining approximation guarantees. Among them, TIM$^{+}$ (notation and details in \cite{tim}), that runs in $O((k+l)(m+n)\log {n}/\epsilon^{2})$ expected time and produces a $(1-\frac{1}{e}-\epsilon)$-approximate solution, where $\epsilon$ is a constant, with probability as high as $1-{n^{-l}}$, and its improvement IMM \cite{imm} can be considered to be the most efficient. Despite of their superior efficiency, these algorithms lacked on the \emph{scalability} aspect owing to their exorbitantly high memory footprint. The worst case space complexity of TIM$^{+}$ is $O(n^2\log{{n \choose k}}/\epsilon^{2})$, which can be very high for small values of $\epsilon$. In addition to the above mentioned techniques, literature has witnessed many heuristic algorithms \cite{mia, pmia, irie, ldag, simpath}. Among them, IRIE \cite{irie} and SIMPATH \cite{simpath} are considered state-of-the-art heuristics for the IC and LT models respectively. Although, these techniques build upon a similar idea as that of \emph{EaSyIM} the algorithm design and analysis for \emph{EaSyIM} is very different from them with an additional advantage of running in linear space and time. Moreover, as shown in Sec.~\ref{sec:exp} both of these techniques are not capable of scaling to larger datasets.
\ignore{
It has been shown in the Analysis section that EaSyIM algorithm, which is applicable in case of normal IC and WC models has theoretically the time complexity of $O(k(m+n)d)$. In this  section, we compare the expected value of spread achieved by EaSyIM with CELF++ and TIM algorithms, the state of the art algorithms for Independent cascade model. We adopt  the C++ implementation of the code for CELF++ and TIM that was made available by the authors. As a conventional practice, we run the monte carlo simulations 10,000 times to get expected value of spread.
Figure~\ref{fig:spreadNethept} clearly shows that the spread achieved by EasyIM is nearly the same as that of CELF++ and TIM's algorithm for the seeds varying from 1 to 100.  This clearly justifies the correctness of the algorithm devised for influence maximization. A similar plot for  the running time figure ~\ref{fig:timeNethept}  depicts the running time of the algorithms for different seeds.  CELF++ takes the maximum amount of time and EaSyIM takes the least amount of time. The running time (in sec) reported here are for running one single iteration of the influence maximization problem.\\\\
\textbf{Running time}
The time complexity of EaSyIM is linear and so increases linearly with the increase  in number of seeds. This is clearly visible from the figure ~\ref{dblpvsk}. Also the increase in time complexity with the increase in depth is also linear, keeping the memory footprint to be constant. 

Table~\ref{tab:resultsCELFPP} compares the running time and memory consumption of \emph{EaSyIM} with \emph{CELF++}. The results for the YouTube dataset using \emph{CELF++} have not been reported since it required $>200$ hours for $400$ iterations and was still not completed. Table~\ref{tab:resultsTIM} shows a similar comparison of \emph{EaSyIM} with \emph{TIM}. Although, TIM is $\approx 10$ times faster when compared to \emph{EaSyIM} for the DBLP dataset, we argue that this difference can be easily mitigated (as discussed in Sec.~\ref{sec:intro}) by running \emph{EaSyIM} in parallel on 10 cores while ensuring the memory gain to be the same. The memory consumed by \emph{TIM} was $> 100$GB for $\epsilon=0.1,\ k=50$ on YouTube, thus we could not report the results for the same.\\\\
\textbf{Memory footprint}
It is evident from the table  that the memory footprint of EaSyIM is 200 times less than the memory footprint of TIM's algorithm though being 8-9 times slower than TIM. An important fact to note here is that running TIM in parallel is not possible as it will consume way  more memory if run in parallel. On the other hand,  EaSyIM can be easily run in  parallel, as the 10K simulations can be run  in parallel on different cores and thus leading to a much improvement in running  time. Also, an interesting fact to note here is that EaSyIM takes the same time over 10K MC simulations as that  of TIM+, when run on 10 cores in parallel and still providing an improvement in memory by 10 times. This clearly justifies the scalability aspect of EaSyIM over TIM.
}
\ignore{
On running our ten simulation in parallel, EaSyIM consumed less than 6GB memory for SocLive Journal dataset which contains 5M edges. This memory consumption is easily  available on a commodity hardware. This encourages the scalability of EaSyIM whose efficiency can further be improved by running the monte carlo simulations in parallel as each run is independent. 
}

%
\begin{ignore}
{
\begin{theorem}
Given a node $u$, $Error_u(v) \leq p (Adj(v)-1)$, where $Error_u(v) = \gamma^*_u(v) - \gamma_u(v)$
\end{theorem}
\begin{corollary}
For DAG,  $Error_u(v) = 0$ i.e.  $\gamma^*_u(v) = \gamma_u(v) \implies \sigma(S) = \sigma(S^*)$ 
\end{corollary}
ADD UNDIRECTED ANALYSIS IN APPENDIX, for any path $>$ diameter and  contributes to error in score for all nodes
\begin{algorithm}
\caption{Score2($G(V,E), k, d, u$)}\label{Score2}
\begin{algorithmic}[1]
\REQUIRE Graph $G = (V,E)$,\#seeds $k=|S|$, depth $d$
\FOR{each $v \in V$}
\STATE $\gamma_u [v] = 0$
\ENDFOR
\FOR{each $i \in \{ 1, \ldots ,  d\}$}
    \FOR{each $v \in BFSlevl(i,u)$}
       \STATE  $\gamma_u[v] = \sum_{e=(w,v) \mid  w \in \vec{Adj(v)} } {\gamma_u[w]p(e)}$
    \ENDFOR
\ENDFOR

\FOR{each $v \in V$}
\STATE $Score1[u] = Score1[u] +o_v  \gamma_u[v] $
\ENDFOR

\STATE return $Score1[u]$

\end{algorithmic}
\end{algorithm}

\begin{theorem}
\label{thm:score2}
$\gamma[v] - Score1[v] \leq $
\end{theorem}
A higher error to a node which has more number of unions with him. And we are overcounting here. Also the fraction of error is  $< p^2$ for a score of 2p
}
\end{ignore}



Difference between Topic-aware IM and Opinion-Aware IM:
1) The former considers static distributions of influence probabilities across all the topics, between each pair of nodes. A user influences another user with a fixed-probability for a given topic. While the latter, associates the effect of topic on the personal user's opinion (strength + polarity). The influence of a user on someone, still depends upon the relationship across all different topics combined. We believe this is a more generic notion, as we give the user a choice of their personal opinion on that topic and then let the relationship between two people also define the way one will be influenced by the other.
The former does not have any notion about polarity/sentiment with the topic opinions.

2) Moreover, we consider change of opinions as the information propagates, which the topic-influence does not consider.

3) Although, We maximize the opinion about a topic instead of a collection of topics, we model a more generic thing where the idea is to maximize the net opinion and not the total number of people. Topic-Aware IM still tries to maximize the number of people.

4) One can never negatively influence a person in this scenario. As the meaning of a negative influence probability is not clear.

[To Do]: Download citations for Competitive IM.
Difference between Competitive IM and Opinion-Aware IM:
1) Node polarities are fixed in competitive IM, thus it is submodular as well. The idea is maximize info. from one source and minimize info from the other. 
2) The idea of assuming that if someone gets activated by a positive source will always spread positively, and vice-versa from the negative source, is too simplistic an assumption. This is the place we differ from Competitive IM.

Both topic-aware IM and competitive IM consider maximization of the num nodes activated, while multiple topics are considered (former) or multiple competing sources are propagating information (latter). There problem setting is a bit simple, as the nodes don't change their polarity as the information propagates for the competitive case and also for the topic-aware IM case the probabilites of influence for a topic by a node don't change.
\textbf{Methods with efficiency \& approximation guarantees:} In their seminal work, Kempe et al. \cite{kempe} introduced the influence maximization problem and proved that finding an optimal solution for this problem is NP-Hard. Moreover, they were the first to come up with an approximation algorithm with provable bounds on the information spread that could be achieved by any other algorithm in polynomial time. Although, this algorithm can still be safely attributed to possess the best possible approximation guarantees on the achieved spread, it cannot work well in practice given the humongous scale at which networks of today operate. Works have been done to improve the running time of these algorithms by introducing optimizations \cite{celf, celfPlus} while still maintaining the same spread guarantees. Both CELF and CELF++ portrayed significant improvements from the efficiency fronts, however still this was not enough.\\
\\
\textbf{Memoization based Techniques:} Given that the optimizations as discussed earlier were not enough, researchers resorted to memoization based techniques to come up with faster algorithms. These techniques pre-compute a lot of information at an offline pre-processing step and store this in the main-memory. The works belonging to these categories can though be attributed to being, \emph{efficient} -- one of the fastest existing algorithms in the literature, they cannot be deemed \emph{scalable} owing to their huge memory footprints. More recently, Borgs et al. \cite{opt1} came up with a $1-\frac{1}{\mathrm{e}} - \epsilon$, for any $\epsilon>0$, approximate algorithm which runs in $O((m+n){\epsilon}^{-3}\log n)$ time. Tang et al. \cite{tim} came up with an algorithm called \emph{TIM}, with better running time bound under the same approximation factor. Eventually, in these line of works \emph{SKIM} \cite{skim} is the most recent work, however its performance comparisons with \emph{SKIM} seem to be better on some grounds while being worse on many others. Thus, both \emph{TIM} and \emph{SKIM} can be safely assumed to be the current state of the art from an efficiency perspective.
\\
\textbf{Opinion-Aware Models:}
One of the first works in this space was was carried out by Chen et al. \cite{negIC}. They were the first to introduce a model incorporating the diffusion of negative opinions. They proposed a rather simplistic model where a negatively activated user retained its polarity throughout its lifetime and probability for a transition from positive polarity to negative polarity was assumed to be the same for each node. This simplistic setting allowed the model to be proven as submodular and efficient algorithms that scaled well in this setting were introduced. Li et al. \cite{plosone} came up with an influence maximization framework in signed networks. They proposed interesting techniques to maximize either positive or negative influence. However, the model proposed by them considered static opinions and didn't accommodate change of opinions with time. More recently, Zhang et al.\cite{ovm} proposed a model for information diffusion, according to  which the orientation of a user is influenced by the orientation of neighbouring nodes. Their work accommodates change of polarities with time and thus is closer to real-world scenarios. Thus this diffusion model didn't possess the nicer submodularity properties and a greedy solution could not be always referred to as the best solution. We introduce a more generic and realistic model when compared to \cite{ovm} and show scalable and efficient algorithms that scale well for large graphs.\\ 
\textbf{Heuristics and other Interesting Solutions:} There exist many other heuristic solutions for this problem \cite{pmia,simpath,irie,ldag,mia}. Moreover, apart from the classical influence maximization problem researchers have studied techniques to learn the influence edge probabilities \cite{goyal_vldb, im_learn_prob1, im_learn_prob2}. Apart from that, there has been introduction of newer areas like the influence skyline \cite{inf_skyline} and viral ad-targeting with minimum regret \cite{virad2}. Critical alert mining \cite{critical_alert_mining}.

\textbf{[Friday Night] Read Papers [Newly Downloaded Ones]:
	Try to fit them in Introduction or Related work. Cover all related work properly\\
Memoization Techniques:\\
    1) Started with works on MC pruning and removal.\\
        Static Greedy, IM Rank etc.\\
    2) Complete MC removal, randomized techniques with high prob. (SODA, TIM, SKIM and many more).
Efficient Exact Techniques:\\
    1) Kempe.\\
    2) CELF.\\
    3) CELF++.\\
Heuristics:\\
    1) All wei chen papers.\\
    2) SIMPATH for LT.
Models:\\
    1) NegIC.\\
    2) OC -- OVM.\\
Learning Edge Probability:\\
    1) VLDB and etc.\\
Interesting Applications:\\
    1) Influence Skyline.\\
    2) Ad-placement.\\
Write a line or two about all the things were we are better or are proposing a solution. In effect state that our solution is holistic.\\
}
}

\section{Conclusions and Future Work}
\label{sec:conc}
In this paper, we addressed the problem of influence maximization in social networks under a very generic opinion-aware setting, where the nodes can possess any one of the -- positive, neutral and negative opinions. To this end, we introduced the novel \emph{MEO} problem and devised a \emph{holistic} solution to the influence maximization problem; by coming up with an opinion-cum-interaction (\emph{OI}) model, and scalable algorithms -- \emph{OSIM} and \emph{EaSyIM}. Since majority of the works in the literature operate oblivious to the existence of opinions and are not scalable, there are efficiency concerns in real-world scenarios with huge graphs. Consequently, we designed efficient algorithms that run in time linear to the size of the graph. Moreover, the space complexity of our algorithms is also linear; which is orders of magnitude better when compared to the state-of-the-art techniques. Our empirical studies on real-world social network datasets showed that our algorithms are effective, efficient, scale well -- providing the best trade-off between running time and memory consumption, and are practical for large real graphs. In future, we would like to come up with a distributed version of our algorithms, thus enabling it scale to even larger graphs.\\
\\
\textbf{Acknowledgements:} We would like to thank Sayan Ranu, Arnab Bhattacharya, Srikanta Bedathur and Jeff Ullman for providing valuable suggestions and insights throughout the course of this work.

\bibliographystyle{abbrv}
\balance
\bibliography{im}

\begin{thebibliography}{10}

\bibitem{twitter}
http://an.kaist.ac.kr/traces/WWW2010.html.

\bibitem{churn_data}
http://lamda.nju.edu.cn/zhangt/dm2012/Project4.html.

\bibitem{arxiv}
{arXiv}.
\newblock http://www.arxiv.org.

\bibitem{empath}
{Empath: Sentiment, Topic, and Demographic Analysis}.
\newblock https://www.parc.com/services/focus-area/bigdata/.

\bibitem{text_processing}
{Natural Language Sentiment Analysis API}.
\newblock http://text-processing.com/demo/sentiment/.

\bibitem{snap}
{SNAP Datasets}.
\newblock https://snap.stanford.edu/data/.

\bibitem{boost}
{The Boost Graph Library}.
\newblock http://www.boost.org.

\bibitem{topic_aware_im2}
{\c{C}}.~Aslay, N.~Barbieri, F.~Bonchi, and R.~A. Baeza{-}Yates.
\newblock Online topic-aware influence maximization queries.
\newblock In {\em EDBT}, pages 295--306, 2014.

\bibitem{virad2}
C.~Aslay, W.~Lu, F.~Bonchi, A.~Goyal, and L.~V.~S. Lakshmanan.
\newblock Viral marketing meets social advertising: Ad allocation with minimum
  regret.
\newblock {\em PVLDB}, 8:814--825, 2015.

\bibitem{communityFormation1}
L.~Backstrom, D.~Huttenlocher, J.~Kleinberg, and X.~Lan.
\newblock Group formation in large social networks: Membership, growth, and
  evolution.
\newblock In {\em KDD}, pages 44--54, 2006.

\bibitem{topic_aware_im1}
N.~Barbieri, F.~Bonchi, and G.~Manco.
\newblock Topic-aware social influence propagation models.
\newblock In {\em ICDM}, pages 81--90, 2012.

\bibitem{im_community_detect2}
N.~Barbieri, F.~Bonchi, and G.~Manco.
\newblock Cascade-based community detection.
\newblock In {\em WSDM}, pages 33--42, 2013.

\bibitem{opt1}
C.~Borgs, M.~Brautbar, J.~Chayes, and B.~Lucier.
\newblock Maximizing social influence in nearly optimal time.
\newblock In {\em SODA}, pages 946--957, 2014.

\bibitem{reco3}
V.~Chaoji, S.~Ranu, R.~Rastogi, and R.~Bhatt.
\newblock Recommendations to boost content spread in social networks.
\newblock In {\em WWW}, pages 529--538, 2012.

\bibitem{topic_aware_im3}
S.~Chen, J.~Fan, G.~Li, J.~Feng, K.-l. Tan, and J.~Tang.
\newblock Online topic-aware influence maximization.
\newblock {\em PVLDB}, 8:666--677, 2015.

\bibitem{negIC}
W.~Chen, A.~Collins, R.~Cummings, T.~Ke, Z.~Liu, D.~Rincon, X.~Sun, Y.~Wang,
  W.~Wei, and Y.~Yuan.
\newblock {Influence Maximization in Social Networks When Negative Opinions May
  Emerge and Propagate}.
\newblock In {\em SDM}, 2011.

\bibitem{pmia}
W.~Chen, C.~Wang, and Y.~Wang.
\newblock Scalable influence maximization for prevalent viral marketing in
  large-scale social networks.
\newblock In {\em KDD}, pages 1029--1038, 2010.

\bibitem{mia}
W.~Chen, Y.~Wang, and S.~Yang.
\newblock Efficient influence maximization in social networks.
\newblock In {\em KDD}, pages 199--208, 2009.

\bibitem{ldag}
W.~Chen, Y.~Yuan, and L.~Zhang.
\newblock Scalable influence maximization in social networks under the linear
  threshold model.
\newblock In {\em ICDM}, pages 88--97, 2010.

\bibitem{similarity_learning}
Y.~Chen, E.~K. Garcia, M.~R. Gupta, A.~Rahimi, and L.~Cazzanti.
\newblock Similarity-based classification: Concepts and algorithms.
\newblock {\em Journal of Machine Learning Res}, 10:747--776, 2009.

\bibitem{im_rank}
S.~Cheng, H.~Shen, J.~Huang, W.~Chen, and X.~Cheng.
\newblock Imrank: influence maximization via finding self-consistent ranking.
\newblock In {\em SIGIR}, pages 475--484, 2014.

\bibitem{im_staticgreedy}
S.~Cheng, H.~Shen, J.~Huang, G.~Zhang, and X.~Cheng.
\newblock Staticgreedy: solving the scalability-accuracy dilemma in influence
  maximization.
\newblock In {\em CIKM}, pages 509--518, 2013.

\bibitem{skim}
E.~Cohen, D.~Delling, T.~Pajor, and R.~F. Werneck.
\newblock Sketch-based influence maximization and computation: Scaling up with
  guarantees.
\newblock In {\em CIKM}, pages 629--638, 2014.

\bibitem{virad1}
T.~N. Dinh, H.~Zhang, D.~T. Nguyen, and M.~T. Thai.
\newblock Cost-effective viral marketing for time-critical campaigns in
  large-scale social networks.
\newblock {\em IEEE/ACM Trans. Netw.}, 22:2001--2011, 2014.

\bibitem{disease}
S.~Eubank, H.~Guclu, V.~S. Anil~Kumar, M.~V. Marathe, A.~Srinivasan,
  Z.~Toroczkai, and N.~Wang.
\newblock {Modelling disease outbreaks in realistic urban social networks}.
\newblock {\em Nature}, 429(6988):180--184, 2004.

\bibitem{asim}
S.~Galhotra, A.~Arora, S.~Virinchi, and S.~Roy.
\newblock Asim: A scalable algorithm for influence maximization under the
  independent cascade model.
\newblock In {\em WWW (Companion Volume)}, 2015.

\bibitem{opinion}
H.~Gao, J.~Mahmud, J.~Chen, J.~Nichols, and M.~X. Zhou.
\newblock Modeling user attitude toward controversial topics in online social
  media.
\newblock In {\em ICWSM}, 2014.

\bibitem{celfPlus}
A.~Goyal, W.~Lu, and L.~V. Lakshmanan.
\newblock Celf++: Optimizing the greedy algorithm for influence maximization in
  social networks.
\newblock In {\em WWW (Companion Volume)}, pages 47--48, 2011.

\bibitem{simpath}
A.~Goyal, W.~Lu, and L.~V. Lakshmanan.
\newblock Simpath: An efficient algorithm for influence maximization under the
  linear threshold model.
\newblock In {\em ICDM}, pages 211--220, 2011.

\bibitem{competitive_im1}
X.~He, G.~Song, W.~Chen, and Q.~Jiang.
\newblock Influence blocking maximization in social networks under the
  competitive linear threshold model.
\newblock In {\em SDM}, pages 463--474, 2012.

\bibitem{irie}
K.~Jung, W.~Heo, and W.~Chen.
\newblock Irie: Scalable and robust influence maximization in social networks.
\newblock In {\em ICDM}, pages 918--923, 2012.

\bibitem{kempe}
D.~Kempe, J.~Kleinberg, and E.~Tardos.
\newblock Maximizing the spread of influence through a social network.
\newblock In {\em KDD}, pages 137--146, 2003.

\bibitem{parallel_im}
J.~Kim, S.-K. Kim, and H.~Yu.
\newblock Scalable and parallelizable processing of influence maximization for
  large-scale social networks?
\newblock In {\em ICDE}, pages 266--277, 2013.

\bibitem{diameter_degree}
J.~Leskovec, J.~Kleinberg, and C.~Faloutsos.
\newblock Graphs over time: Densification laws, shrinking diameters and
  possible explanations.
\newblock In {\em KDD}, pages 177--187, 2005.

\bibitem{celf}
J.~Leskovec, A.~Krause, C.~Guestrin, C.~Faloutsos, J.~VanBriesen, and
  N.~Glance.
\newblock Cost-effective outbreak detection in networks.
\newblock In {\em KDD}, pages 420--429, 2007.

\bibitem{competitive_im3}
H.~Li, S.~S. Bhowmick, J.~Cui, Y.~Gao, and J.~Ma.
\newblock Getreal: Towards realistic selection of influence maximization
  strategies in competitive networks.
\newblock In {\em SIGMOD}, pages 1525--1537, 2015.

\bibitem{friend_foe_voter_model}
Y.~Li, W.~Chen, Y.~Wang, and Z.-L. Zhang.
\newblock Influence diffusion dynamics and influence maximization in social
  networks with friend and foe relationships.
\newblock In {\em WSDM}, pages 657--666, 2013.

\bibitem{twitter_topic}
K.~W. Lim and W.~Buntine.
\newblock Twitter opinion topic model: Extracting product opinions from tweets
  by leveraging hashtags and sentiment lexicon.
\newblock In {\em CIKM}, pages 1319--1328, 2014.

\bibitem{competitive_im2}
W.~Lu, F.~Bonchi, A.~Goyal, and L.~V.~S. Lakshmanan.
\newblock The bang for the buck: fair competitive viral marketing from the host
  perspective.
\newblock In {\em KDD}, pages 928--936, 2013.

\bibitem{im_community_analysis1}
Y.~Mehmood, N.~Barbieri, F.~Bonchi, and A.~Ukkonen.
\newblock {CSI:} community-level social influence analysis.
\newblock In {\em ECML/PKDD}, pages 48--63, 2013.

\bibitem{submodular}
G.~Nemhauser, L.~Wolsey, and M.~Fisher.
\newblock An analysis of approximations for maximizing submodular set
  functions—i.
\newblock {\em Mathematical Programming}, 14(1):265--294, 1978.

\bibitem{im_prunedmc}
N.~Ohsaka, T.~Akiba, Y.~Yoshida, and K.~Kawarabayashi.
\newblock Fast and accurate influence maximization on large networks with
  pruned monte-carlo simulations.
\newblock In {\em AAAI}, pages 138--144, 2014.

\bibitem{twitter_opinion}
A.~Pak and P.~Paroubek.
\newblock Twitter as a corpus for sentiment analysis and opinion mining.
\newblock In {\em LREC}, 2010.

\bibitem{inf2}
M.~Richardson and P.~Domingos.
\newblock Mining knowledge-sharing sites for viral marketing.
\newblock In {\em KDD}, pages 61--70, 2002.

\bibitem{imm}
Y.~Tang, Y.~Shi, and X.~Xiao.
\newblock Influence maximization in near-linear time: A martingale approach.
\newblock In {\em SIGMOD}, pages 1539--1554, 2015.

\bibitem{tim}
Y.~Tang, X.~Xiao, and Y.~Shi.
\newblock Influence maximization: Near-optimal time complexity meets practical
  efficiency.
\newblock In {\em SIGMOD}, pages 75--86, 2014.

\bibitem{paths_valiant}
L.~G. Valiant.
\newblock The complexity of enumeration and reliability problems.
\newblock {\em SIAM Journal on Computing}, 8:410--421, 1979.

\bibitem{ovm}
H.~Zhang, T.~N. Dinh, and M.~T. Thai.
\newblock Maximizing the spread of positive influence in online social
  networks.
\newblock In {\em ICDCS}, pages 317--326, 2013.

\bibitem{reco2}
X.~W. Zhao, Y.~Guo, Y.~He, H.~Jiang, Y.~Wu, and X.~Li.
\newblock We know what you want to buy: A demographic-based system for product
  recommendation on microblogs.
\newblock In {\em KDD}, pages 1935--1944, 2014.

\bibitem{label_prop}
X.~Zhu and Z.~Ghahramani.
\newblock {Learning from labeled and unlabeled data with label propagation}.

\end{thebibliography}

\appendix

\ignore{
\section{Proof of Lemma 2}
\label{app_thm:lemma2}

{\small

\begin{proof}
Let us consider a bipartite graph $G(V,E)$ (Figure~\ref{fig:monotone}) containing two sets of nodes $X$ and $Y=V\setminus X$. The set $X$ contains $n_x=|X|$ nodes while the set $Y$ contains $n_y=|Y|, n_y \geq 2n_x$ nodes. Moreover, let us assume that the opinions of the nodes $\forall x_i \in X$ to be $o_{x_i}=+1$, while that of nodes $\forall y_j \in Y$ to be $o_{y_j}=0$. There exist directed edges from each node $x_i \in X$ to two consecutive nodes $y_{2i-1},y_{2i} \in Y$. Further, the influence probabilities are initialized as $p_{(u,v)}=1, \forall (u,v) \in E$ and the interaction probabilities $\varphi$ associated with all but the last two edges is $1$, which attain a value of $0$. Mathematically, $\varphi_{x_i,y_j}=1 \mid 1\leq i\leq n_x-1, j=2i-1$ or $j=2i$ while $\varphi_{x_i,y_j}=0 \mid i=n_x, j=2i-1$ or $j=2i$. An analysis of the properties of the spread function $\Gamma(\cdot)$, assuming the above construction, is presented next.

Consider a scenario where the set of seeds $S$ is initialized with a single node $x_i \mid 1 \leq i \leq n_x-1$. With this seed, both of the nodes $y_{2i-1}$ and $y_{2i}$ will get activated with final opinion $+1/2$. Hence, the effective spread $\Gamma^o(S) = 2\times(1/2) = +1$. Now if the node $x_{|X|}$ is added to the set $S$, the nodes $y_{2|X|-1}$ and $y_{2|X|}$ will also get activated. The final opinion of these newly activated nodes will be $-1/2$ each. The effective spread of the updated seed set is $\Gamma^o(S) = 1 - 1 = 0$. On adding another node $x_j \mid j\neq i$ and $1\leq j\leq n_x-1$ to the set $S$, the nodes $y_{2j-1}$ and $y_{2j}$ will become active with a final opinion of $1/2$ each. Thus, the effective spread is now $\Gamma^o(S) = 0+1 = 1$. In the above case it can be seen that the effective spread varied from $1\rightarrow0\rightarrow1$ on subsequent additions of nodes to the seed set. This clearly shows that the opinion spread function $\flow^o(S)$ is neither monotone nor submodular.
\hfill{}
\end{proof}	
}

\section{Proof of Theorem 1}
\label{app_thm:theorem1}

{\small

\begin{proof}
We show that the classical set cover problem can be decided in polynomial time if an approximation, with a constant ratio, for MEO exists in polynomial time. To this end, we reduce the set cover problem to an instance of MEO. Given a set of elements $Q = \{ q_1,\ldots,q_n\}$ and a collection of subsets $R = \{R_1,\ldots,R_m\}$ where $R_i \subseteq Q, \forall i \in \{1,\ldots,m\}$, the set-cover decision problem returns true if $\exists C\subseteq R; |C|=k$ and $\cup C_j  = Q,\ \forall C_j \in C$. 

Now given the set-cover problem, we construct a graph as shown in Figure~\ref{fig:submod} to obtain an instance of MEO by adding three layers of nodes in addition to a sink node. For each subset $R_i \in R$, a node $x_i$ is added in the first layer of the graph with opinion $o_{x_i} = 0$. In the second  layer, for each element $q_i\in Q$, we add a node $y_i$ with opinion $o_{y_i}=\frac{1}{n}$. We add a third layer containing $m + n -2$ nodes denoted by $z_i$ with opinion $o_{z_i} =-\frac{1}{2n} $. Along with this, we add a sink node $s$ with opinion $o_s = -1+\frac{1}{n}$. Now, a directed edge  $(x_i,y_j)$ is  added iff $q_i \in R_j$. For each node $y_i$ in the second layer, edges $(y_i,z_j), \forall j \in \{1,2,\ldots,m+n-2\}$ are added in the graph. To finalize this construction, we add the edge $(z_i,s), \forall i \in \{1,2,\ldots,m+n-2 \}$ . 

For the IC model, all the edges $(u,v) \in E$ are assigned $p_{(u,v)} = 1$ and $ \varphi_{(u,v)} = 1$. Moreover, the LT model uses the same activation threshold for each node $\theta_u = 1$. It can be seen that for both these models, the seeds should be chosen from the first layer (any of the $x_i$'s), because they will activate $y_j$'s which in turn will activate any of the $z_i$'s and thus, $s$ would always be activated.

We instantiate MEO with $\lambda = 1$. Without loss of generality, assuming $\exists$ a set-cover of size $k$, $C = \{x_1,x_2,\ldots,x_k \}$, then choosing these nodes as seeds ensures the maximum spread. The final opinion of $y_i$, $o_{y_i}' = (0+\frac{1}{n})/2=\frac{1}{2n}$. Similarly, the final opinion of $z_i$, $o_{z_i}' = \frac{1}{2n}-\frac{1}{2n}=0$ and the final opinion of $s$, $o_{s}'=(0-1+\frac{1}{n})/2=-\frac{1}{2}+\frac{1}{2n}$. Hence, the spread = $n\frac{1}{2n} + 0  - \frac{1}{2} + \frac{1}{2n} = \frac{1}{2n}>0$. When a set cover of size $k$ does not exist, the maximum spread achieved is $|Q'|\frac{1}{2n} + 0   -\frac{1}{2} + \frac{1}{2n}  \leq (n-1)\frac{1}{2n} - \frac{1}{2} + \frac{1}{2n} =0$, where $Q', |Q'|\leq n-1$ is the maximum number of elements covered by the $k$ chosen sets. Hence, the maximum spread achieved when the set-cover does not exist is $0$.

If there is a polynomial time algorithm which approximates MEO within a constant ratio, then we can decide the set-cover problem in polynomial time using the above mentioned reduction. In other words, if an approximate algorithm gives a spread $\leq0$ on the reduced graph, then a set-cover does not exist while a set-cover exists if the spread $>0$. This renders the set-cover problem decidable in polynomial time,  which means P=NP. Therefore, approximating MEO within a constant ratio is NP-hard.
\hfill{}
\end{proof}	
}

\section{Proof of Lemma 3}
\label{app_thm:lemma3}

{\small

\begin{proof}
Using Kempe's ~\cite{kempe} analysis of the IC model it can be seen that,
\begin{align}
	\label{ic}
	 \sigma(&A\cup B) = \sum_{X} P(X) \sigma_X(A\cup B) \nonumber\\
    	&= \sum_{X} P(X) \big( \sigma_X(A) + \sigma_X(B) - \sigma_X(A\cap B) \big) \nonumber\\
	&= \sum_{X} P(X) \sigma_X(A)  + \sum_{X} P(X) \big(\sigma_X(B)-\sigma_X(A\cap B) \big) \nonumber\\
	&= \sigma(A) + \sigma_{\sigma(A)}(B).
\end{align}
A similar analysis can be done for the LT model using the reduction to the \emph{live-edge} model \cite{kempe}. Under the live-edge model each node possesses a single (live) incoming edge, thus, a node $v$ is activated via a path either from a node in $A$ or in $B$. 
\begin{align}
	\label{lt}
	\sigma(A\cup B) &= \sum_{X} P(X) \sigma_X(A\cup B) \nonumber\\
 	   &= \sum_{X} P(X) \big( \sigma_X(A) + \sigma_X(B) \big) \nonumber\\
	&= \sum_{X} P(X) \sigma_X(A)  + \sum_{X} P(X) \sigma_X(B)\nonumber
\end{align}
Since, $\sigma(A)$ and $\sigma(B)$ are disjoint,%
\begin{align}
	\sigma(A\cup B) &= \sigma(A) + \sigma_{\sigma(A)}(B).
\end{align}
\hfill{}
\end{proof}	
}

\begin{figure}[t]
	\centering
	\subfloat[Submodularity]
	{
		\scalebox{0.39}{
		\begin{tikzpicture}[node distance=15mm,
		round/.style={fill=green!50!black!20,draw=green!50!black,text width=10mm,align=center,circle,thick}]
		\centering
		\node at (0,2.5) {\huge $\mathbf{o_s = 0}$};
		\node at (3.2,2.5) {\huge$\mathbf{o_t = 1}$};
		\node at (0,-7.5) {\huge Layer 1};
		\node at (3.3,-7.5) {\huge Layer 2};

		\node[round, scale=0.9, label={\huge $\mathbf{s_1}$}] (u1) {};
		\node[round, scale=0.9, label={\huge $\mathbf{s_{n_s-1}}$}] (u2) [below = 2.4cm of u1] {};
		\node[round, scale=0.9, label={\huge $\mathbf{s_{n_s}}$}] (um) [below =0.7cm of u2] {};

		\node[round, scale=0.5, minimum size=5pt, label={\LARGE $\mathbf{t_1}$}] (v1) [above right =.13cm and 2.7cm of u1 ] {};
		\node[round, scale=0.5, label={\LARGE $\mathbf{t_2}$}] (v2) [below = .6cm of v1] {};
		\node[round, scale=0.5, label={\LARGE $\mathbf{t_{2n_s-3}}$}] (v3) [below = 1.8cm of v2] {};
		\node[round, scale=0.5, label={\LARGE $\mathbf{t_{2n_s-2}}$}] (v4) [below = .5cm of v3] {};
		\node[round, scale=0.5, label={\LARGE $\mathbf{t_{2n_s-1}}$}] (v5) [below = 0.5cm of v4] {};
		\node[round,scale = 0.5, minimum size=5pt, label={\LARGE $\mathbf{t_{2n_s}}$}] (v6) [below =.5cm of v5] {};

		\draw[thick] (0,-2.4) ellipse (1.1cm and 4.4cm);
		\draw[thick] (3.3,-2.4) ellipse (1.1cm and 4.4cm);

		\draw[loosely dotted,thick, line width=2.5pt](0,-1.2) -- (0,-1.8);
		\draw[loosely dotted,thick, line width=2.5pt](3.3,-1.2) -- (3.3,-1.8);

		\path [->, -triangle 45,line width=1.4pt] (u1) edge (v1)
		edge (v2)
		(u2) edge (v3)
		(u2) edge (v4)
		(um) edge (v5)
		(um) edge (v6)
		; 
		\end{tikzpicture}
		}
		\label{fig:monotone}
	}
	\subfloat[Tractability]
	{
		\scalebox{0.35}{
			\begin{tikzpicture}[node distance=15mm,
			round/.style={fill=green!50!black!20,draw=green!50!black,minimum size=5mm,text width=10mm,align=center,circle,thick, line width=1.6pt}]
			\centering
			\node at (0,2) {\huge $\mathbf{o_x = 0}$};
			\node at (3.1,2) {\huge$\mathbf{o_y = \frac{1}{n}}$};
			\node at (6.2,2) {\huge$\mathbf{o_z = -\frac{1}{2n}}$};
			\node at (9.7,-4.5) {\huge$\mathbf{o_s = \frac{1}{n}-1}$};
			\node at (9.7,-3.5) {\huge \textbf{sink}};

			\node at (0,-6) {\huge\textbf{ Layer 1}};
			\node at (3.3,-6) {\huge \textbf{Layer 2}};
			\node at (6.6,-6) {\huge \textbf{Layer 3}};

			\node[round] (u1) {\huge $\mathbf{x_1}$};
			\node[round] (u2) [below of=u1] {\huge $\mathbf{x_2}$};
			\node[round] (um) [below =1cm of u2] {\huge $\mathbf{x_m}$};

			\node[round] (v1) [right =2cm of u1] {\huge $\mathbf{y_1}$};
			\node[round] (v2) [below of=v1] {\huge$\mathbf{y_2}$};
			\node[round] (vm) [below =1cm of v2] {\huge$\mathbf{y_n}$};

			\node[round] (x1) [right =2cm of v1] {\huge$\mathbf{z_1}$};
			\node[round] (x2) [below of=x1] {\huge$\mathbf{z_2}$};
			\node[round] (xm) [below =1cm of x2] {\huge$\mathbf{z_{m+n-2}}$};

			\node[round] (s) [rectangle, below right=.3cm and 2cm of x2] {\huge $\mathbf{s}$};

			\draw[thick] (0,-2) ellipse (1cm and 3.5cm);
			\draw[thick] (3.3 ,-2) ellipse (1cm and 3.5cm);
			\draw[thick] (6.6,-2.) ellipse (1cm and 3.5cm);

			\draw[loosely dotted,thick, line width = 1.6pt](0,-2.5) -- (0,-3);
			\draw[loosely dotted,thick, line width = 1.6pt](3.3,-2.5) -- (3.3,-3);
			\draw[loosely dotted,thick, line width = 1.6pt](6.6,-2.5) -- (6.6,-3);

			\path [->, ,-triangle 45,line width=1.6pt] (u1) edge (v1)
			edge (v2)
			edge (vm) 
			(u2) edge (v1)
			edge (vm)
			(um) edge (v1)
			edge (vm)
			(v1)edge (x1)
			edge (x2)
			edge (xm)
			(v2)edge (x1)
			edge (x2)
			edge (xm)
			(vm)edge (x1)
			edge (x2)
			edge (xm)
			(x1)edge (s)
			(x2)edge (s)
			(xm)edge (s)
			; 
			\end{tikzpicture}
		}
		\label{fig:submod}
	}
	\figcaption{Reductions for the submodularity and tractability analysis of MEO.}
	\label{fig:reductions_meo}
\end{figure}

\section{Proof of Lemma 4}
\label{app_thm:lemma4}

{\small

\begin{proof}
If $k = 1$,\ $\sigma (S_k) = \sigma(S_k^*)$ because {\ourgreedy}$(G,k,l)$ chooses the node with the maximum score. Since the score assigned to each node captures the expected value of its spread, the seed-node selected by the former is the same as that selected by the {\kempegreedy} algorithm. Let us assume that $\sigma (S_k) = \sigma(S_k^*)$ holds for $k=i$. Now for $k = i+1$, 
\begin{align*}
\sigma (S_{i+1}^*) &= \sigma(S_{i+1}^*) - \sigma(S_{i}^*) + \sigma(S_{i}^*)
\end{align*}
Paritioning $S_{i+1}$ as $S_i \cup \{u_{i+1}\}$, Lemma~\ref{thm:ICLT} yields the following,
\begin{align*}
                   &= \sigma_{\sigma(S^*_i)}(\{u_{i+1}\}) + \sigma(S_i^*)\\
                   &= \sigma_{\sigma(S_i)}(\{u_{i+1}\}) + \sigma(S_i) \\
	        &= \Delta^l(\{u_{i+1}\}) + \sigma(S_i)\\
                   &= \sigma(S_{i+1}).
\end{align*}

Since $\sigma(S_{i+1}) = \sigma(S_{i+1}^*)$, by the principle of mathematical induction $\sigma(S_k) = \sigma(S_k^*)$ holds $\forall k$.
\hfill{}
\end{proof}	
}

\section{Proof of Lemma 5}
\label{app_thm:lemma5}

{\small

\begin{proof}
The exact contribution of a node $v$ to the score of another node $u$ is
\begin{align}
\label{eq:exact_contri}
	\gamma^*_v(u)= \sum_{w \in \In(u)}\Big(p_{(w,v)}\bigcup_{\rho \in \pathset_{uw}} \prod_{e\in \rho}p_e \Big).
\end{align}
Similarly, the contribution of $v$ to the score of $u$ using the $PU$ algorithm is
\begin{align}
\label{eq:approx_contri}
	\gamma_v(u)=\bigcup_{\rho \in \pathset_{uv}} \prod_{e\in \rho}p_e.
\end{align}
Using the standard inclusion-exclusion principle,
\begin{flalign*}
\label{eq:exact_inc_exc}
	\bigcup_{\rho \in \pathset_{uv}} \prod_{e\in \rho}p_e = \bigg(& \sum_{\rho \in \pathset_{uv}}\prod_{e\in \rho}p_e - \sum_{\rho' \in \pathset_{uv}}\sum_{\rho \in \pathset_{uv}\setminus\{\rho'\}} \prod_{e\in \rho} p_e \prod_{e' \in \rho'}p_{e'}\nonumber\\
	&+ \ldots + (-1)^{|\pathset_{uv}|-1}\prod_{\rho \in \pathset_{uv}} \prod_{e \in \rho} p_e \bigg).
\end{flalign*}
Thus, Eq.~\ref{eq:exact_contri} and~\ref{eq:approx_contri} can be written as follows.
\begin{flalign*}
\gamma^*_v(u)=\sum_{w \in \In(u)}\Big(p_{(w,v)}\big(A_1 - A_2 + \ldots + (-1)^{t-1}A_{t^*} \big)\Big).	
\end{flalign*}
\begin{flalign*}
\gamma_v(u)=B_1 - B_2 + \ldots + (-1)^{t-1}B_t.	
\end{flalign*}
where $t^*=|\pathset_{uw}|$, $t=|\pathset_{uv}|$ and $A_i, i\leq t^*$, $B_j,j\leq t$ stand for the larger terms in the expansion of $\bigcup_{\rho \in \pathset_{uw}}\prod_{e\in \rho} p_e$ and $\bigcup_{\rho \in \pathset_{uv}}\prod_{e\in \rho} p_e$ respectively.\\
\\
Assuming $A_2<<A_1$ and henceforth $B_2<<B_1$, all the higher order terms can be neglected (significance of this assumption is described in Sec.~\ref{subsubsec:discussion}). Moreover, the relationship between $A_1$ and $B_1$ is defined as follows.
\begin{align*}
\sum_{w \in \In(v)} p_{(w,v)}A_1 &= \sum_{w \in \In(v)} p_{(w,v)}\big(\sum_{\rho \in \pathset{uw}}\prod_{e \in \rho} p_e\big)\\
	&=\sum_{\rho \in \pathset{uv}}\prod_{e \in \rho} p_e = B_1.
\end{align*}
Neglecting the summation over the pair-wise product of the paths between different intermediary nodes $w_i,w_j, \forall(i,j) \mid w_i, w_j \in \In(v)$, we devise a relationship between $B_2$ and $A_2$.
\begin{align*}
B_2&=\sum_{\rho \in \pathset_{uv}}\sum_{\rho'  \in (\pathset_{uv}\setminus\{\rho\})}\prod_{e \in \rho} p_e \prod_{e' \in \rho'} p_e'\\ 
&\geq \sum_{w \in \In(v)} p_{(w,v)}^2 \sum_{\rho \in \pathset_{uw}}\sum_{\rho'  \in (\pathset_{uw}\setminus\{\rho\})}\prod_{e \in \rho} p_e \prod_{e' \in \rho'} p_e'\\
&= \sum_{w \in \In(v)} p_{(w,v)}^2 A_2.
\end{align*}
The relative error $\epsilon^{DAG}_1$, is now defined as,
\begin{align}
\epsilon^{DAG}_1 &= \frac{|\gamma^*_v(u) - \gamma_v(u)|}{\gamma^*_v(u)} \approx\frac{B_2 - \sum\limits_{w \in \In(v)} p_{(w,v)}A_2}{\sum\limits_{w \in \In(v)} p_{(w,v)}(A_1-A_2)} \nonumber\\
Since&,\ A_2<<A_1,\ A_2<A_1^2\nonumber\\
	&\approx\frac{\sum\limits_{w \in \In(v)}p_{(w,v)}(p_{(w,v)}-1)A_2}{\sum\limits_{w \in \In(v)}p_{(w,v)}A_1}\nonumber\\
	&\leq\sum\limits_{w \in \In(v)}\big((p_{(w,v)}-1)A_1\big).
\end{align}
\ignore{
\begin{flalign}
	\Rightarrow \sum_{\rho in P} \prod_{e \in \rho} p_e  - \sum_{w \in In(u)}\left(p_{(w,v)}\big(   \sum_{\rho' \in P}\sum_{\rho \in P\setminus\{\rho'\}} \prod_{e\in \rho} p_e \prod_{e' \in \rho'}p_{e'} \big) \right)\\
               &\geq \sum_{\rho in P} \prod_{e \in \rho} p_e  - \sum_{w \in In(u)}\left(p_{(w,v)}\big(   \sum_{\rho' \in P}\sum_{\rho \in P\setminus\{\rho'\}} \prod_{e\in \rho} p_e \prod_{e' \in \rho'}p_{e'} \big) \right)
\end{flalign}
}
\hfill{}
\end{proof}	
}

\section{Proof of Lemma 6}
\label{app_thm:lemma6}

{\small

\begin{proof}
Using the same nomenclature as Lemma~\ref{thm:path_union_dags}, the contribution of $v$ in the score of $u$ using the $PU$ and the \emph{EaSyIM} algorithm is stated as follows
\begin{align}
	\gamma_v^{PU}(u)&=\sum_{i=1}^{t} (-1)^{i-1}B_i 	\approx B_1 - B_2.\\
	\gamma_v^{EaSyIM}(u)&=B_1.
\end{align}
The relative error (w.r.t $PU$) $\epsilon^{DAG}_2$, can now be stated as follows
\begin{align}
	\epsilon^{DAG}_2=\frac{B_2}{B_1 - B_2} \approx \frac{B_2}{B_1} \leq B_1.
\end{align}
\hfill{}
\end{proof}
}
}
\begin{figure*}[t]
\centering
	\subfloat[DBLP and YouTube]
	{
		\scalebox{0.185}{
			\includegraphics[width = 0.99\linewidth]{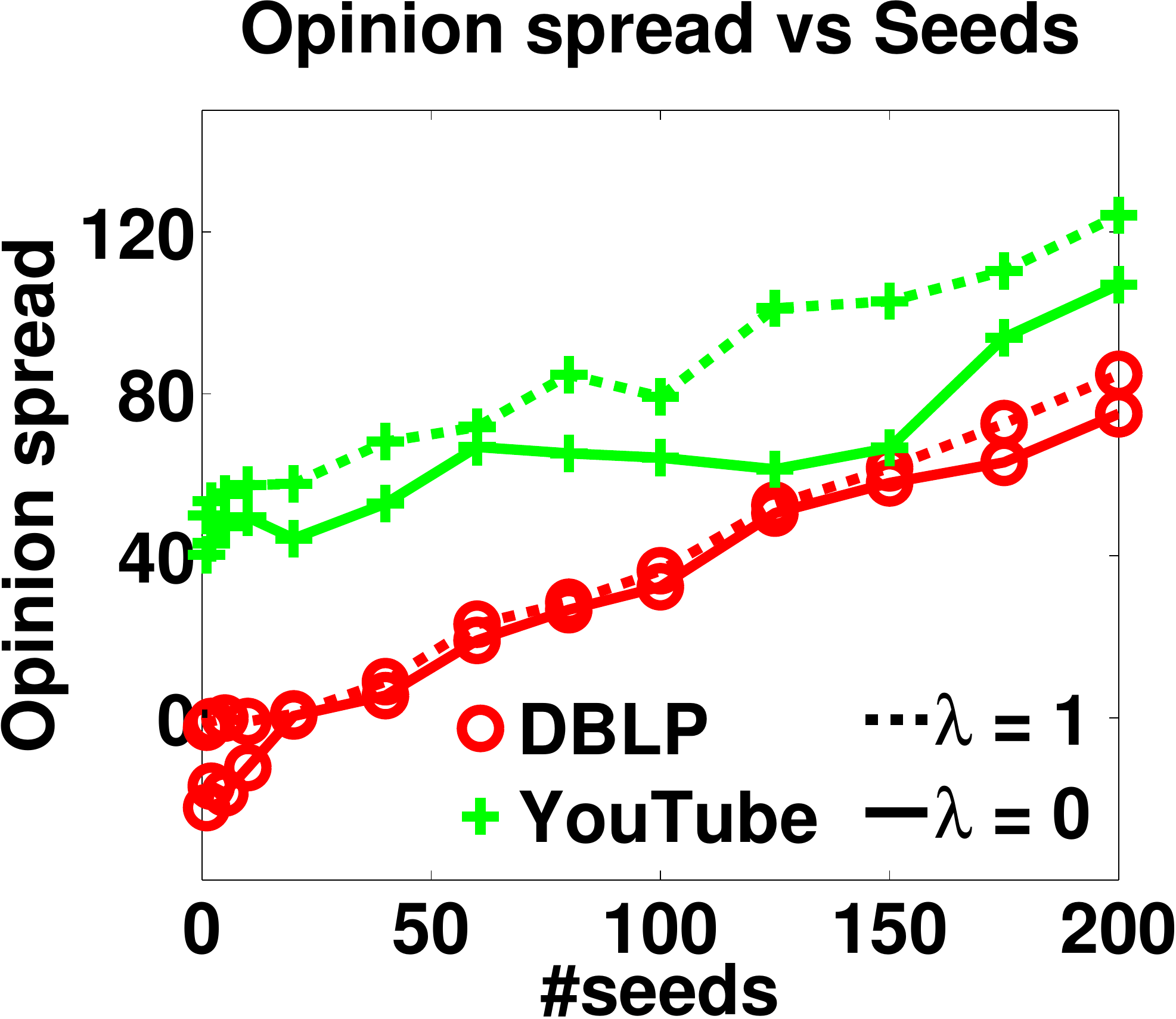}
		}
		\label{fig:dblp_youtube_lambda}
	}
	\subfloat[HepPh]
	{
		\scalebox{0.185}{
			\includegraphics[width = 0.99\linewidth,trim=0cm 0cm 0cm 0cm]{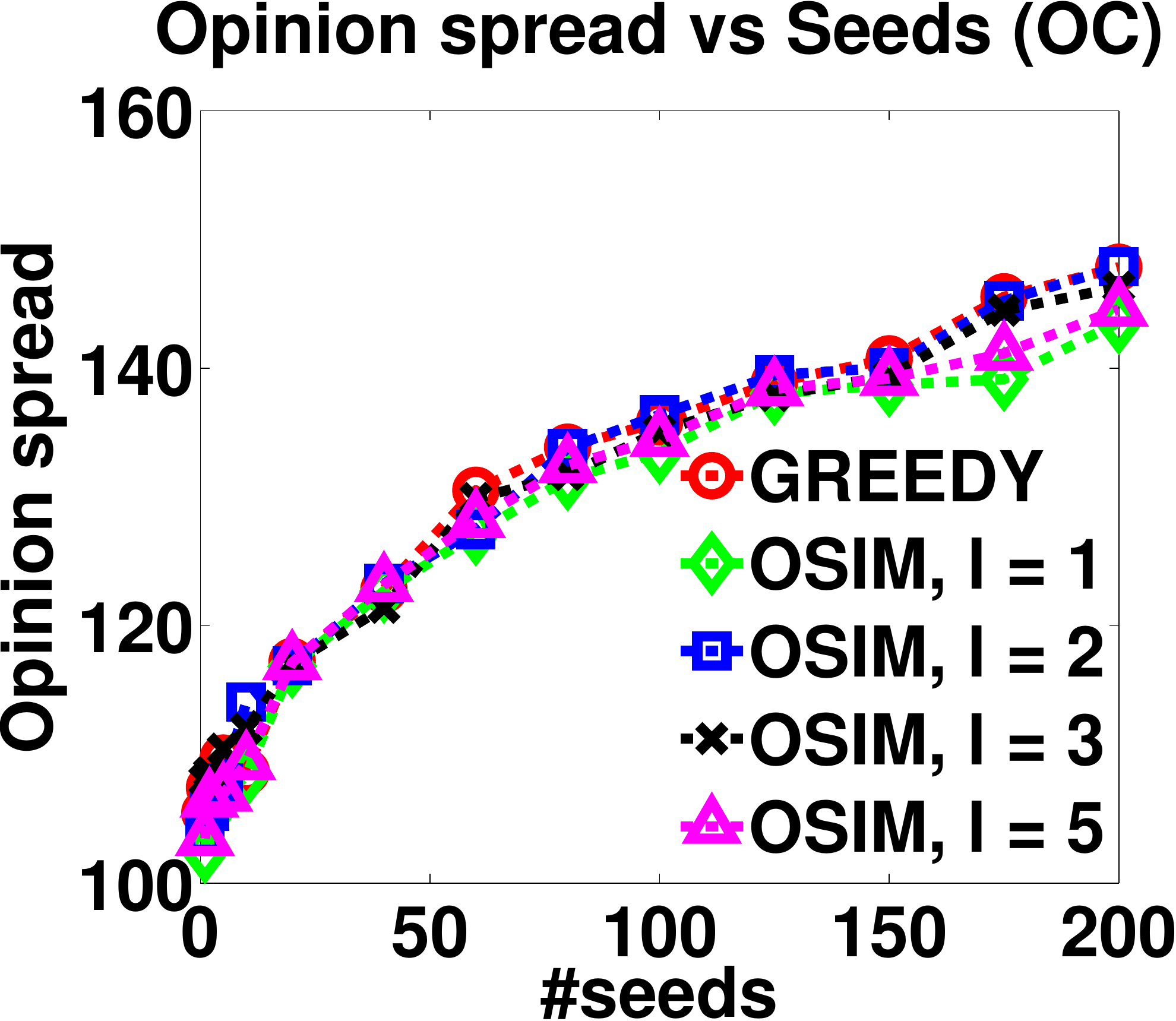}
		}
		\label{fig:eop_hepph_oc}
	}
	\subfloat[DBLP and YouTube]
	{
		\scalebox{0.1855}{
			\includegraphics[width = 0.99\linewidth]{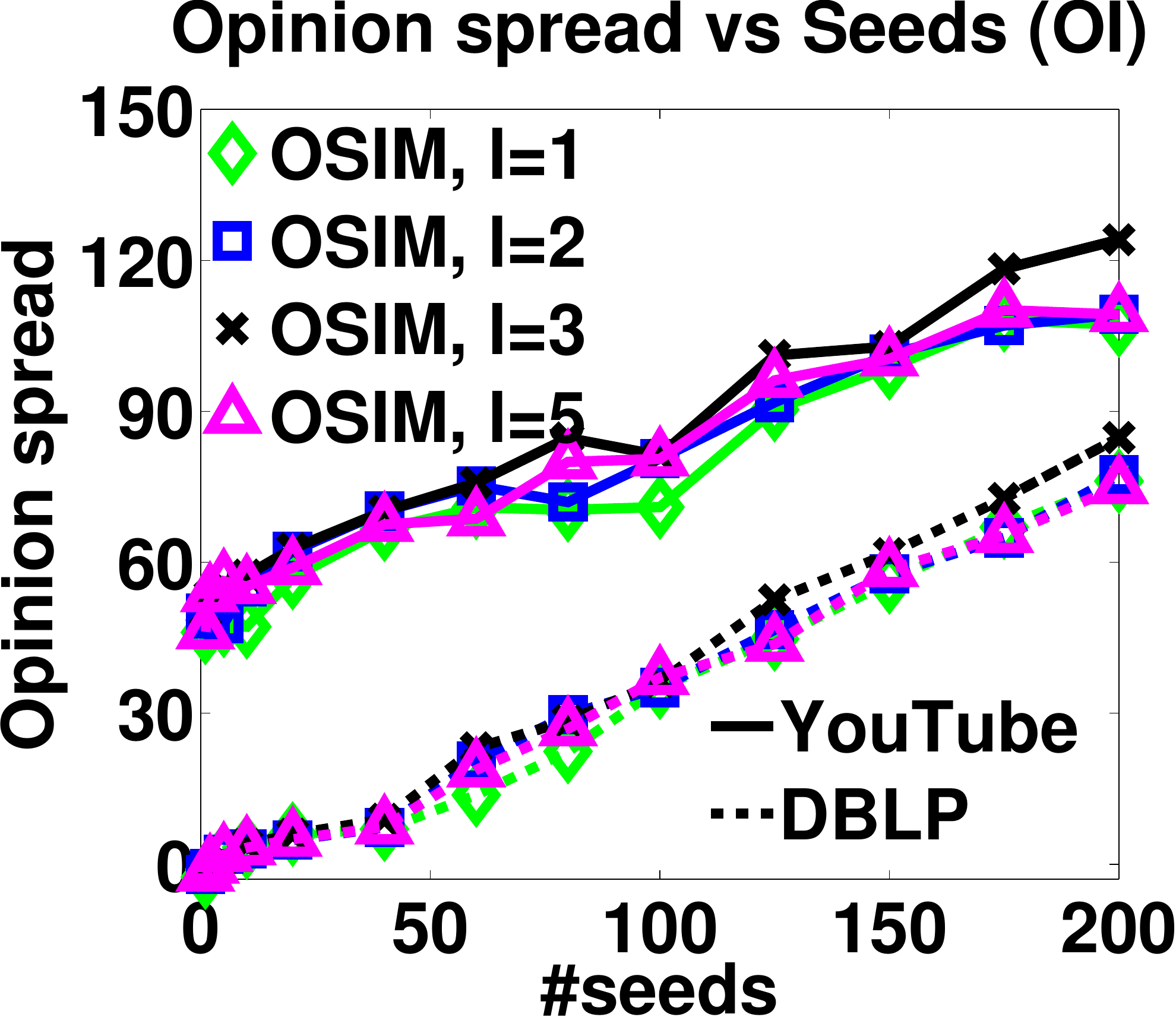}
		}
		\label{fig:eop_dblp_youtube_oi}
	}
	\subfloat[NetHEPT]
	{
		\scalebox{0.1855}{
			\includegraphics[width = 0.99\linewidth]{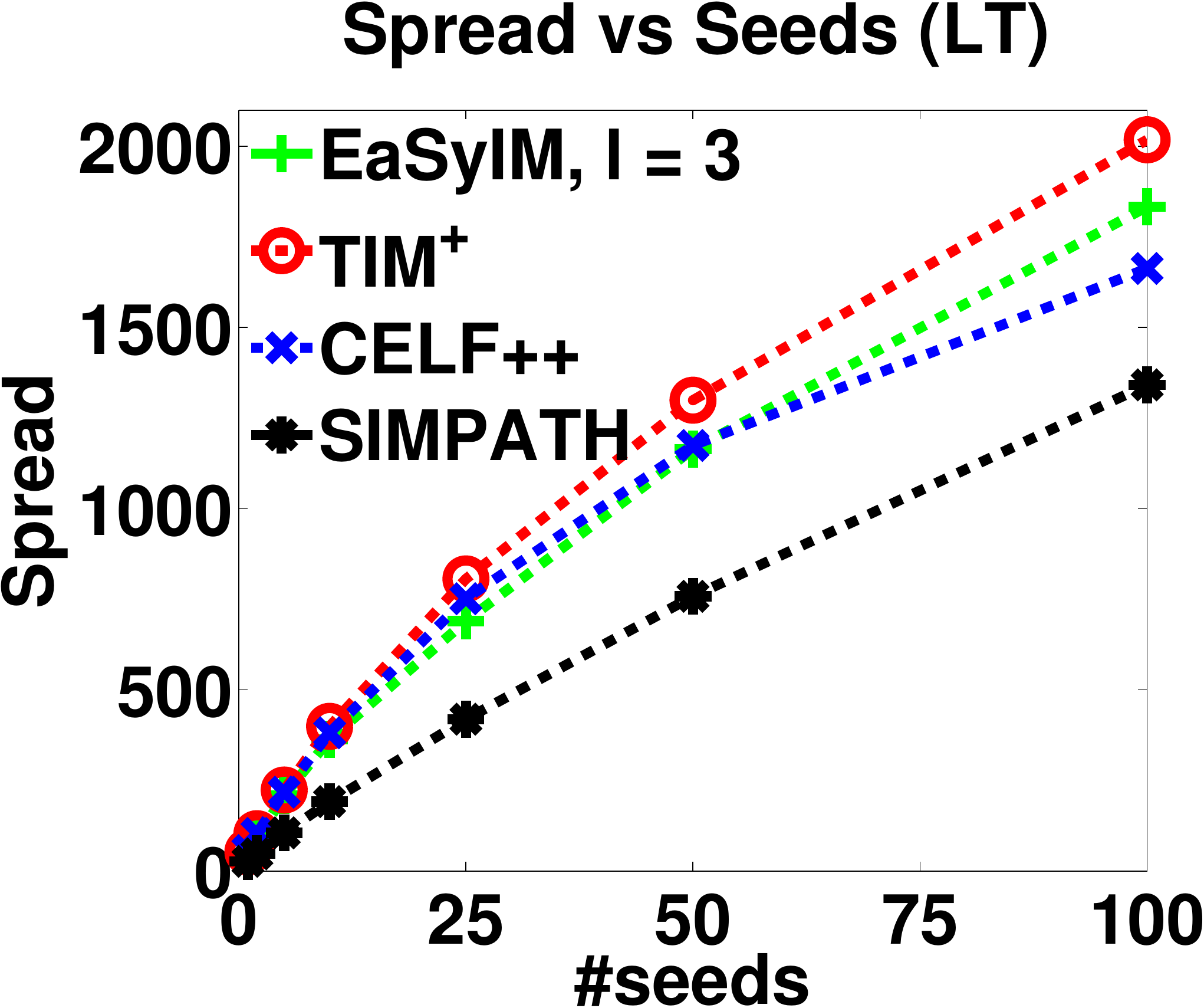}
		}
		\label{fig:spread_nethept_lt_compare}
	}
	\subfloat[YouTube]
	{
		\scalebox{0.1855}{
			\includegraphics[width = 0.99\linewidth]{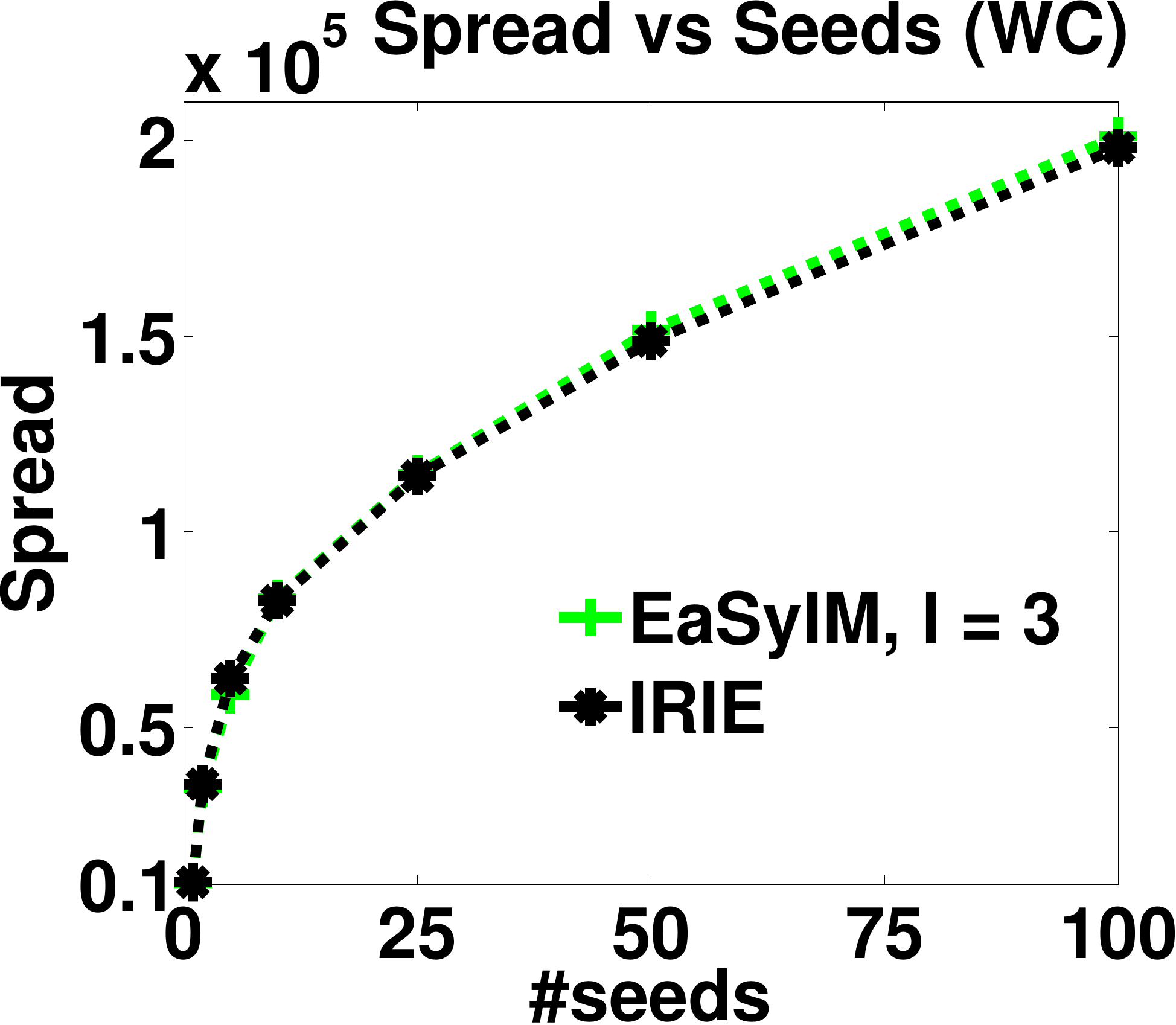}
		}
		\label{fig:spread_youtube_wc_compare}
	}
	\\
	\vspace{-3.5mm}	
	\subfloat[HepPh]
	{
		\scalebox{0.1855}{
			\includegraphics[width = 0.99\linewidth, trim=0cm 0cm 10cm 0cm]{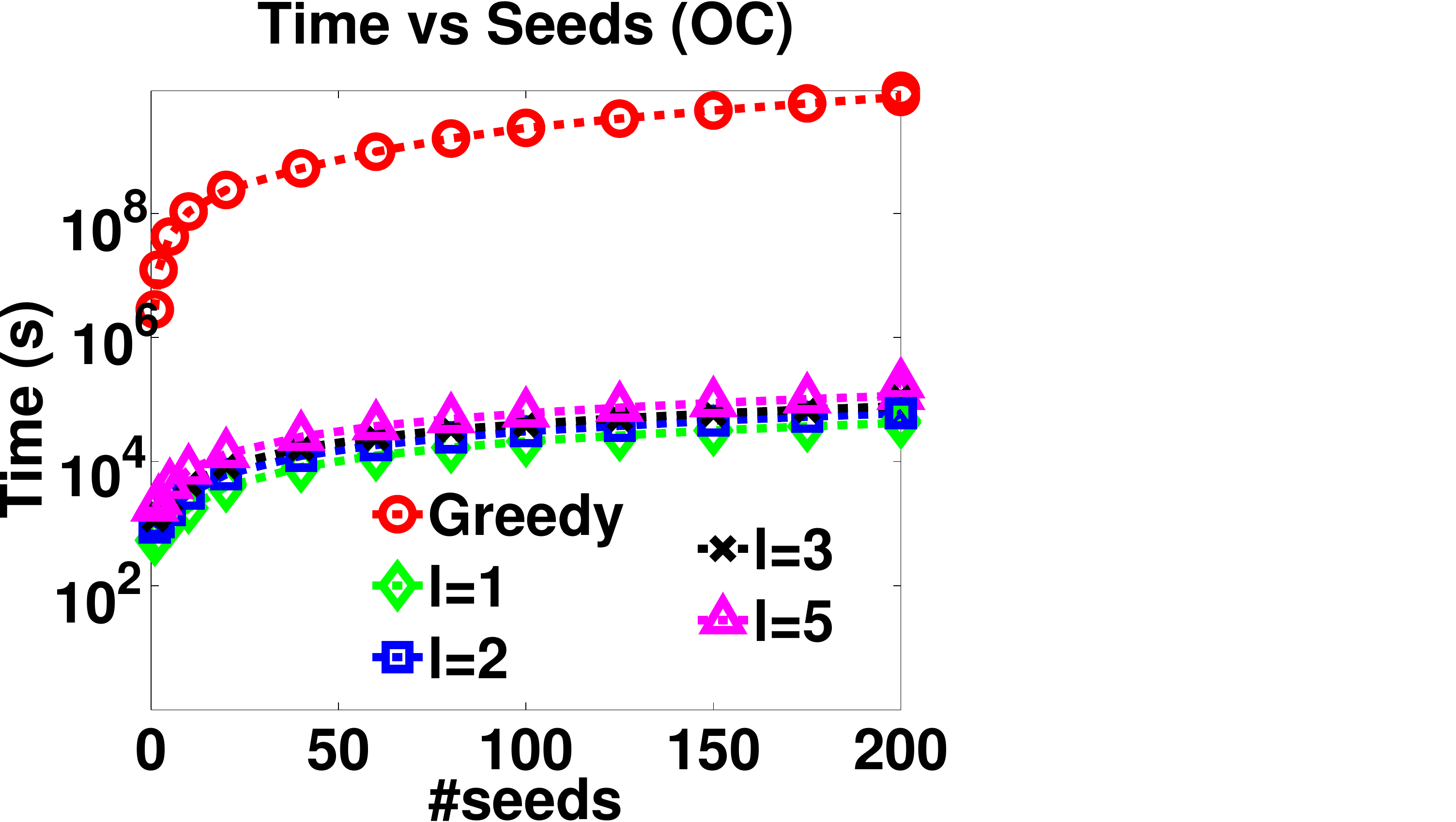}
		}
		\label{fig:time_hepph_oc}
	}
	\subfloat[DBLP and YouTube]
	{
		\scalebox{0.182}{
			\includegraphics[width = 0.99\linewidth, trim=0cm 0cm 10cm 0cm]{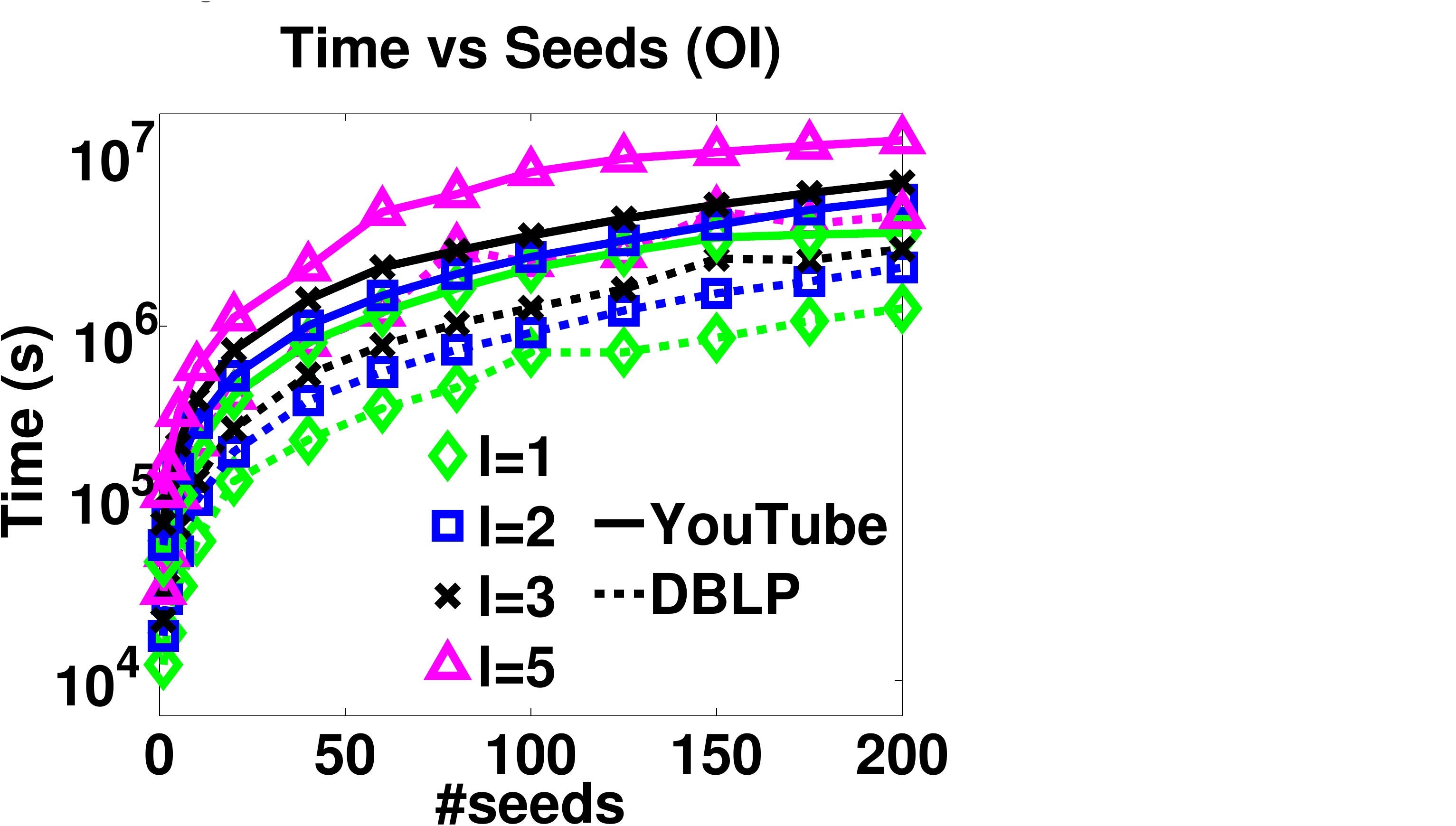}
		}
		\label{fig:time_dblp_youtube_oi}
	}
	\subfloat[Medium Datasets]
	{
		\scalebox{0.1855}{
			\includegraphics[width = 0.99\linewidth, trim=0cm 0cm 10cm 0cm]{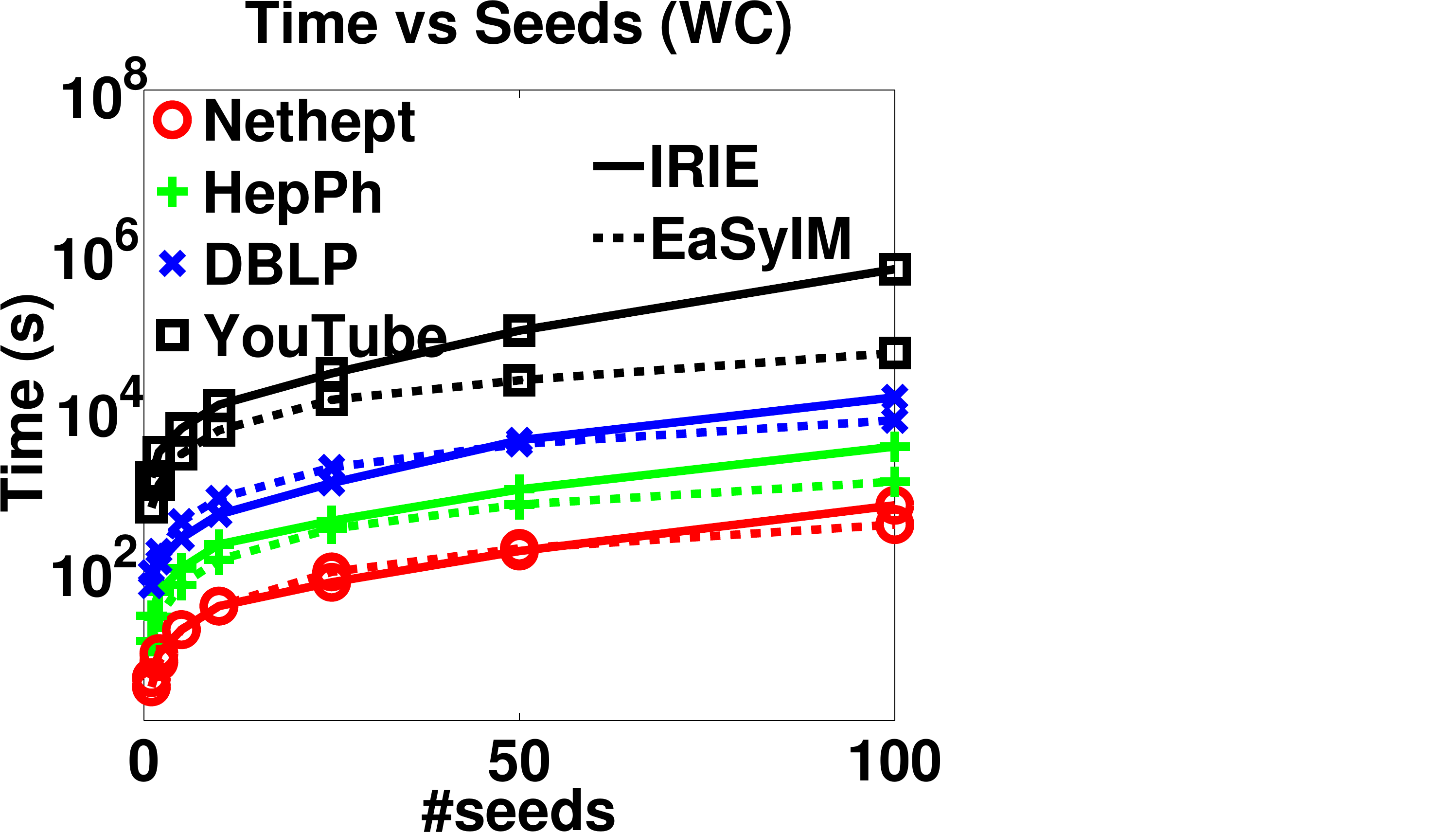}
		}
		\label{fig:time_all_irie}
	}
	\subfloat[Medium Datasets]
	{
		\scalebox{0.1855}{
			\includegraphics[width = 0.99\linewidth, trim=0cm 0cm 10cm 0cm]{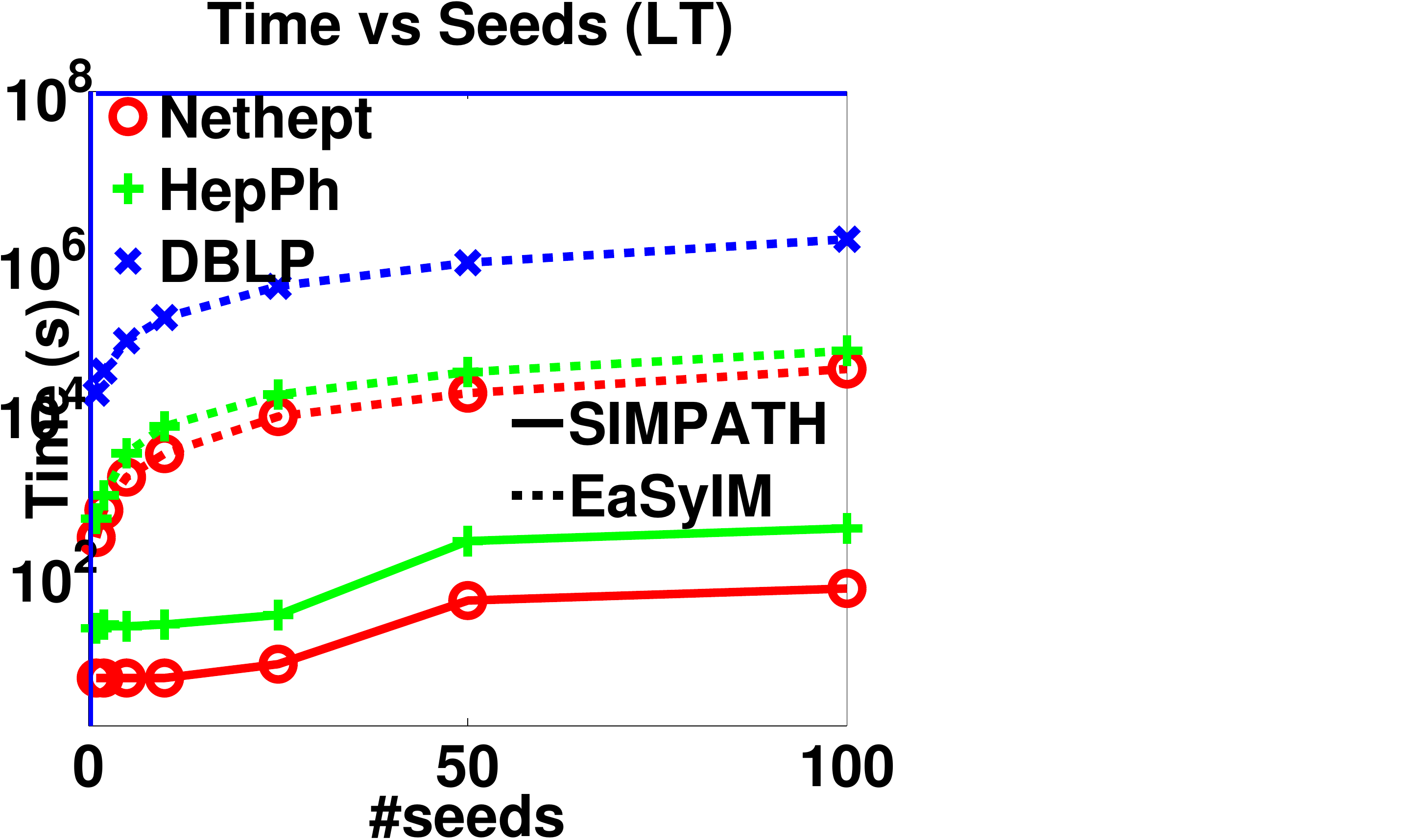}
		}
		\label{fig:time_all_simpath}
	}
	\subfloat[Large Datasets]
	{
		\scalebox{0.21}{
			\includegraphics[width = 0.99\linewidth]{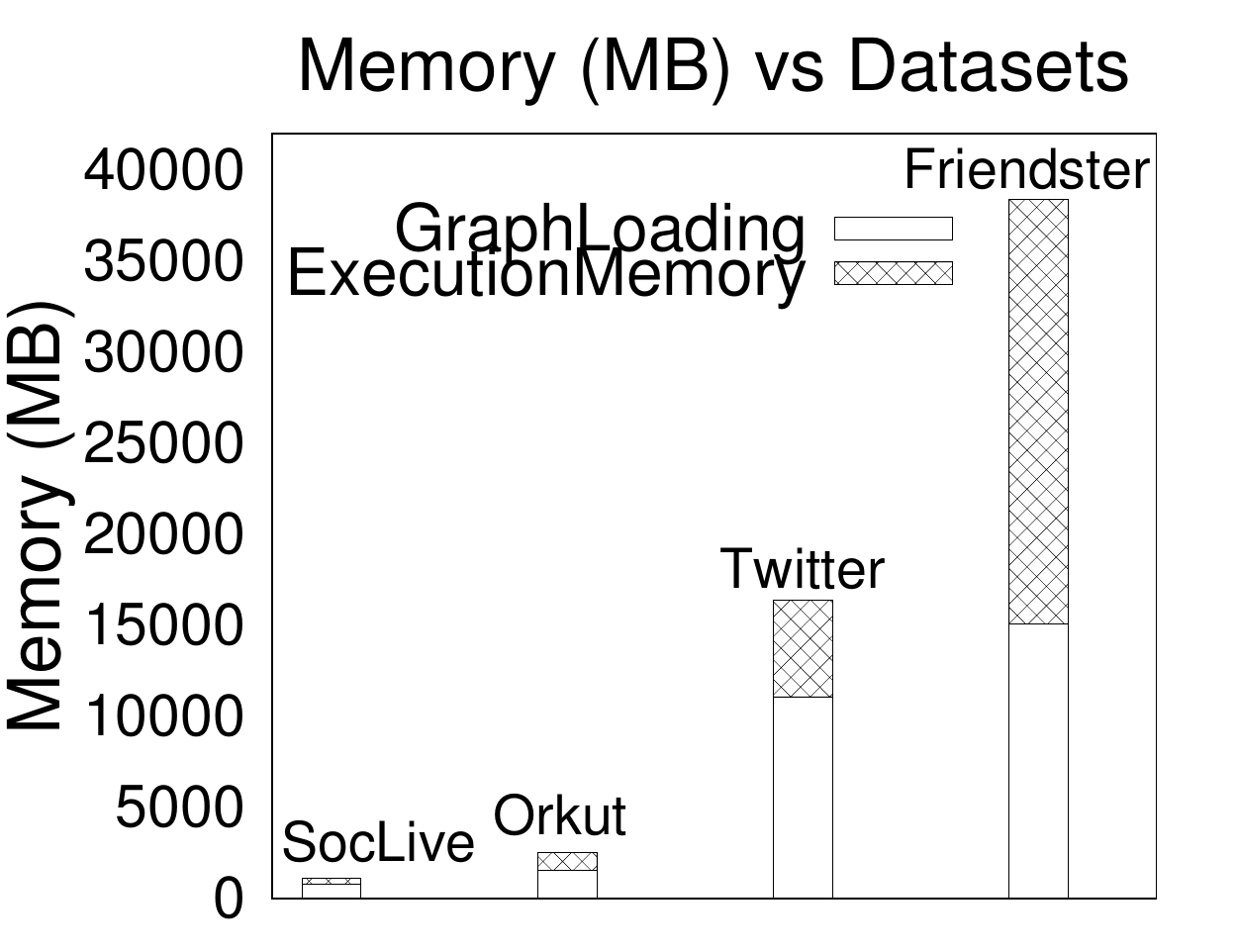}
		}
		\label{fig:memory_easyim_large}
	}
\figcaption{(a) Opinion-spread comparison of $\lambda=1$ with $\lambda=0$ on DBLP and YouTube. Growth of opinion-spread ($\lambda=1$) with $l$ and $k$ on (b) HepPh (OC), and (c) DBLP and YouTube (OI). Spread comparison of \emph{EaSyIM} with SIMPATH and IRIE on (d) NetHEPT (LT) and (e) YouTube (WC). Growth-rate of running time of \emph{OSIM} with $l$ and $k$ on (f) HepPh (OC) and (g) DBLP and YouTube (OI). Running time comparison of \emph{EaSyIM} with IRIE on (h) Medium Datasets (WC), and SIMPATH on (i) Medium Datasets (LT). (j) Memory consumption of \emph{EaSyIM} for $k=100$ on Large Datasets.}
\label{app_fig:all_results}
\end{figure*}

\ignore{
\section{Proof of Lemma 7}
\label{app_thm:lemma7}

{\small

\begin{proof}
The contribution of a node $w$ in the score of another node $u$, considering the effect of cycles exclusively, is dependent upon all possible cyclic paths that contain $w$. Moreover, since the combined contribution of the participating nodes in each cycle $\rho$ should only be considered once to the score of $u$, we scale $\gamma_w(u)$ by $1/|\rho|$. Mathematically,
\begin{align*}
\gamma_w(u) = \sum_{\rho \in \pathset_{ww}}p_{(u,w)}(1+1/|\rho|)\prod_{e \in \rho}p_e.
\end{align*}
However, the contribution using the {\kempegreedy} algorithm is
\begin{align*}
\gamma^*_w(u) = p_{(u,w)}.
\end{align*}
Now, the relative error $\epsilon^{cycle}$, is
\begin{align}
\label{eq:cycle}
\epsilon^{cycle} = \frac{|\gamma^*_w(u) - \gamma_w(u)|}{\gamma^*_w(u)} &= \frac{\sum\limits_{\rho \in \pathset_{ww}}p_{(u,w)}(1/|\rho|)\prod_{e \in \rho}p_e}{p_{(u,w)}}\nonumber\\
		&= \sum_{\rho \in \pathset_{ww}}\Big(\prod_{e \in \rho}p_e\Big)/|\rho|.
\end{align}
\hfill{}
\end{proof}
}

\section{Proof of Lemma 8}
\label{app_thm:lemma8}

{\small

\begin{proof}
It is evident that the contribution of a path of length $l$ in the expected opinion spread of a seed node $u_0$, is the sum of expected effective opinions of nodes in that path. Now we try to solve the following recursive relation to get a closed expression for the expected opinion of the nodes.
\begin{align*}
o_{u_i}' &= \frac{o_{u_i}}{2} + \psi_{i-1} o_{u_{i-1}}'\\
o_{u_{i-1}}' &= \frac{o_{u_{i-1}}}{2} + \psi_{i-2} o_{u_{i-2}}'\\
&.\\
&.\\
&.\\ 
o_{u_1}' & = \frac{o_{u_1}}{2} + \psi_{0} o_{u_{0}}'
\end{align*}
On solving the telescopic sum and substituting $o_{u_0}' = o_{u_0}$
\begin{align*}
o_{u_i}' &= \sum_{j=1}^{i}\Big( \frac{ o_{u_j}}{2} \prod_{k=1}^{i-j} \psi_{i-k}\Big) + o_{u_0} \prod_{k=1}^{i} \psi_{i-k}\\
&=  \sum_{j=0}^{i}\Big( \frac{ (1+\delta_j(\{0\}))  o_{u_j}}{2} \prod_{k=1}^{i-j} \psi_{i-k}\Big).
\end{align*}
\begin{align*}
\sigma^o(\{u_0\}) &= \sum_{i=1}^{l} \bigg(o_{u_i}'\prod_{j=1}^i \thatsymbol_{j-1}\bigg)\\
&= \sum\limits_{i=1}^{l} \bigg(\Big( \prod\limits_{j=1}^i \thatsymbol_{j-1} \Big) \sum\limits_{j=0}^{i} \Big(\frac{o_{u_j}}{2}\big(1+\delta_j(\{0\})\big)\prod\limits_{k=1}^{i-j} \psi_{i-k} \Big) \bigg).
\end{align*}

\ignore{
\begin{align*}
OIspread &= \sum_{i=1}^{d} \left(o_{u_i}'\prod_{j=1}^i p_{(u_{j-1},u_j)}\right)\\
& = \sum_{i=1}^{d} \left(\left( \prod_{j=1}^i p_{(u_{j-1},u_j)} \right) \left( \sum_{j=0}^{i} \frac{(1+\delta_j(\{0\}))  o_{u_j}}{2}\prod_{k=1}^{i-j} \frac{(2\varphi_{(u_{i-k,i-k+1})}-1)}{2}   \right) \right)
\end{align*}
}
\hfill{}
\end{proof}
}

\section{Proof of Lemma 9}
\label{app_thm:lemma9}

{\small

\begin{proof}
In order to prove this, we first obtain an expression of the score assigned to a node as a function of $u_i,\ i\in \{0,1,\ldots, l\}$. From the line $6$ of Algorithm~\ref{alg:OSIM}, we evaluate the expression for $or_i(u_k)$ as follows:
$or_{i}(u_k) = or_{i-1}(u_{k+1}) p_{(u_k,u_{k+1})} = or_{i-1}(u_{k+1}) \thatsymbol_{k} $.
This expression can be written recursively as follows:

\begin{align}
\label{eq:or_recursive}
or_{i}(u_k) &= or_{i-1}(u_{k+1}) \thatsymbol_{k}\nonumber\\
or_{i-1}(u_{k+1}) &= or_{i-2}(u_{k+2}) \thatsymbol_{k+1}\nonumber\\
&.\nonumber\\
&.\nonumber\\
&.\nonumber\\ 
or_{1}(u_{k+i-1}) &= or_{0}(u_{k+i} ) \thatsymbol_{k+i-1}
\end{align}
On solving the telescopic sum and putting $or_0(u_k) = o_{u_k}$, we get
\begin{align*}
or_i(u_k) & = o_{u_{k+i}}\prod_{j=0}^{i-1}\thatsymbol_{k+j}.
\end{align*}
We evaluate $\alpha$ using the line $7$ of Algorithm~\ref{alg:OSIM}, and obtain a similar set of recursive expressions as in Eq.~\ref{eq:or_recursive}.
\ignore{
\begin{align*}
\alpha_{i}(u_k) &= \alpha_{i-1}(u_{k+1}) \thatsymbol_{k}\psi_k\\
\alpha_{i-1}(u_{k+1}) &= \alpha_{i-2}(u_{k+2}) \thatsymbol_{k+1}\psi_{k+1}\\
&.\\
&.\\
&.\\ 
\alpha_{1}(u_{k+i-1}) &= \alpha_{0}(u_{k+i} ) \thatsymbol_{k+i-1}\psi_{k+i-1}
\end{align*}
}
On solving the telescopic sum and putting $\alpha_0(u_k) = 1$, we obtain the following:
\begin{align*}
\alpha_i(u_k) & = \prod_{j=0}^{i-1}\thatsymbol_{k+j}\psi_{k+j}.
\end{align*}
Now using line $10$ of Algorithm~\ref{alg:OSIM}, we obtain an expression of $sc_{i}(u_k)$ defined as follows:
\begin{align*}
sc_{i}(u_k) &= sc_{i-1}(u_{k+1}) \thatsymbol_{k} + o_{u_k} \alpha_i(u_k).
\end{align*}
Putting the evaluated value of $\alpha$ in the above expression, we again obtain a similar set of recursive expressions as in Eq.~\ref{eq:or_recursive}.
\ignore{
\begin{align*}
sc_{i}(u_k) =& sc_{i-1}(u_{k+1}) \thatsymbol_{k} + o_{u_k} \prod_{j=0}^{i-1}\thatsymbol_{k+j}\psi_{k+j} \\
sc_{i-1}(u_{k+1}) =& sc_{i-2}(u_{k+2}) \thatsymbol_{k+1} + o_{u_{k+1}} \prod_{j=0}^{i-2}\thatsymbol_{k+j+1}\psi_{k+j+1} \\
&.\\
&.\\
&.\\ 
sc_{1}(u_{k+i-1}) =& o_{u_{k+i-1}}\thatsymbol_{k+i-1} \psi_{k+i-1}
\end{align*}
\vspace{0.5mm}
}
On solving the telescopic sum, we get:
\begin{align*}
sc_i(u_k) & = \sum_{j=0}^{i-1}\bigg (o_{u_{k+j}}\big(\prod_{r=0}^{j-1}\thatsymbol_{k+r}\big) \bigg)\bigg( \prod_{s=0}^{i-j-1} \thatsymbol_{k+s+j}(\psi_{k+s+j})\bigg)\\
& = \prod_{j=k}^{i+k-1}  \left( \sum_{j=0}^{i-1}o_{u_{k+j}} \prod_{s=0}^{i-j-1} (\psi_{k+s+j})\right).
\end{align*}
Now we evaluate $sc_i(u_0) + or_i(u_0) + o_{u_0}\alpha_{i}(u_0)$ as:
\begin{align*}
& \prod_{j=0}^{i-1} \thatsymbol_{j} \bigg( \sum_{j=0}^{i-1}o_{u_{j}} \prod_{s=0}^{i-j-1} (\psi_{s+j})\bigg) + \prod_{j=0}^{i-1} \thatsymbol_{j} o_{u_i} + o_{u_0} \prod_{j=0}^{i-1}\thatsymbol_{j}\psi_j\\
& = \bigg( \prod_{j=1}^i\thatsymbol_{j-1}\bigg) \bigg(    \Big( \sum_{j=0}^{i}o_{u_{j}} \prod_{s=0}^{i-j-1} (\psi_{s+j})\Big) + o_{u_0} \prod_{j=0}^{i-1}\thatsymbol_{j}\psi_j  \bigg).
\end{align*}

From line $11$ of Algorithm~\ref{alg:OSIM}, it is evident that $\Delta^l(u_0) = \sum\limits_{i=0}^l\bigg(\frac{or_i(u_0)}{2}+ \frac{sc_i(u_0)}{2} + o_{u_0} \frac{\alpha_i(u_0)}{2}\bigg)$ which means:
\begin{align*}
\Delta^l(u_0) & = \sum_{i=1}^{l} \bigg(\Big( \prod_{j=1}^i \thatsymbol_{j-1} \Big) \Big( \sum_{j=1}^{i} \frac{o_{u_j}}{2}\prod_{k=1}^{i-j} \psi_{i-k} + o_{u_0} \prod_{k=1}^{i} \psi_{i-k} \Big) \bigg) \\
& = \sigma^o(\{u_0\}).
\end{align*}
\hfill{}
\end{proof}
}
}

\section{Modified-Greedy Algorithm}
\label{app_alg:mod_greedy}

{\small
\setlength{\textfloatsep}{0pt}
\begin{algorithm}[!htbp]
\caption{Modified-{\kempegreedy}}
\label{algo:modified_greedy}
\begin{algorithmic}[1]
{\scriptsize
\REQUIRE Graph $G = (V,E)$,\#seeds $k=|S|$
\ENSURE Seed set $S$
\STATE $S \leftarrow \emptyset$
\FOR{i = 1 to k}
\STATE $maxId = \argmax\limits_{w\in V\setminus S} \big(\flow^o_\lambda(S\cup\{w\}) - \flow^o_\lambda(S)\big)$ 
\STATE $S \leftarrow S \cup \{ maxId \}$
\ENDFOR
}
\end{algorithmic}
\afterpage{\global\setlength{\textfloatsep}{\oldtextfloatsep}}
\end{algorithm}
}

\section{Additional Experimental Results}
\label{app_res:opinion_aware}

{\small
\textbf{Quality}: It is evident from Figure~\ref{fig:dblp_youtube_lambda} that $\lambda=1$ outperforms $\lambda=0$ and hence, maximizing the effective opinion-spread is better. Figures~\ref{fig:eop_hepph_oc} and~\ref{fig:eop_dblp_youtube_oi} present the results with varying seeds for HepPh under the OC model ($o\sim\mathcal{N}(0,1)$), and DBLP \& YouTube under the OI model ($o\sim rand(-1,1)$) respectively. We omit the results for the Modified-{\kempegreedy} algorithm in Figure~\ref{fig:eop_dblp_youtube_oi} owing to its lack of scalability. Figures~\ref{fig:spread_nethept_lt_compare} and~\ref{fig:spread_youtube_wc_compare} present the spreads obtained for the NetHEPT and YouTube datasets under the LT and WC models respectively.

\textbf{Efficiency}: Figures~\ref{fig:time_hepph_oc} show that \emph{OSIM} is at least $10^3$ times more efficient when compared to Modified-{\kempegreedy} with the gain going as high as $10^5$ for some cases. Figure~\ref{fig:time_dblp_youtube_oi} portrays the running time of \emph{OSIM} for the DBLP and the YouTube datasets. Note that the solid line overlapping with the y-axis and the horizontal top part of the plot shows the time for Modified-{\kempegreedy}, which is indicative of the fact that it did not complete even after a month's time. Figures~\ref{fig:time_all_irie} and~\ref{fig:time_all_simpath} show a comparison of the running times of \emph{EaSyIM} with the state-of-the-art heuristics -- IRIE (WC) and SIMPATH (LT). It is evident that both of these techniques do not scale well. \emph{EaSyIM} is $2-3$ and $4-6$ times faster when compared to IRIE for small and medium sized datasets respectively. It is slower when compared to SIMPATH on smaller datasets, however, highly efficient for the medium sized datasets as SIMPATH did not complete even after $5$ days on the DBLP dataset.


\textbf{Scalability}: Eventually as a part of our scalability experiments, we present the memory required by \emph{EaSyIM} for the $4$ large graph datasets (Table~\ref{tab:dataset}) in Figure~\ref{fig:memory_easyim_large}. This shows that \emph{EaSyIM} possesses the capability of scaling to graphs with billion-scale edges.
}
\section{Additional Note on CELF++}
\label{app_alg:celf_tim}

{\small
This algorithm exploits the sub modularity of the problem by using a ``lazy-forward'' optimization while selecting seeds. This technique calculates marginal gain of nodes in decreasing order of their probability of being a seed node, which in turn is calculated form its marginal gain in previous iterations. Another optimization which helped engineer the running time of this algorithm was to ignore the nodes which do not hold a chance to become a seed node in consequent iterations.
}

\end{document}